\documentclass[11pt]{article}

\usepackage[margin=1in]{geometry}
\usepackage{amsmath,amssymb,amsthm}
\usepackage{mathtools}
\usepackage{algorithm}
\usepackage[noend]{algpseudocode}
\usepackage{array}
\usepackage{enumitem}
\usepackage{microtype}

\usepackage{thmtools}
\usepackage{thm-restate}
\usepackage{url}
\usepackage[colorlinks=true,urlcolor=blue!70!black,linkcolor=blue!70!black,citecolor=purple!80!black]{hyperref}

\usepackage[capitalise,nameinlink]{cleveref}
\usepackage{booktabs}
\usepackage{pifont}
\usepackage{soul}

\newtheorem{theorem}{Theorem}
\newtheorem{lemma}[theorem]{Lemma}
\newtheorem{proposition}[theorem]{Proposition}

\theoremstyle{definition}
\newtheorem{definition}[theorem]{Definition}

\declaretheorem[name=Theorem,numbered=no]{theorem*}

\algnewcommand{\Input}{\item[\textbf{Input:}]}

\usepackage{xspace}
\newcommand{\OPT}{\operatorname{OPT}}
\newcommand{\OPTR}{\operatorname{OPT}^{\text{rob}}}

\newcommand{\tsAlg}{\textsc{ThresholdSplit}\xspace}
\newcommand{\btiAlg}{\textsc{BudgetToInterval}\xspace}
\newcommand{\bpAlg}{\textsc{BlockPacking}\xspace}

\usepackage{titlesec}
\titlespacing*{\paragraph}{0pt}{1ex plus 0.1ex}{0.8em}
\usepackage[
 backend=biber,
 style=alphabetic,
 citetracker,
 hyperref=auto,
  maxcitenames=5,  sortcites,        sorting=nyt,
 maxbibnames=12,   date=year,       isbn=false,      url=false,       doi=false,       eprint=false,   ]{biblatex}
\newbibmacro{string+doiurlisbn}[1]{  \iffieldundef{doi}{    \iffieldundef{url}{      #1
    }{      \href{\thefield{url}}{#1}    }  }{    \href{http://dx.doi.org/\thefield{doi}}{#1}  }}

\DeclareFieldFormat[article,inbook,incollection,inproceedings,patent,thesis,unpublished,misc]
  {title}{\usebibmacro{string+doiurlisbn}{#1}}

\usepackage{xcolor}
\usepackage{pifont}
\usepackage[colorinlistoftodos,prependcaption,textsize=scriptsize]{todonotes}

\newcommand{\cmark}{{\color{green!70!black} \Large\ding{51}}}\newcommand{\xmark}{{\color{red}\Large\ding{55}}}
\title{
A Consistency--Robustness Framework for Robust Optimization: \\
Integrating Predictions into Robust Scheduling}

\author{
	Yasser Alghouass\footnote{Columbia University, New York, USA. \texttt{\{ya2581, eb3224, vg2277\}@columbia.edu}}
    \and Eric Balkanski\footnotemark[1]
	\and Vineet Goyal\footnotemark[1]
    \and Nicole Megow\footnote{University of Bremen, Bremen, Germany. \texttt{nicole.megow@uni-bremen.de}}
}

\date{}

\begin{document}
\maketitle

\begin{abstract}
Robust optimization protects against uncertainty by optimizing for the worst case over a prescribed uncertainty set. This protection can be overly conservative when forecasts, historical data, or learned predictions indicate a more likely scenario. We introduce a framework for robust optimization with predictions. The input consists of an uncertainty set together with a distinguished predicted scenario, and the goal is to compute a single solution that is both consistent, meaning near-optimal for the predicted scenario, and robust, meaning competitive with the classical min--max robust optimum. Unlike in standard learning-augmented algorithms, the prediction does not merely estimate the realized input; it creates a separate benchmark, the predicted optimum, which must be balanced against the min--max robust optimum.

We study this framework for makespan scheduling with uncertain processing times and give a structural classification across standard uncertainty models and machine environments. For interval uncertainty, we obtain a smooth $(1+1/\lambda,1+\lambda)$ consistency--robustness tradeoff for restricted-assignment and related machines. Furthermore, we prove that unrelated machines admit no constant tradeoff. For budgeted uncertainty, we obtain a $(1+1/\lambda,2+\lambda)$ tradeoff for restricted assignment. Our analysis is based on a duality-based reduction to an interval-like upper envelope. We complement this with a lower bound showing that related machines admit no constant tradeoff even when only one job may deviate. For arbitrary uncertainty sets, we obtain constant tradeoffs for identical machines via a support-function block construction, and prove impossibility for restricted assignment.

Our results show that the possibility of combining consistency and robustness in robust scheduling depends critically on the interaction between the uncertainty model and the machine environment.
\end{abstract}

\thispagestyle{empty} \clearpage
\setcounter{page}{1}

\section{Introduction}
\label{sec:intro}

Optimization under uncertainty is a central challenge in algorithm design. Robust optimization is one of the foundational frameworks for addressing this challenge: it replaces a single input instance with an uncertainty set of possible scenarios and seeks a solution that optimizes the worst-case performance over all of them. This paradigm has been highly influential because it provides deterministic guarantees and often leads to tractable models even when the uncertainty set is specified implicitly; see Kouvelis and Yu~\cite{KouvelisYu1997} and Ben-Tal, El Ghaoui,~and~Nemirovski~\cite{Ben-talGN2009-book-robust}.

At the same time, the classical robust-optimization objective can be overly conservative, as it optimizes for a worst-case scenario that may be highly unlikely. In settings such as scheduling, routing, or resource allocation, uncertainty sets are often chosen to cover rare disruptions or extreme demand patterns, such as delayed jobs, traffic spikes, resource failures, or unusually high demand, while historical data, forecasts, or machine-learned predictions may indicate a much more plausible outcome. This tension motivates the question we study in this paper: can one incorporate predictions into robust optimization so as to obtain solutions that perform well when the prediction is accurate, while retaining provable robustness against the entire uncertainty set?

A natural way to formalize the use of such side information is through the emerging area of learning-augmented algorithms, also known as algorithms with predictions~\cite{MitzenmacherV22,LykourisV21}. This line of work studies algorithms that are given auxiliary predictions about the input and aims to combine two guarantees: \emph{consistency}, meaning near-optimal performance when the prediction is accurate, and \emph{robustness}, meaning worst-case performance even when the prediction is arbitrarily wrong. While this idea already appeared conceptually in earlier work of Mahdian, Nazerzadeh, and Saberi~\cite{DBLP:conf/sigecom/MahdianNS07,DBLP:journals/talg/MahdianNS12}, the work of Lykouris and Vassilvitskii~\cite{LykourisV21} initiated a rapidly growing body of research. Since then, the learning-augmented paradigm has been applied to a wide range of algorithmic problems, most prominently in online optimization, but also in offline and approximation algorithms, dynamic algorithms and data structures, mechanism design, and streaming; see the curated repository~\cite{alps}.

In most of this literature, consistency and robustness are two guarantees for the same underlying input problem: the prediction may help to approach the optimum of the realized instance, while robustness protects against prediction error. In robust optimization, the role of a prediction is different. A prediction identifies one plausible scenario and hence one natural benchmark, the optimum for that scenario. The uncertainty set, however, defines a second benchmark, the min–max robust optimum.
Thus, incorporating predictions into robust optimization creates a bicriteria tradeoff between performance for the predicted scenario and worst-case performance over~the~uncertainty~set.

Motivated by this distinction, we initiate the study of robust optimization with predictions as a framework for incorporating learned or forecasted side information into robust optimization. The input consists of an uncertainty set $\mathcal U$ of possible scenarios together with a distinguished predicted scenario $\hat q \in \mathcal U$. The goal is to compute a single solution that performs well for the predicted scenario while remaining robust over the entire uncertainty set. We measure consistency against the optimum for the predicted scenario, and robustness against the classical min--max robust optimum over $\mathcal U$.
This perspective is distinct from data-driven and learning-based robust optimization~\cite{TulabandhulaR14,DBLP:journals/mp/BertsimasGK18,DBLP:journals/mansci/BertsimasK20}, where data or learned models are typically used to construct, estimate, or calibrate the uncertainty set. Here, the uncertainty set is part of the input, and learning enters only through a distinguished predicted scenario.

We develop this framework for makespan scheduling with uncertain processing times, a central problem class in optimization and algorithm design. Classical robust scheduling asks for a schedule with good worst-case performance over an uncertainty set, typically under min--max or min--max regret objectives; see, e.g.,~\cite{DanielsKouvelis1995,KouvelisYu1997,DBLP:journals/dam/BougeretPP19,DBLP:journals/mst/BougeretJPR21} and the references therein. These works do not distinguish a predicted scenario from the remaining uncertainty. In our setting, the challenge is simultaneous attainability: understanding which consistency--robustness tradeoffs are achievable, and how they depend on the uncertainty model and  machine environment. Our guiding question is:
\[
\textit{What consistency--robustness tradeoffs are achievable for robust scheduling with predictions?}
\]

Our results give a structural classification of this question for makespan scheduling. For each uncertainty model and machine environment considered, we either give explicit consistency--robustness tradeoffs or show that no instance-independent constant tradeoff is possible. We first study this question independently of computational restrictions and  return to  efficiency  later.

\subsection{Our results}
\label{sec:results-summary}

In robust makespan scheduling,  jobs have uncertain processing times $q \in \mathcal U$, and the goal is to find a schedule $\mathcal S$ that assigns the $n$ jobs to $m$ machines such that the worst-case makespan is minimized.
 We study four canonical  machine environments, which differ in the structural assumptions they impose on the processing times $q$: identical, restricted-assignment, related, and unrelated machines. The uncertainty models we consider are three classical models from robust optimization that capture increasingly expressive descriptions of uncertainty: interval, budgeted, and general uncertainty.

Interval uncertainty, also known as box uncertainty,  dates back to the work of Soyster~\cite{soyster1973convex} and has been widely used in robust scheduling; see, e.g., Kouvelis and Yu~\cite{KouvelisYu1997}. In this model, studied in  \Cref{sec:interval-uncertainty}, there are lower and upper bounds $q^-$ and $q^+$ on the processing times and $\mathcal U = \{q : q^- \leq q \leq q^+\}$. Our main result for interval uncertainty is a classification of which  machine environments admit an $O(1)$-consistent and $O(1)$-robust algorithm.

\begin{theorem*}[\Cref{thm:threshold-splitting-rr,thm:interval-unrelated-no-constant-tradeoff}]
    Identical, restricted-assignment, and related machines admit, for every $\lambda>0$,  a deterministic algorithm with consistency $1+1/\lambda$ and robustness $1+\lambda$ for interval uncertainty. Unrelated machines admit no algorithm that is both $O(1)$-consistent and $O(1)$-robust.
\end{theorem*}

We also show that for every integer  $1 \leq \lambda \leq m-1$, there is no randomized $(1+1/\lambda)$-consistent algorithm that achieves  robustness strictly better than $\lambda/4$, even for identical machines. Thus, the linear dependence on $\lambda$ in the robustness of $(1+1/\lambda)$-consistent algorithms is asymptotically tight.

 Budgeted uncertainty, often referred to as $\Gamma$-uncertainty, was introduced by Bertsimas and Sim~\cite{DBLP:journals/ior/BertsimasS04} and is well-studied in robust scheduling; see, e.g., \cite{lu2014robust,tadayon2015algorithms, DBLP:journals/dam/BougeretPP19, DBLP:journals/mst/BougeretJPR21}. This model, studied in \Cref{sec:budgeted-uncertainty}, generalizes interval uncertainty by allowing the adversary to modify the processing times of at most $\Gamma$ jobs, for some  budget $\Gamma$, instead of all of the jobs. Our main result for budgeted uncertainty is an analogous classification of the machine environments, but with related machines that now no longer admit an $O(1)$-consistent and $O(1)$-robust algorithm.

\begin{theorem*}[\Cref{thm:budgeted-cardinality-restricted,thm:budgeted-related-no-constant-tradeoff}]
    For budgeted uncertainty, for any $\lambda > 0$, there is a deterministic algorithm with consistency $1+1/\lambda$ and robustness $2 + \lambda$ for identical and restricted-assignment machines. For related and unrelated machines, no algorithm is both $O(1)$-consistent and $O(1)$-robust.
\end{theorem*}

We consider two extensions of budgeted uncertainty for which our results extend up to constant factor losses: oblivious budgets where the budget $\Gamma$ is unknown a priori in \Cref{sec:oblivious-cardinality}  and  weighted budgets where each job contributes a different weight towards the budget in \Cref{app:budgeted-additional-models}.

In the arbitrary uncertainty model, studied in~\Cref{sec:general-uncertainty}, $\mathcal U$ is an arbitrary set of scenarios. The classification we obtain is similar to that for budgeted uncertainty for identical machines. However, restricted-assignment machines no longer admit an $O(1)$-consistent and $O(1)$-robust algorithm in the arbitrary uncertainty model. We note that the linear dependence on $\lambda$ in the robustness is  tight for both budgeted and arbitrary uncertainty since the previously mentioned lower bound  for the special case of interval uncertainty applies.

\begin{theorem*}[\Cref{thm:general-uncertainty-tradeoff,thm:general-restricted-no-constant-tradeoff}]
    For arbitrary uncertainty, for any $\lambda > 0$, there is a deterministic algorithm with consistency $1+2/\lambda$ and robustness $4 + 2\lambda$ for identical machines. For restricted-assignment, related, and unrelated machines, no algorithm is both $O(1)$-consistent and $O(1)$-robust.
\end{theorem*}

\vspace{-.05cm}

Our results are summarized in \Cref{tab:results-summary} and yield a diagonal pattern of positive and negative results. This classification highlights an unusual feature in scheduling: restricted machines are easier than related machines for budgeted uncertainty. This feature interestingly reverses the more usual pattern in scheduling, where related machines are often viewed as the more structured special case.

\newcolumntype{C}[1]{>{\centering\arraybackslash}m{#1}}
\newcolumntype{L}[1]{>{\raggedright\arraybackslash}m{#1}}

\begin{table}[t]
\centering
\scriptsize
\setlength{\tabcolsep}{3pt}
\renewcommand{\arraystretch}{0.8}

\begin{tabular}{|C{0.22\textwidth}|L{0.17\textwidth}|L{0.17\textwidth}|L{0.17\textwidth}|L{0.17\textwidth}|}
\hline
\textbf{Machine environment} \(\rightarrow\) \newline
\textbf{Uncertainty model} \(\downarrow\)
& \multicolumn{1}{C{0.17\textwidth}|}{Identical}
& \multicolumn{1}{C{0.17\textwidth}|}{Restricted}
& \multicolumn{1}{C{0.17\textwidth}|}{Related}
& \multicolumn{1}{C{0.17\textwidth}|}{Unrelated} \\
\hline
Interval
& \cmark (\Cref{thm:threshold-splitting-rr})
& \cmark (\Cref{thm:threshold-splitting-rr})
& \cmark (\Cref{thm:threshold-splitting-rr})
& \xmark\ (\Cref{thm:interval-unrelated-no-constant-tradeoff}) \\
\hline
Budgeted
& \cmark (\Cref{thm:budgeted-cardinality-restricted})
& \cmark (\Cref{thm:budgeted-cardinality-restricted})
& \xmark\ (\Cref{thm:budgeted-related-no-constant-tradeoff})
& \xmark\ (\Cref{thm:interval-unrelated-no-constant-tradeoff}) \\
\hline
Arbitrary
& \cmark (\Cref{thm:general-uncertainty-tradeoff})
& \xmark\ (\Cref{thm:general-restricted-no-constant-tradeoff})
& \xmark\ (\Cref{thm:budgeted-related-no-constant-tradeoff})
& \xmark\ (\Cref{thm:interval-unrelated-no-constant-tradeoff}) \\
\hline
\end{tabular}
\vspace{-0.8em}
\caption{\small Classification of robust makespan scheduling with predictions according to whether a machine environment and uncertainty model
admit $O(1)$-consistent and $O(1)$-robust algorithms. }
\label{tab:results-summary}
\vspace{-0.8em}
\end{table}

\textit{Prediction error.} The prediction error $\eta$ is defined as the maximum multiplicative discrepancy between a predicted processing time and the corresponding true processing time. We show that any $\alpha$-consistent algorithm also achieves an $\eta^2\alpha$-approximation to the optimal makespan for the true scenario whenever the prediction error is $\eta$. Therefore, the performance guarantee of any algorithm degrades smoothly from its consistency guarantee as the prediction error increases. We therefore focus on consistency–robustness tradeoffs and obtain this smoothness guarantee for free.

\textit{Efficient algorithms.} We state our algorithms with exact optimization subroutines in order to separate the consistency-robustness tradeoff from the computational complexity of the problem. The same constructions have polynomial time versions whenever the corresponding optimization problems admit polynomial time approximation algorithms. Moreover, if we have  black-box access to  polynomial time  algorithms for  makespan minimization and robust makespan minimization that achieve a $\gamma$ and $\gamma_{\text{rob}}$ approximation, respectively, then, we can obtain polynomial time implementations of each of our algorithms up to some multiplicative losses dependent on  $\gamma$ and $\gamma_{\text{rob}}$. We present these results as corollaries of our main consistency-robustness tradeoff results.

\vspace{-.05cm}
\subsection{Technical overview}
\vspace{-.03cm}

We highlight the main technical ideas underlying our positive results and the structural obstructions behind our lower bounds. Across the different uncertainty models, the common challenge is that the prediction creates a second benchmark: a schedule should be near-optimal for the predicted scenario while also remaining competitive with the classical min--max robust optimum over all scenarios.

For interval uncertainty, the worst scenario is the upper-endpoint vector $q^+$ independent of the schedule. Thus, interval uncertainty collapses to a two-scenario comparison: consistency is measured with respect to  predicted scenario $\hat q$, while robustness is measured with respect to $q^+$. The central idea is a threshold-splitting principle. We compare each job by its contribution to the predicted optimum and to the robust optimum, after normalizing by the respective optimum values. Jobs whose  contribution to the robust optimum is large relative to their predicted contribution are scheduled according to $q^+$, while the rest follow $\hat q$. The two sub-instance schedules are then merged machine by machine. Sub-instance monotonicity bounds the parts scheduled for their own objectives, and the threshold inequality converts the remaining parts between the two objectives at a loss of $\lambda$ or $1/\lambda$. This yields the tradeoff $(1+1/\lambda,1+\lambda)$. The argument is insensitive to the speed spread for related machines and preserves eligibility constraints in restricted assignment.

Budgeted uncertainty is more subtle because the worst case is no longer described by a single fixed vector: the adversary may choose which jobs deviate, and this choice depends on the chosen schedule. A main technical contribution is a duality-based reduction from budgeted uncertainty to an interval-like upper envelope. For any fixed set of jobs, the worst cardinality-budgeted load is the lower load plus the largest deviations that fit within the budget. Dualizing this load maximization introduces a cutoff $\theta$: each deviation is either charged to a global budget term or contributes only its residual above the cutoff. Writing $p(q^-)_j$ and $p(q^+)_j$ as the lower and upper sizes of job $j$ in scenario $q$, this yields effective sizes $p(q^-)_j+(p(q^+)_j-p(q^-)_j -\theta)_+$. Choosing $\theta$ from the min--max robust optimum makes the reduction compatible with the robust benchmark: the effective interval instance has optimum value at most the original min--max optimum, and the budget term contributes only one additional copy of this optimum when translating the interval guarantee back to the budgeted uncertainty set. This gives robustness $\lambda+2$ while preserving consistency $1+1/\lambda$. The same dualization technique also yields the weighted fractional variant, while the weighted binary variant follows by comparison with the fractional relaxation.

For arbitrary uncertainty sets, we consider an arbitrary set of scenarios (specified implicitly or explicitly). For any set of jobs $J$, the key abstraction is the one-machine support function
$\Psi_{\mathcal U}(J)=\sup_{q\in\mathcal U}\sum_{j\in J}p(q)_j$. This viewpoint encodes all scenarios in a single set function: the algorithm uses only monotonicity and subadditivity of $\Psi_{\mathcal U}$. In addition, we observe that for identical machines $\sup_{q\in\mathcal U}C_{\text{max}}(\mathcal S,q)=\max_i \Psi_{\mathcal U}(S_i)$. For identical machines, these weak properties are already sufficient to obtain the consistency--robustness tradeoff. The algorithm starts from a prediction-optimal schedule and a min--max robust schedule. Jobs that are large in the prediction stay on their prediction-optimal machines. The remaining jobs are cut into small predicted-load blocks inside the machines of the robust schedule; hence, each block has a support value at most the robust optimum. These blocks are then packed into the residual predicted-load capacities of the prediction-optimal schedule by a greedy fill rule, which exceeds each target capacity by at most one block. In this way, the construction keeps the predicted load almost unchanged while controlling the number of blocks on each machine, yielding the tradeoff~$(1+2/\lambda,2\lambda+4)$.

The lower bounds show that the positive results rest on structural properties of the models rather than on artifacts of the analysis. Threshold splitting cannot extend to unrelated machines, where prediction and upper-endpoint processing times may favor different machines for the same job. The budget-to-interval reduction relies on restricted-assignment structure; for related machines, a single deviating job on a slow machine rules out speed-spread-independent tradeoffs. The support-function block construction uses the symmetry of identical machines; with restricted assignment, arbitrary uncertainty can encode correlations that make the robust load of every consistent schedule arbitrarily large. These obstructions explain the diagonal pattern of the results: stronger uncertainty requires more symmetric machines, while more heterogeneous machines require more structured uncertainty.

\subsection{Further related work}

{\em Classical robust scheduling} studies schedules that hedge against uncertainty in processing times, typically under min--max or min--max regret objectives. For interval uncertainty, min--max makespan minimization is often structurally simple in monotone machine environments, since the worst case is obtained by setting all processing times to their upper endpoints; much of the nontrivial interval-uncertainty scheduling literature therefore concerns min--max regret variants, see the survey of Kasperski and Zieli{\'n}ski~\cite{KasperskiZielinski2014} and, e.g.,~\cite{DBLP:journals/orl/KasperskiZ08,DBLP:journals/orl/DrwalR16,DBLP:journals/anor/SiepakJ14}.

The closest min--max robust makespan results to ours concern budgeted uncertainty in the sense of Bertsimas and Sim~\cite{DBLP:journals/mp/BertsimasS03,DBLP:journals/ior/BertsimasS04}, where a budget parameter $\Gamma$ limits the number of jobs whose processing times may deviate from their nominal values. Bougeret, Pessoa, and Poss~\cite{DBLP:journals/dam/BougeretPP19} and Bougeret, Jansen, Poss, and Rohwedder~\cite{DBLP:journals/mst/BougeretJPR21} give approximation algorithms for classical min--max makespan scheduling under cardinality-budgeted processing-time uncertainty, including an EPTAS for identical machines, a $(2+\varepsilon)$-approximation for related machines, and a $3$-approximation for unrelated machines. In contrast to this literature, which optimizes the robust objective alone, we add a predicted scenario and study the bicriteria tradeoff between performance on the prediction and performance against the min--max robust benchmark.

{\em Vector scheduling and scheduling over scenarios} offer related algorithmic viewpoints. In vector scheduling, each job has a multidimensional load vector, and the goal is to minimize the maximum load over all machines and dimensions; Chekuri and Khanna~\cite{DBLP:journals/siamcomp/ChekuriK04} gave a PTAS when the dimension is fixed. Thus, for identical machines and a fixed finite set of scenarios, scenarios can be encoded as dimensions of a vector-scheduling instance. In particular, for interval uncertainty, the predicted scenario and the upper-endpoint scenario form a two-dimensional vector-scheduling instance, giving a polynomial-time way to approximate the best schedule for any fixed normalization of the predicted and robust objectives; see also Appendix~\ref{app:vector-scheduling-polytime-identical}. Scheduling over scenarios studies a related model in which one seeks a single assignment that performs well across a given collection of scenarios, often represented as subsets of jobs~\cite{DBLP:conf/cocoon/FeuersteinMSSSSZ14}. These viewpoints are close to robust scheduling with explicitly listed uncertainty sets, but they do not identify the consistency--robustness tradeoff created by a separate prediction benchmark.

{\em Data-driven and learning-based robust optimization} uses data, samples, covariates, or learned models mainly to construct, estimate, or calibrate uncertainty sets~\cite{TulabandhulaR14,DBLP:journals/mp/BertsimasGK18,DBLP:journals/mansci/BertsimasK20}. Related work connects predictive and prescriptive analytics~\cite{DBLP:journals/mansci/BertsimasK20}, studies distributionally robust variants~\cite{DBLP:journals/ior/DelageY10,DBLP:journals/ior/WiesemannKS14}, and develops machine-learning-based uncertainty sets~\cite{BertsimasB25}. Our setting is complementary: the uncertainty set is part of the input, and predictions enter as a distinguished scenario that creates a separate consistency benchmark.

\section{Preliminaries}
\label{sec:preliminaries}

We define the problem for the most general setting with unrelated machines and arbitrary uncertainty, and then introduce the special cases. Throughout, $[k]:=\{1,\ldots,k\}$. We slightly abuse notation and write $\textsc{ALG}$ for an algorithm and $\textsc{ALG}(I)$ for the schedule returned by this algorithm on instance $I$.

\par\medskip\noindent In \textit{makespan minimization},  an instance with $m$ machines and $n$ jobs is given by $q \in \mathbb{R}_{\geq 0}^{m \times n}$, where $q_{i,j}$ is the processing time of job $j \in [n]$ on machine $i \in [m]$. A schedule is an $m$-partition $\mathcal{S} =(S_1,\ldots,S_m)$ of the jobs on the machines. The goal is to find a schedule $\mathcal{S}$ that minimizes the makespan $C_{\text{max}}(\mathcal{S}, q) = \max_{i \in [m]} \sum_{j \in S_i} q_{i,j}$.
An algorithm $\textsc{ALG}$
achieves  an $\alpha$-approximation if, for every instance $q$, it returns a schedule $\textsc{ALG}(q)$ with  $C_{\text{max}}(\textsc{ALG}(q), q) \leq \alpha \cdot \textsc{OPT}(q)$ where $\textsc{OPT}(q) = \min_{\mathcal{S}} \max_{i \in [m]} \sum_{j \in S_i} q_{i,j}$ denotes the optimal makespan for instance $q$.

\par\medskip\noindent In \textit{robust makespan minimization}, we are given an uncertainty set $\mathcal{U}$ of scenarios $q$. The goal is to find a schedule $\mathcal{S}$ that minimizes the worst-case makespan $C_{\text{max}}^{\text{rob}}(\mathcal{S}, \mathcal{U})= \max_{q \in \mathcal{U}} C_{\max}(\mathcal{S}, q)$.
An algorithm $\textsc{ALG}$
is $\beta$-robust if, for every $\mathcal{U}$, it returns a schedule $\textsc{ALG}(\mathcal{U})$ with $C_{\text{max}}^{\text{rob}}(\textsc{ALG}(\mathcal{U}), \mathcal{U}) \leq \beta \cdot \OPTR(\mathcal{U})$ where $\OPTR(\mathcal{U}) = \min_\mathcal{S} C_{\text{max}}^{\text{rob}}(\mathcal{S}, \mathcal{U})$ is the optimal robust makespan.

\par\medskip\noindent In \textit{robust makespan minimization with predictions}, in addition to the uncertainty set $\mathcal{U}$, we are also  provided with a predicted scenario $\hat{q} \in \mathcal{U}$. The two main performance measures used to evaluate an algorithm with predictions are consistency and robustness

\begin{definition}
    An algorithm $\textsc{ALG}(\mathcal U, \hat q)$ is $\alpha$-consistent and $\beta$-robust for uncertainty set $\mathcal{U}$ if, for any predicted scenario     $\hat{q} \in \mathcal{U}$, we have
    $$C_{\max}(\textsc{ALG}(\mathcal{U}, \hat{q}), \hat q) \leq \alpha \cdot \textsc{OPT}(\hat q) \ \ \ \text{ and } \ \ \ C_{\text{max}}^{\text{rob}}(\textsc{ALG}(\mathcal{U}, \hat q), \mathcal{U}) \leq \beta \cdot \OPTR(\mathcal{U}).$$
\end{definition}

Note that the definition of $\alpha$-consistency matches the definition of $\alpha$-approximation for makespan minimization when the prediction is correct. The definition of $\beta$-robustness matches the definition of $\beta$-robustness for robust makespan minimization over $\mathcal{U}$, thus the prediction may differ arbitrarily from the realized scenario $q \in \mathcal U$. Despite the shared terminology, there is an important conceptual difference between robustness in robust optimization, which is the one we use, and the notion of robustness that is commonly used in algorithms with predictions. In algorithms with predictions, the benchmark for robustness $\beta_{\text{ALPS}}$ is typically the optimal solution had the true instance been known, which would correspond to  robustness being defined in  robust makespan minimization as
$$\beta_{\text{ALPS}} = \max_{\mathcal{U}, \hat{q} \in \mathcal{U}} \max_{q \in \mathcal{U}} \frac{C_{\max}(\textsc{ALG}(\mathcal{U}, \hat{q}), q)}{\min_{\mathcal{S}}C_{\max}(\mathcal{S}, q)}  \ \text{ in contrast to } \ \beta_{\text{RO}} = \max_{\mathcal{U}, \hat{q} \in \mathcal U}  \frac{\max_{q \in \mathcal{U}} C_{\max}(\textsc{ALG}(\mathcal{U}, \hat{q}), q)}{ \min_{\mathcal{S}} \max_{q \in \mathcal{U}} C_{\max}(\mathcal{S}, q)},$$
which is an equivalent expression for the robustness $\beta_{\text{RO}}$ from robust optimization
 where the benchmark is the optimal robust schedule $\mathcal{S}$ that does not know the true instance $q$. In Appendix~\ref{app:impossibility},
we show that, under the alternative measure $\beta_{\text{ALPS}}$ for robustness, no deterministic algorithm can achieve an $o(m)$-approximation and no randomized algorithm can achieve an $o(\log m / \log \log m)$-approximation, even in the special case of identical machines and interval uncertainty.

\par\medskip\noindent \emph{Machine models.}
The general model for makespan minimization defined above is the \emph{unrelated machines} model. We also study several special cases. In the \emph{identical machines} model, each job~$j$ has a size $p_j$ and $q_{i,j} = p_j$ for all machines $i$.
In \emph{restricted assignment}, each job $j$ has a size $p_j$ and an eligibility set $M_j \subseteq [m]$ of machines, and $q_{i,j} = p_j$ if $i \in M_j$ and $q_{i,j} = \infty$ if $i \notin M_j$. In the \emph{related machines} model, each job $j$ has a size $p_j$, each machine $i$ has a speed $s_i$, and $q_{i,j} = p_j / s_i$. The identical machines model is a special case of both restricted assignment and related machines. In the restricted-related machines model, each job $j$ has a size $p_j$ and an eligibility set $M_j \subseteq [m]$ of machines, each machine $i$ has a speed $s_i$, and $q_{i,j} = p_j / s_i$ if $i \in M_j$ and $q_{i,j} = \infty$ if $i \notin M_j$.

\par\medskip\noindent \emph{Uncertainty models.} The general model for robust makespan minimization defined above is the \emph{arbitrary uncertainty} model, where $\mathcal{U}$ is an arbitrary collection of scenarios. We also consider two structured special cases.
In the \emph{interval uncertainty} model for identical machines, each job~$j$ has an interval $[p^-_j, p^+_j]$ of possible sizes and $\mathcal{U} = \{q : q_{i,j} = p_j \in [p^-_j, p^+_j] \text{ for all } i, j\}.$ The interval uncertainty model extends from identical machines to related machines, restricted assignment, and restricted-related machines by maintaining the same uncertainty on the job sizes and having fixed machine speeds and eligibility sets across all scenarios in the uncertainty set. For unrelated machines, we have intervals $[q^-_{ij}, q^+_{ij}]$  for all $i, j$ and  $\mathcal{U} = \{q : q_{i,j} \in [q^-_{ij}, q^+_{ij}] \text{ for all } i, j\}$.

 In the \emph{budgeted uncertainty} model for identical machines, in addition to   intervals $[p^-_j, p^+_j]$,  there is a cardinality budget $\Gamma \geq 1$ such that the adversary can modify the  sizes of at most $\Gamma$ jobs: $$\mathcal{U} = \{q : \exists \ T \subseteq [n] \text { s.t. } |T| \leq \Gamma, q_{i,j}  = p_j^{-} \text{ if } j \not \in T \text{ and } q_{i,j} = p_j \in [p_{j}^{-}, p_{j}^+] \text{ if } j \in T\}.$$
 The budgeted uncertainty model extends to other machine models similarly to interval uncertainty.  The interval uncertainty model is the special case of the budgeted uncertainty model where $\Gamma = n$.

\par\medskip\noindent \emph{Prediction error.} We say that a predicted scenario $\hat{q}$ has prediction error $\eta$ with respect to scenario $q$ if $\eta = \max_{i,j}\left\{\frac{q_{i,j}}{\hat{q}_{i,j}}, \frac{\hat{q}_{i,j}}{q_{i,j}}\right\}.$\footnote{We abuse notation and write $x/0=1$ if $x= 0 $ and $x/0 = \infty$ otherwise.} An algorithm \textsc{ALG} is $\gamma(\eta)$-smooth for uncertainty set $\mathcal U$ if, for any prediction error $\eta \geq 1$,  scenario $q\in \mathcal{U}$, and predicted scenario $\hat{q} \in \mathcal{U}$ with prediction error $\eta$ with respect to $q$, we have that $C_{\text{max}}(\text{ALG}(\mathcal U, \hat{q}), q) \leq \gamma(\eta)  \cdot \textsc{OPT}(q).$ Note that $\gamma$-smoothness implies $\gamma$-consistency by taking $q=\hat{q}$. In fact, we would like to note that,  for all our algorithms, $\alpha$-consistency implies $\eta^2\alpha$-smoothness.

\begin{restatable}{lemma}{LemConsImpliesSmoothness}
\label{lem:consistency-implies-smoothness}
Let \textsc{ALG} be an algorithm for a fixed machine environment and uncertainty set $\mathcal U$. Suppose that \textsc{ALG} is $\alpha$-consistent for $\mathcal U$. Then \textsc{ALG} is $\eta^2\alpha$-smooth for $\mathcal U$.
\end{restatable}

The proof (in \Cref{app:proof-consistency-implies-smoothness}) follows from the costs of a fixed schedule and an optimal schedule both changing by at most a factor $\eta$ when evaluated under $q$ instead of $\hat{q}$.
\section{Interval uncertainty}
\label{sec:interval-uncertainty}

We begin with the interval uncertainty model. This case already captures the basic consistency--robustness tension: robustness is measured against one extreme scenario, while the prediction defines a second benchmark. Our main results  are a
consistency--robustness tradeoff for restricted-related machines~(\Cref{subsec:threshold-splitting-rr}) and an impossibility result for unrelated machines (\Cref{subsec:interval-unrelated-impossibility}). In \Cref{subsec:interval-identical-hardness}, we show that the consistency-robustness tradeoff achieved by our algorithm is asymptotically tight when the consistency tends to~$1$, even for identical machines. Recall that the robustness objective is to minimize the worst-case makespan $\max_{q \in \mathcal U} C_{\max}(\mathcal S,q).$ For any interval uncertainty set $\mathcal U$, the worst-case makespan for every schedule is attained at the  worst-case scenario $q^+ \in \mathcal U$, i.e., the scenario such that $q \le q^+$ for all $q \in \mathcal U$. Therefore, robustness reduces to minimizing the makespan on  $q^+$. In the remainder of this section, we  refer to an interval uncertainty set $\mathcal U$ only through its worst-case scenario $q^+$.

\subsection{The algorithm  for restricted-related machines}
\label{subsec:threshold-splitting-rr}

Recall that for restricted-related machines, a scenario $q$ is given by job sizes $p=(p_1,\ldots,p_n)$, eligibility sets $M_1,\ldots,M_n$, and machine speeds $s_1,\ldots,s_m$, with $q_{i,j}=p_j/s_i$ if $i\in M_j$ and $q_{i,j}=\infty$ otherwise. Given $J\subseteq[n]$, let $q[J]$ denote the subinstance restricted to jobs $J$.

\paragraph{Algorithm description.}
 \tsAlg, formally defined in \Cref{alg:threshold-splitting-rr},  first partitions the jobs between predicted heavy jobs $\hat{H}$ and $[n] \setminus \hat{H}$. Predicted heavy jobs  $\hat{H}$ are jobs  whose normalized predicted size $\hat{p}_j / \OPT(\hat{q})$ is at least   their normalized worst-case size $p^+_j / \OPT(q^+)$ up to a factor $\lambda$, where $\lambda>0$ is a parameter chosen by the algorithm that controls the consistency-robustness tradeoff. An optimal schedule $(\hat{S}_1, \ldots, \hat{S}_m)$ is then computed for the predicted heavy jobs $\hat{H}$ over the predicted scenario $\hat{q}$, as well as an optimal schedule $(S_1^+, \ldots, S_m^+)$ for $[n] \setminus \hat{H}$, but over the worst-case scenario $q^+$. The final schedule merges the schedules for jobs $\hat{H}$ and jobs $[n] \setminus \hat{H}$.

\begin{algorithm}[H]
\caption{\tsAlg algorithm for interval uncertainty and restricted-related machines}
\label{alg:threshold-splitting-rr}
\begin{algorithmic}[1]
\setlength{\itemsep}{0.25em}
\Input   predicted and worst-case scenarios $\hat{q}$ and $q^+$, corresponding predicted and worst-case job sizes $\hat{p}$ and $p^+$, and parameter $\lambda>0$
\State $\hat{H} \leftarrow \Big\{ j\in [n]: \lambda  \cdot \frac{\hat{p}_j}{\OPT(\hat{q})} >  \frac{p^+_j}{\OPT(q^+)}\Big\}$ \Comment{Predicted heavy jobs}
\State $(\hat{S}_1, \ldots, \hat{S}_m) \leftarrow \arg\min_{\mathcal S} C_{\text{max}}(\mathcal S, \hat{q}[\hat{H}])$ \Comment{Optimal schedule for jobs $\hat{H}$ over scenario $ \hat{q}$}
\State $(S_1^+, \ldots, S_m^+) \leftarrow \arg\min_{\mathcal S} C_{\text{max}}(\mathcal S, q^+[[n] \setminus \hat{H}])$ \Comment{Optimal  schedule for jobs $[n] \setminus \hat{H}$ over $q^+$}
\State \Return $(\hat{S}_1 \cup S^+_1,\ldots,\hat{S}_m \cup S^+_m)$ \Comment{Merge the two schedules}
\end{algorithmic}
\end{algorithm}

Our main result for \tsAlg is the following.

\begin{theorem}
\label{thm:threshold-splitting-rr}
For any $\lambda>0$, \tsAlg has consistency $1+1/\lambda$ and robustness $1+ \lambda$ for interval uncertainty and restricted-related machines.
\end{theorem}
In particular, for any constant $\lambda > 0$, the algorithm obtains constant consistency and constant robustness. Since identical machines, restricted assignment, and related machines are all special cases of restricted-related machines, the theorem applies to all three models.
\begin{proof}[Proof of \Cref{thm:threshold-splitting-rr}.]
We first show the consistency guarantee. The load of machine $i$ over scenario~$\hat{q}$ with job sizes $\hat{p}$ is $\frac{1}{s_i}\sum_{j\in \hat{S}_i \cup S^+_i}\hat{p}_j$. Next, we have  $\frac{1}{s_i} \sum_{j\in \hat{S}_i}\hat{p}_j \leq \text{OPT}(\hat{q}[\hat{H}]) \leq \text{OPT}(\hat{q})$ where the first inequality is since $(\hat{S}_1, \ldots, \hat{S}_m)$ is an optimal schedule for subinstance $\hat{q}[\hat{H}]$ and the second since $\hat{H} \subseteq [n]$ and by the monotonicity of the optimal makespan. We also have

$$ \frac{1}{s_i} \sum_{j\in S^+_i}\hat{p}_j \leq \frac{ \text{OPT}(\hat{q})}{\lambda s_i\cdot \text{OPT}(q^+)} \sum_{j\in S^+_i} p^+_j \leq \frac{ \text{OPT}(\hat{q}) \cdot \text{OPT}(q^+[[n] \setminus \hat{H}])}{ \lambda \cdot \text{OPT}(q^+)} \leq \frac{\text{OPT}(\hat{q})}{\lambda}$$
where the first inequality is by definition of $\hat{H}$ and since jobs $j \in S_i^+$ are not in $\hat{H}$, the second since $(S_1^+, \ldots, S_m^+)$ is an optimal schedule for subinstance $q^+[[n] \setminus \hat{H}]$, and the third  by the monotonicity of \text{OPT}. Thus, each machine has load at most $(1+1/\lambda) \text{OPT}(\hat{q})$ over scenario $\hat{q}$.

 The robustness analysis follows similarly to the consistency analysis, but is for scenario~$q^+$ instead of $\hat{q}$. We have  $\frac{1}{s_i} \sum_{j\in S^+_i} p^+_j \leq \text{OPT}(q^+[[n] \setminus \hat{H}]) \leq \text{OPT}(q^+)$
and
$$\frac{1}{s_i} \sum_{j\in \hat{S}_i} p^+_j \leq \frac{\lambda \cdot \text{OPT}(q^+)}{ s_i\cdot \text{OPT}(\hat{q})} \sum_{j\in \hat S_i} \hat{p}_j \leq \frac{\lambda \cdot \text{OPT}(q^+)\cdot \text{OPT}(\hat{q}[\hat{H}])}{\text{OPT}(\hat{q})} \leq \lambda \cdot \text{OPT}(q^+).$$ Therefore, each machine has load at most $(1+\lambda) \text{OPT}(q^+)$ over scenario $q^+$.
\end{proof}

We would like to emphasize that \tsAlg does not run in polynomial time. By estimating $\OPT(\cdot)$ and  approximating $\arg\min_{\mathcal S} C_{\text{max}}(\mathcal S, \cdot)$  with an approximation algorithm for makespan minimization, we get the following consistency and robustness. The proof is deferred to \Cref{app:proof-polynomial-implementation-threshold-splitting}.

\begin{restatable}{corollary}{CorPolyimeThresholdSplitting}
\label{thm:polynomial-implementation-threshold-splitting}
Given a polynomial-time $\gamma$-approximation algorithm for makespan minimization on restricted-related machines, there is, for every $\lambda>0$,  a polynomial-time algorithm for interval uncertainty and restricted-related machines with consistency $\gamma(1+\gamma/\lambda)$ and robustness $\gamma(1+\gamma\lambda)$.
\end{restatable}
Using the $2$-approximation of \cite{lenstra1990approximation}, we get a consistency of $2 + 4/\lambda$ and a robustness of $2 + 4 \lambda$.

\subsection{Asymptotically tight lower bound for consistency-robustness tradeoff}
\label{subsec:interval-identical-hardness}

We show that the linear dependence on $\lambda$ in the robustness guarantee of \Cref{thm:threshold-splitting-rr} is asymptotically tight. This lower bound is structural and independent of computational limitations.

\begin{restatable}{theorem}{IntervalIndenticalLowerBound}
\label{thm:interval-identical-lower}
For every $m\geq2$, every integer $\lambda$ with $1\leq\lambda\leq m-1$, and every $\alpha<1+1/\lambda$, no $\alpha$-consistent algorithm can achieve robustness strictly better than $\lambda$ if deterministic, or strictly better than $\lambda/4$ if randomized, even for interval uncertainty and identical machines.
\end{restatable}

We give the construction of the hard instance for deterministic algorithms and defer the full proof to \Cref{app:proof-interval-identical-lower}. Fix $L\geq m-1$. There are $\lambda$ jobs with $(\hat p_j,p^+_j)=(1/\lambda,L)$ and $m-1$ jobs with $(\hat p_j,p^+_j)=(1,1)$. The first type is negligible for the prediction but dominates the worst case, while the second type determines the predicted optimum. Enforcing consistency better than $1+1/\lambda$ forces the first type to be grouped together, creating a worst-case load $\lambda L$.

\subsection{Impossibility for unrelated machines}
\label{subsec:interval-unrelated-impossibility}

Interval uncertainty does not admit an instance-independent consistency-robustness tradeoff on unrelated machines. Already with one job and two machines, the prediction can make one machine uniquely favorable while the worst-case scenario makes the other machine uniquely robust.

\begin{restatable}{theorem}{IntervalUnrelatedNoTradeoff}
\label{thm:interval-unrelated-no-constant-tradeoff}
For every constants $\alpha\geq 1$ and $\beta\geq 1$, there is an unrelated-machines instance with interval uncertainty, two machines, and one job such that no  randomized algorithm is both $\alpha$-consistent and $\beta$-robust.
\end{restatable}

Consequently, unrelated machines admit no instance-independent constant consistency-robustness tradeoff for the min-max interval benchmark. We give  the construction (full proof in \Cref{app:proof-interval-unrelated-no-constant-tradeoff}).
Fix constants $\alpha,\beta\geq1$, choose $A>\alpha$, set $\rho_0:=\frac{A-\alpha}{A-1}>0$,
and choose $B>\max\{A,\beta A/\rho_0\}$.
There are two machines and one job $j$. On machine $1$, let $(\hat p_{1,j},p^+_{1,j})=(1,B)$, and on machine $2$, let $(\hat p_{2,j},p^+_{2,j})=(A,A)$. Any $\alpha$-consistent schedule must assign $j$ to machine $1$, but the robust optimum assigns it to machine $2$ and has value $A$. Thus, every $\alpha$-consistent schedule has worst-case makespan $B>\beta A$.
\section{Budgeted uncertainty}
\label{sec:budgeted-uncertainty}

In this section, we consider the budgeted uncertainty model. In the interval uncertainty model, the size of every job $j$ may vary within $[p^-_j,p^+_j]$, unrelated to other jobs. Under a cardinality budget $\Gamma$, in contrast, at most $\Gamma$ jobs may deviate from their lower size, while every remaining job has size $p^-_j$. The worst-case scenario in this uncertainty model depends on the schedule, unlike the interval uncertainty model where a single fixed scenario is the worst-case scenario for all schedules. Note the interval uncertainty model is recovered when $\Gamma=n$.

In \Cref{subsec:budgeted-cardinality}, we study the cardinality-budget model and present our main positive result: a reduction of the budgeted uncertainty set to an interval instance, based on a cutoff chosen from the min--max robust optimum.
Combining this reduction with \Cref{alg:threshold-splitting-rr} gives a consistency-robustness tradeoff for restricted assignment and identical machines. In \Cref{subsec:budgeted-related-impossibility}, we show that this positive result relies on the restricted-assignment structure. For related machines, even the cardinality-budget model with $\Gamma=1$ admits no instance-independent constant consistency-robustness tradeoff with respect to the min-max benchmark. In \Cref{sec:oblivious-cardinality}, we show that for identical machines the budget need not be known in advance: a single schedule can be constant-consistent and constant-robust simultaneously for all cardinality budgets. Natural weighted extensions are deferred to \Cref{app:budgeted-additional-models}; weighted fractional budgets admit the same tradeoff as cardinality budgets, while weighted binary budgets lose an additional factor of two in the robustness guarantee.

\subsection{Cardinality budget}
\label{subsec:budgeted-cardinality}

Under a cardinality budget, each job $j\in[n]$ has an interval $[p^-_j,p^+_j]$, and the adversary may modify the sizes of at most $\Gamma$ jobs. For an integer budget $\Gamma\ge 1$, the uncertainty set,
$\mathcal U=\{q:\exists\; T\subseteq[n], |T|\le\Gamma,\ q_{i,j}=p^-_j \text{ if } j\notin T,\text{ and } q_{i,j}=p_j\in[p^-_j,p^+_j] \text{ if } j\in T\}$.
Since the makespan of a fixed schedule is nondecreasing in every job size, the worst realization supported on $T$ sets every job $j\in T$ to $p^+_j$. Consistency is measured with respect to the predicted scenario $\hat q$, and robustness with respect to the min--max benchmark $\OPT^{\mathrm{rob}}(\mathcal U)$.

\paragraph{Algorithm description.}
Our algorithm \btiAlg, described in~\Cref{alg:budgeted-cardinality-restricted}, reduces the cardinality-budgeted problem to an interval instance. Specifically, we choose a cutoff $\tau$ based on the min--max robust optimum and subtract this cutoff from every possible deviation. The residual deviations define effective worst-case sizes~$p_j^{+,\tau}$. We then apply
\tsAlg to the predicted scenario $\hat q$ and the resulting worst-case scenario~$q^{+,\tau}$. Let $q(p, \mathcal M)$ denote the processing times corresponding to job sizes $p$ and eligibility sets $\mathcal M = \{M_j\}_{j\in [n]}$. Similarly, we let $\mathcal U(p^-, p^+, \mathcal M, \Gamma)$ denote the budgeted uncertainty set for restricted-assignment machines corresponding to $p^-, p^+, \mathcal M,$ and~$\Gamma$.

\begin{algorithm}[H]
\caption{\btiAlg for budgeted uncertainty and restricted assignment}
\label{alg:budgeted-cardinality-restricted}
\begin{algorithmic}[1]
\setlength{\itemsep}{0.25em}
\Input Lower, predicted, and worst-case sizes  $p^-, \hat{p},$ and $p^+$, eligibility sets $\mathcal M = \{M_j\}_{j}$, budget $\Gamma$,  and parameter $\lambda>0$
\State $\tau \leftarrow \OPT^{\text{rob}}(\mathcal U(p^-, p^+, \mathcal M, \Gamma))/\Gamma$ \Comment{Truncation threshold using the optimal robust makespan}
\State $p_j^{+,\tau} \leftarrow p^-_j +\big(p^+_j -  p^-_j -\tau\big)_+$ for every job $j\in[n]$ \Comment{Effective worst-case job sizes}
\State \Return $\tsAlg(q(\hat{p}, \mathcal M), q(p^{+, \tau}, \mathcal M),\lambda)$ \Comment{Apply the interval-uncertainty algorithm}
\end{algorithmic}
\end{algorithm}

\paragraph{Algorithm analysis.} We show that \btiAlg inherits the same consistency guarantee as \tsAlg, while losing only an additional additive constant in the robustness bound. As before, the parameter $\lambda$ controls how much the algorithm trusts the predicted scenario. We have the following theorem.

\begin{theorem}
\label{thm:budgeted-cardinality-restricted}
For any $\lambda>0$, \btiAlg has consistency $1+1/\lambda$ and robustness $2+\lambda$ for cardinality-budget uncertainty and restricted assignment.
\end{theorem}

We would like to note that the linear dependence on $\lambda$ in the robustness guarantee is asymptotically tight since the lower bound from \Cref{thm:interval-identical-lower} applies to budgeted uncertainty.

The analysis of \Cref{alg:budgeted-cardinality-restricted} relies on a duality argument. For a fixed set of jobs, the worst cardinality-budgeted load can be
upper bounded by the load in an effective interval instance, up to an additive term $\Gamma\theta$, where the effective worst-case sizes are obtained by truncating deviations at a cutoff $\theta$. The next lemma formalizes this observation. We let $\mathcal{P}(\mathcal U) = \{p(q) : q \in \mathcal U\}$ denote the uncertain job sizes $p(q)$ corresponding to the uncertain processing times $q \in \mathcal U$.
\begin{lemma}
\label{lem:budgeted-cardinality-reduction}
Let $ p_j^{+,\theta} := p^-_j+\bigl(p_j^+-p^-_j-\theta\bigr)_+$, then, for any budgeted uncertainty set $\mathcal U$ with budget $\Gamma$, jobs $J\subseteq [n]$ and $\theta\geq 0$, we have $\max_{p\in\mathcal P(\mathcal U)}\sum_{j\in J} p_j\leq \Gamma\theta+\sum_{j\in J}p_j^{+,\theta}$.
\end{lemma}
\begin{proof}
Fix  uncertainty set $\mathcal U$ and corresponding uncertain sizes $\mathcal P$, $J\subseteq [n]$ and $\theta \geq 0$. The maximum additional size $\max_{p \in \mathcal P} \sum_{j\in J} p_j - p_{j}^-$ of $J$ above its lower size  is given by the  linear program $(P)$ below. Since $\Gamma$ is integral, the feasible region is a cardinality-constrained unit-box polytope  and is integral. Hence the optimal value of (P) is $\max_{p \in \mathcal P} \sum_{j\in J} p_j - p_{j}^-$.
\begin{figure}[H]
\centering
\begin{minipage}[t]{0.45\textwidth}
\vspace{0pt}
\[
\begin{alignedat}{3}
\text{(P)}\quad
\max \quad & \sum_{j\in J} (p_j^+ - p^-_j) z_j \\
\text{s.t.}\quad
& \sum_{j\in J} z_j \le \Gamma \\
& z_j \le 1 \qquad && \forall j\in J \\
& z_j \ge 0 \qquad && \forall j\in J .
\end{alignedat}
\]
\end{minipage}
\hfill
\begin{minipage}[t]{0.45\textwidth}
\vspace{0pt}
\[
\begin{alignedat}{3}
\text{(D)}\quad
\min \quad & \Gamma\theta' + \sum_{j\in J}\mu_j \\
\text{s.t.}\quad
& \theta' + \mu_j \ge p_j^+ - p^-_j
\qquad && \forall j\in J \\
& \theta' \ge 0 \\
& \mu_j \ge 0
\qquad && \forall j\in J .
\end{alignedat}
\]
\end{minipage}
\label{fig:cardinality-budget-primal-dual}
\end{figure}

The dual of linear program (P) is linear program (D). For any value of dual variable $\theta'\geq0$, the corresponding optimal value of dual variable $\mu_j$ is $\mu_j=\bigl(p_j^+-p^-_j-\theta'\bigr)_+$ for every $j\in J$. Strong duality applies since the primal linear program is feasible and bounded. We conclude that
\begin{align}
\max_{p\in\mathcal P}\sum_{j\in J}p_j & =\sum_{j\in J}p^-_j+ \max_{p \in \mathcal P} \sum_{j\in J} p_j - p_{j}^- \nonumber \\
& = \sum_{j\in J}p^-_j+  \min_{\theta'\geq0}\left\{\Gamma\theta'+\sum_{j\in J}\bigl(p_j^+-p^-_j-\theta'\bigr)_+\right\} \label{eq:one} \\
& \leq\sum_{j\in J}p^-_j+\Gamma\theta+\sum_{j\in J}\bigl(p_j^+-p^-_j-\theta\bigr)_+ \nonumber \\
& =\Gamma\theta+\sum_{j\in J}p_j^{+,\theta} \label{eq:three}
\end{align}
 where \eqref{eq:one} is since the optimal value of (P) is $\max_{p \in \mathcal P} \sum_{j\in J} p_j - p_{j}^-$ and by strong duality and \eqref{eq:three} by the definition of $p_j^{+,\theta}$.
\end{proof}

We now apply \Cref{lem:budgeted-cardinality-reduction} to prove \Cref{thm:budgeted-cardinality-restricted}. We choose the cutoff $\tau$ from the min-max optimum $\OPT^{\text{rob}}(\mathcal U)$, and we define the effective worst-case job sizes $p_j^{+,\tau}$ and the corresponding scenario $q^{+,\tau}$, and we run \Cref{alg:threshold-splitting-rr} on the interval instance $(\hat q,q^{+,\tau})$. \Cref{lem:budgeted-cardinality-reduction} bounds the loss when transferring the robustness guarantee from $q^{+,\tau}$ to the original cardinality-budget uncertainty set.

\begin{proof}[Proof of \Cref{thm:budgeted-cardinality-restricted}.]
Let $\tau:=\OPT^{\mathrm{rob}}(\mathcal U)/\Gamma$. For every job $j\in[n]$, define $p_j^{+,\tau}:=p^-_j+\bigl(p_j^+-p^-_j-\tau\bigr)_+$, and let $q^{+,\tau}$ be the restricted-assignment scenario with job sizes $p^{+,\tau}$ and the same eligibility sets as $\hat q$. Let $\mathcal T=(T_1,\ldots,T_m)$ be a schedule attaining $\OPT^{\mathrm{rob}}(\mathcal U)$. Thus, for every machine $i\in[m]$, $\max_{q\in\mathcal U}\sum_{j\in T_i}q_{i,j}\leq\OPT^{\mathrm{rob}}(\mathcal U)$.

We first show that $\OPT(q^{+,\tau})\leq\OPT^{\mathrm{rob}}(\mathcal U)$. Fix a machine $i\in[m]$, and let $R_i:=\{j\in T_i:p_j^+-p^-_j>\tau\}$. We have $|R_i|\leq\Gamma$. Indeed, if $|R_i|\geq\Gamma+1$, then there are $\Gamma$ jobs in $T_i$ whose deviations are strictly larger than $\tau$. A scenario that modifies these jobs gives a load strictly larger than $\Gamma\tau=\OPT^{\mathrm{rob}}(\mathcal U)$ on machine $i$, which contradicts the preceding inequality.

Let $(p^+-p^-)_{T_i,(\ell)}$ denote the $\ell$-th largest value of $p_j^+-p^-_j$ among the jobs in $T_i$, where missing terms are taken to be zero. Since $|R_i|\leq\Gamma$, we have $\sum_{j\in T_i}\bigl(p_j^+-p^-_j-\tau\bigr)_+\leq\sum_{\ell=1}^{\Gamma}(p^+-p^-)_{T_i,(\ell)}$. Therefore,
$$
\sum_{j\in T_i}p_j^{+,\tau}\leq\sum_{j\in T_i}p^-_j+
\sum_{\ell=1}^{\Gamma}(p^+-p^-)_{T_i,(\ell)}
\leq
\OPT^{\mathrm{rob}}(\mathcal U),
$$
where the first inequality is by the definition of $p_j^{+,\tau}$ and the previous inequality, and the second since the right-hand side is the maximum load of machine $i$ over the cardinality-budget uncertainty set. Since this holds for every machine of schedule $\mathcal T$, we have $\OPT(q^{+,\tau})\leq\OPT^{\mathrm{rob}}(\mathcal U)$.

\Cref{alg:budgeted-cardinality-restricted} runs \Cref{alg:threshold-splitting-rr} on scenarios $\hat q$ and $q^{+,\tau}$. By \Cref{thm:threshold-splitting-rr}, the returned schedule $\mathcal S=(S_1,\ldots,S_m)$ satisfies $C_{\text{max}}(\mathcal S,\hat q)\leq\left(1+\frac{1}{\lambda}\right)\OPT(\hat q)$ and $C_{\text{max}}(\mathcal S,q^{+,\tau})\leq(1+\lambda)\OPT(q^{+,\tau})\leq(1+\lambda)\OPT^{\mathrm{rob}}(\mathcal U)$,
where the second inequality follows from $\OPT(q^{+,\tau})\leq\OPT^{\mathrm{rob}}(\mathcal U)$.

It remains to bound the makespan of $\mathcal S$ over the original uncertainty set. By \Cref{lem:budgeted-cardinality-reduction}, for every machine $i\in[m]$, $\max_{q\in\mathcal U}\sum_{j\in S_i}q_{i,j}\leq\Gamma\tau+\sum_{j\in S_i}p_j^{+,\tau}=\OPT^{\mathrm{rob}}(\mathcal U)+\sum_{j\in S_i}p_j^{+,\tau}.$
Thus,
$$
C_{\max}^{\mathrm{rob}}(\mathcal S,\mathcal U)=\max_{q\in\mathcal U}C_{\text{max}}(\mathcal S,q)\leq\OPT^{\mathrm{rob}}(\mathcal U)+C_{\text{max}}(\mathcal S,q^{+,\tau})\leq (\lambda+2)\OPT^{\mathrm{rob}}(\mathcal U),
$$
where the first inequality follows by taking the maximum over machines in the preceding inequality and the second by the robustness guarantee of \Cref{alg:threshold-splitting-rr}.
\end{proof}

In \Cref{thm:polytime-cardinality-budget} in \Cref{app:proof-polynomial-implementation-cardinality-budget}, we show that by estimating $\OPT^{\mathrm{rob}}(\mathcal U)$ with the $3$-approximation algorithm for cardinality-budgeted uncertainty on unrelated machines due to Bougeret et~al.~\cite{DBLP:journals/mst/BougeretJPR21} and applying \Cref{thm:polynomial-implementation-threshold-splitting} to the resulting interval instance, we get a polynomial-time implementation of \Cref{alg:budgeted-cardinality-restricted} with consistency $2+4/\lambda$ and robustness $5+4\lambda$. For the special case of identical machines, an even stronger EPTAS is known for the robust black box~\cite{DBLP:journals/mst/BougeretJPR21}.

\subsection{Impossibility for related machines}
\label{subsec:budgeted-related-impossibility}

The guarantee of Algorithm~\ref{alg:budgeted-cardinality-restricted} relies on the restricted-assignment structure and does not extend to related machines. For related machines, even the cardinality-budget model with $\Gamma=1$ admits no instance-independent constant consistency-robustness tradeoff. In particular, once machine speeds may differ, a single deviating job is sufficient to rule out any constant tradeoff that is independent of the ratio between the largest and smallest machine speeds.
The next theorem shows that, for $\Gamma=1$, no schedule can simultaneously achieve constant consistency with respect to the predicted scenario $\hat q$ and constant robustness with respect to $\OPT^{\mathrm{rob}}(\mathcal U)$.

\begin{restatable}{theorem}{BudgetedRelatedNoTradeoff}
\label{thm:budgeted-related-no-constant-tradeoff}
For every constants $\alpha \ge 1$ and $\beta \ge 1$, there exists a related-machines instance with cardinality budget $\Gamma = 1$ such that no randomized algorithm is both $\alpha$-consistent and $\beta$-robust.
\end{restatable}

Consequently, related machines admit no universal constant consistency-robustness guarantee for cardinality-budget uncertainty relative to the min-max benchmark. We give the construction (full proof in \Cref{app:proof-budgeted-related-no-constant-tradeoff}). Fix constants $\alpha\geq 1$ and $\beta\geq 1$, and choose integers $k>2\beta$ and $\ell\geq 2\alpha k$. The instance consists of one machine of speed $k$, $\ell$ machines of speed $1$, and $n:=k+\ell$ jobs. Every job $j\in[n]$ has $(p^-_j,\hat p_j,p^+_j)=(1,1,L)$, where $L$ is chosen sufficiently large so that $kL/(2(n+L-1))>\beta$. The predicted optimum is $1$. Consistency forces the schedule to put enough probability mass on the slow machines: deterministically, at least one job is assigned to a machine of speed $1$, and in expectation, some fixed job is assigned to a machine of speed $1$ with probability at least $1/2$. The adversary then uses the budget $\Gamma=1$ to deviate this job to size $L$. The robust optimum is at most $(n+L-1)/k$, by assigning every job to the machine of speed $k$. Thus, the deterministic robustness ratio is at least $kL/(n+L-1)$, and the randomized expected robustness ratio is at least $kL/(2(n+L-1))>\beta$.

\subsection{Oblivious cardinality budgets}
\label{sec:oblivious-cardinality}

The guarantees above are budget-aware: the algorithm receives the cardinality budget $\Gamma$ and constructs a schedule for the corresponding uncertainty set $\mathcal U_\Gamma$. Somewhat surprisingly, for identical machines, one can still obtain constant consistency and robustness without knowing $\Gamma$. That is, there is a single schedule whose robustness guarantee holds simultaneously for all budgets $\Gamma\in[n]$.

The construction differs from the budget-to-interval reduction used above, which depends explicitly on the value of $\Gamma$. Instead, we normalize each budgeted uncertainty set by its own robust optimum and take the union of the resulting normalized sets. This produces an arbitrary uncertainty set. The algorithm for arbitrary uncertainty sets developed in \Cref{sec:general-uncertainty} then becomes the key tool: applying it to this normalized uncertainty set yields a schedule whose robust makespan controls the robustness ratio for every budget at once. The details are deferred to \Cref{app:oblivious-cardinality}.

\begin{restatable}{theorem}{ThmObliviousCardinalityTradeoff}
\label{thm:oblivious-cardinality-tradeoff}
For identical machines with cardinality-budget uncertainty, for every $\lambda>0$, there is a deterministic algorithm that computes a schedule independently of the budget $\Gamma$ and is $(1+2/\lambda)$-consistent and, simultaneously for every $\Gamma\in[n]$, $(10\lambda+20)$-robust with respect to $\mathcal U_\Gamma$.
\end{restatable}
\section{General uncertainty}
\label{sec:general-uncertainty}

Under a general uncertainty set, the adversary may choose any processing-time vector from a nonempty bounded set $\mathcal U\subseteq\mathbb R_+^n$. Consistency is measured with respect to the predicted scenario $\hat q$, while robustness is measured with respect to the min-max benchmark $\OPT^{\mathrm{rob}}(\mathcal U)$. For every scenario $q\in\mathcal U$, we write $p(q)$ the corresponding vector of job sizes, and for every set of jobs $J\subseteq[n]$, let $\Psi_{\mathcal U}(J):=\sup_{q\in\mathcal U}\sum_{j\in J}p(q)_j$ denote its maximum load over the uncertainty set.

General uncertainty sets are the most expressive uncertainty model considered in this paper. For identical machines, we obtain a constant consistency-robustness tradeoff by using only the predicted loads and the set function $\Psi_{\mathcal U}$. For restricted assignment, however, arbitrary uncertainty sets are sufficiently expressive to rule out any instance-independent constant tradeoff.

\subsection{Identical machines}
\label{subsec:general-identical}

\paragraph{Algorithm description.}
Fix a parameter $\lambda>0$. The algorithm combines an optimal schedule $(\hat S_1,\ldots,\hat S_m)$ for the predicted scenario with an optimal min-max schedule $(T_1,\ldots,T_m)$ over $\mathcal U$. The predicted schedule determines which jobs remain fixed and how much predicted load should be reassigned to each machine, while the min-max schedule determines which jobs are grouped together during the repacking.

A job is predicted large if $\hat p_j>\OPT(\hat q)/\lambda$. These jobs form the set $\hat H$ and retain their assignments from the predicted optimum. Thus, machine $i$ keeps the jobs $H_i=\hat S_i\cap\hat H$. The remaining predicted load on machine $i$ is $y_i=\sum_{j\in\hat S_i\setminus\hat H}\hat p_j$. This is the load contributed in the predicted optimum by the jobs that are not fixed, and it is used as the target load when blocks are reassigned to machine $i$.

The jobs outside $\hat H$ are grouped according to the optimal min-max schedule. For every machine $i$, the jobs in $T_i\setminus\hat H$ are ordered arbitrarily and partitioned into blocks. A non-tiny block is formed by taking the shortest remaining prefix whose predicted load exceeds $\OPT(\hat q)/(2\lambda)$. Since every job outside $\hat H$ has predicted size at most $\OPT(\hat q)/\lambda$, the predicted load of such a block is at most $3\OPT(\hat q)/(2\lambda)$. When the remaining jobs on $T_i$ have total predicted load at most $\OPT(\hat q)/(2\lambda)$, they form one final tiny block. Hence, every machine of the min-max schedule contributes at most one tiny block, while every block remains contained in one machine of that schedule. This containment is used to control its worst-case load over $\mathcal U$.

The tiny blocks are assigned injectively to distinct machines. The non-tiny blocks are then ordered arbitrarily and assigned using the residual predicted loads $y_1,\ldots,y_m$. Starting with the first unassigned block, machine $i$ receives consecutive blocks until their total predicted load first exceeds $y_i$. If all remaining blocks have total predicted load at most $y_i$, they are all assigned to machine $i$. This first-crossing rule ensures that the predicted load assigned to machine $i$ exceeds its target $y_i$ by at most one non-tiny block. Finally, the algorithm combines the fixed jobs $H_i$ with the repacked jobs $L_i$ on every machine.

\begin{algorithm}[h]
\caption{\bpAlg algorithm for general uncertainty and identical machines}
\label{alg:general-uncertainty-identical}
\begingroup
\small
\algrenewcommand\algorithmicindent{0.8em}
\begin{algorithmic}[1]
\setlength{\itemsep}{0.15em}
\Input predicted scenario $\hat q$, corresponding predicted job sizes $\hat p$, uncertainty set $\mathcal U$, and parameter $\lambda>0$
\State $(\hat S_1,\ldots,\hat S_m)\leftarrow\arg\min_{\mathcal S}C_{\text{max}}(\mathcal S,\hat q)$ \Comment{Optimal predicted schedule}
\State $(T_1,\ldots,T_m)\leftarrow\arg\min_{\mathcal T}C_{\max}^{\mathrm{rob}}(\mathcal T,\mathcal U)$ \Comment{Optimal min-max schedule}
\State $\hat H\leftarrow\{j\in[n]:\hat p_j>\OPT(\hat q)/\lambda\}$ \Comment{Predicted large jobs}
\For{$i\in[m]$}
    \State $H_i\leftarrow\hat S_i\cap\hat H$ and $y_i\leftarrow\sum_{j\in\hat S_i\setminus\hat H}\hat p_j$ \Comment{Fixed jobs and residual predicted load}
\EndFor
\State $\mathcal B_{\mathrm{nt}}\leftarrow\varnothing$ and $\mathcal B_{\mathrm{tiny}}\leftarrow\varnothing$
\For{$i\in[m]$}
    \State Let $R\leftarrow T_i\setminus\hat H$, ordered arbitrarily
    \While{$\sum_{j\in R}\hat p_j>\OPT(\hat q)/(2\lambda)$}
        \State Let $B$ be the shortest prefix of $R$ such that $\sum_{j\in B}\hat p_j>\OPT(\hat q)/(2\lambda)$
        \State $\mathcal B_{\mathrm{nt}}\leftarrow\mathcal B_{\mathrm{nt}}\cup\{B\}$ and $R\leftarrow R\setminus B$ \Comment{Create a non-tiny block}
    \EndWhile
    \If{$R\neq\varnothing$}
        \State $\mathcal B_{\mathrm{tiny}}\leftarrow\mathcal B_{\mathrm{tiny}}\cup\{R\}$ \Comment{Create the remaining tiny block}
    \EndIf
\EndFor
\State Assign the blocks in $\mathcal B_{\mathrm{tiny}}$ injectively to distinct machines
\For{$i\in[m]$}
    \State $L_i\leftarrow\bigcup\{B\in\mathcal B_{\mathrm{tiny}}:B\text{ is assigned to machine }i\}$ \Comment{Initialize the repacked jobs}
\EndFor
\State Order $\mathcal B_{\mathrm{nt}}$ arbitrarily as $B^1,\ldots,B^N$, and set $t\leftarrow1$ \Comment{Initialize the first-crossing assignment}
\For{$i\in[m]$}
    \If{$t>N$}
        \State \textbf{break}
    \ElsIf{$\sum_{s=t}^{N}\sum_{j\in B^s}\hat p_j\leq y_i$}
        \State $L_i\leftarrow L_i\cup\bigcup_{s=t}^{N}B^s$ \Comment{Assign all remaining blocks when they fit}
        \State $t\leftarrow N+1$
        \State \textbf{break}
    \Else
        \State Let $\ell\geq t$ be the smallest index such that $\sum_{s=t}^{\ell}\sum_{j\in B^s}\hat p_j>y_i$ \Comment{First crossing of capacity $y_i$}
        \State $L_i\leftarrow L_i\cup\bigcup_{s=t}^{\ell}B^s$ \Comment{Assign the first-crossing block prefix}
        \State $t\leftarrow\ell+1$
    \EndIf
\EndFor
\State \Return $(H_1\cup L_1,\ldots,H_m\cup L_m)$ \Comment{Merge the fixed jobs and the repacked jobs}
\end{algorithmic}
\endgroup
\end{algorithm}

\paragraph{Algorithm analysis.}
\bpAlg combines an optimal schedule for the predicted scenario with an optimal min-max schedule over $\mathcal U$. Lines~1--2 compute these two schedules. Line~3 partitions the jobs into the predicted large jobs $\hat H$ and the predicted small jobs $\hat L=[n]\setminus \hat H$. Lines~4--6 keep every job in $\hat H$ on its machine in the optimal predicted schedule and define $y_i$ as the predicted load of the small jobs assigned to machine $i$ in that schedule.

The jobs in $\hat L$ are then repacked using the min-max schedule. For every machine $i$, Lines~8--15 partition $T_i\cap\hat L$ into non-tiny blocks of predicted load slightly larger than $\OPT(\hat q)/(2\lambda)$ and at most one remaining tiny block. Line~17 assigns the tiny blocks injectively to distinct machines. Lines~18--31 assign the non-tiny blocks in consecutive order using the values $y_i$: each machine receives blocks until its total predicted load first exceeds $y_i$, unless all remaining blocks fit within $y_i$. The next result states the consistency-robustness tradeoff achieved by \Cref{alg:general-uncertainty-identical}.

\begin{restatable}{theorem}{GeneralTradeoff}
\label{thm:general-uncertainty-tradeoff}
For any $\lambda>0$, \bpAlg has consistency $1+2/\lambda$ and robustness $2\lambda+4$ for general uncertainty sets and identical machines.
\end{restatable}

The analysis separates the predicted large and small jobs at the threshold $\OPT(\hat q)/\lambda$. Since every predicted large job has size greater than $\OPT(\hat q)/\lambda$, each machine of the optimal predicted schedule contains fewer than $\lambda$ such jobs. The remaining jobs are partitioned into blocks whose predicted loads are controlled by the cutoff $\OPT(\hat q)/(2\lambda)$, which bounds the additional predicted load introduced by the repacking. For robustness, every block is contained in a single machine of the optimal min-max schedule, and therefore its maximum load over $\mathcal U$ is at most $\OPT^{\mathrm{rob}}(\mathcal U)$. Combining the bounds on the number of predicted large jobs and the number of blocks assigned to each machine gives the stated consistency and robustness. The full proof is deferred to \Cref{app:proof-general-uncertainty-tradeoff}.

In \Cref{thm:polytime-general-uncertainty} in \Cref{app:proof-polytime-general-uncertainty}, we show that by replacing the predicted optimal schedule in \Cref{alg:general-uncertainty-identical} with a poly-time $\gamma$-approximate schedule for makespan minimization on identical machines, and replacing the optimal min-max schedule over $\mathcal U$ with a poly-time $\gamma_{\mathrm{rob}}$-approximate robust schedule, the same block construction and first-crossing assignment give a polynomial-time implementation with consistency $\gamma(1+2/\lambda)$ and robustness $(2\lambda+4)\gamma_{\mathrm{rob}}$ for every $\lambda>0$.

\subsection{Impossibility for restricted assignment}
\label{subsec:general-restricted-impossibility}

The positive result for identical machines relies on the fact that every job can be assigned to any machine. Under restricted assignment, arbitrary uncertainty sets can encode correlations that prevent such a repacking argument. The hard instance partitions the jobs into color classes. The robust optimum assigns all jobs of each color to a dedicated machine, so that every machine is exposed to at most one color. In contrast, any schedule with a small predicted makespan must distribute the jobs of each color across several machines. Hence, some machine contains jobs from many distinct colors. The adversary then selects one job from each of these colors to realize a large processing time, which creates a large load on that machine.

\begin{restatable}{theorem}{GeneralRestrictedNoTradeoff}
\label{thm:general-restricted-no-constant-tradeoff}
For every constants $\alpha\geq1$ and $\beta\geq1$, there is a restricted-assignment instance with a nonempty bounded uncertainty set $\mathcal U$ such that no randomized algorithm is both $\alpha$-consistent and $\beta$-robust.
\end{restatable}
We give the construction of the hard instance and defer the full proof to \Cref{app:proof-general-restricted-no-constant-tradeoff}. Fix constants $\alpha\geq1$ and $\beta\geq1$, and choose integers $k>2\beta$, $\ell\geq2\alpha k$, and $L$ sufficiently large so that $kL/2>\beta(\ell+L)$. There are $\ell$ machines, the first $k$ of which are color machines. For every color $c\in[k]$ and every index $i\in[\ell]$, there is one job $j_{c,i}$ with $\hat p_{j_{c,i}}=1$ and eligibility set $M_{j_{c,i}}=\{c,i\}$. The uncertainty set contains one scenario for every function $\sigma:[k]\to[\ell]$; in that scenario, job $j_{c,\sigma(c)}$ has size $1+L$, while every other job has size $1$. The robust schedule keeps every color on its color machine, giving a robust load at most $\ell+L$. In contrast, consistency forces many jobs to leave their color machines, so some machine contains jobs from many colors. The adversary then chose one job from each represented color on that machine to deviate, creating expected makespan at least $kL/2>\beta(\ell+L)$. Thus, no deterministic or randomized algorithm can achieve a constant consistency-robustness tradeoff.

\section{Concluding Remarks}

We introduce predictions into robust optimization and develop a general consistency--robustness framework, which we instantiate for robust makespan scheduling. Our results give a structural classification across standard uncertainty models and machine environments, showing when constant consistency and constant robustness can be achieved simultaneously.

A central open question is to determine the optimal consistency--robustness Pareto frontiers beyond interval uncertainty. In particular, can budgeted uncertainty admit the same $(1+1/\lambda,1+\lambda)$ tradeoff as interval uncertainty, at least on identical or restricted-assignment machines, or is the additional loss in robustness inherent? More generally, the best possible dependence between consistency and robustness for arbitrary uncertainty on identical machines remains open.

We hope that this work stimulates further research on leveraging predictions in robust optimization beyond makespan scheduling. Simple two-choice examples show that such extensions cannot be expected for arbitrary selection problems in a black-box way; see \Cref{prop:two-solution-selection-separation} in \Cref{app:selection-separation}. A natural next step is therefore to identify robust combinatorial optimization problems with enough structure to combine prediction-optimal and robust solutions with bounded loss.

\newpage

\section*{Acknowledgments}

Eric Balkanski was supported by NSF grant CCF-2210501. Nicole Megow was supported by the Deutsche Forschungsgemeinschaft (DFG, German Research Foundation) through project no. 547924951 and under Germany's Excellence Strategy – EXC 3036 – project no. 533607631.

\printbibliography

\newpage

\appendix
\crefalias{section}{appendix}
\crefalias{subsection}{appendix}
\crefalias{subsubsection}{appendix}

\section*{Appendices}

\section{Impossibility for Alternative Robustness Benchmark}
\label{app:impossibility}

In algorithms with predictions, the benchmark for robustness $\beta_{\text{ALPS}}$ is typically the optimal solution had the true instance been known, which would correspond to this alternative notion of robustness being defined in  robust makespan minimization as
$$\beta_{\text{ALPS}} = \max_{\mathcal{U}, \hat{q} \in \mathcal{U}} \max_{q \in \mathcal{U}} \frac{C_{\max}(\textsc{ALG}(\mathcal{U}, \hat{q}), q)}{\min_{\mathcal{S}}C_{\max}(\mathcal{S}, q)},$$
which is in contrast to the standard robustness notion from robust optimization
 where the benchmark is the optimal robust schedule $\mathcal{S}$ that does not know the true instance $q$.

We show that under this alternative  robustness measure $\beta_{\mathrm{ALPS}}$, no deterministic algorithm can achieve an $o(m)$-approximation and no randomized algorithm can achieve an $o(\log m / \log \log m)$-approximation, even in the special case of identical machines and interval uncertainty. Since these lower bounds hold for the problem without predictions, they also hold without imposing any consistency requirement.

\begin{theorem}
\label{thm:alps-interval-impossibility}
For every sufficiently large integer $m$, there exists an identical-machines instance with $m$ machines, $n=m^2$ jobs, an interval uncertainty set $\mathcal U_m$, and a predicted instance $\hat q\in\mathcal U_m$ such that every deterministic algorithm satisfies
$$
\max_{q\in\mathcal U_m}
\frac{C_{\text{max}}(\textsc{ALG}(\mathcal U_m,\hat q),q)}{\OPT(q)}
\geq
\frac{m}{2}.
$$
Moreover, every randomized algorithm satisfies
$$
\max_{q\in\mathcal U_m}
\frac{\mathbb E\!\left[C_{\text{max}}(\textsc{ALG}(\mathcal U_m,\hat q),q)\right]}{\OPT(q)}
=
\Omega\left(\frac{\log m}{\log\log m}\right),
$$
where the expectation is over the internal randomness of the algorithm.
\end{theorem}

\begin{proof}
Fix a sufficiently large integer $m$ and let $n:=m^2$. Every job $j\in[n]$ has predicted job size $\hat p_j=1$ and upper job size $p_j^+=m^3$. Thus, $\mathcal U_m=[1,m^3]^{m^2}$. Let $\hat q$ be the predicted instance, so that every job has size $1$ under $\hat q$. The total predicted job size is $m^2$, and therefore $\OPT(\hat q)\geq m$. Assigning exactly $m$ jobs to every machine gives predicted makespan $m$, so $\OPT(\hat q)=m$. This equality is not needed in the lower-bound argument, but it shows that the predicted instance is balanced.

We first consider a deterministic algorithm. Let $\mathcal S=(S_1,\ldots,S_m):=\textsc{ALG}(\mathcal U_m,\hat q)$. Since $\sum_{i=1}^m|S_i|=m^2$, some machine $i^\star\in[m]$ satisfies $|S_{i^\star}|\geq m$. Choose any set $T\subseteq S_{i^\star}$ with $|T|=m$, and let $q^T\in\mathcal U_m$ be the instance in which every job in $T$ has size $m^3$ and every job in $[n]\setminus T$ has size $1$. All jobs in $T$ are assigned to machine $i^\star$, and therefore $C_{\text{max}}(\mathcal S,q^T)\geq m\cdot m^3=m^4$.

The instance $q^T$ contains exactly $m$ jobs of size $m^3$ and $m^2-m$ jobs of size $1$. Its total job size is $m^4+m^2-m$, so every schedule has a makespan of at least $(m^4+m^2-m)/m=m^3+m-1$. Conversely, assigning one job of $T$ and exactly $m-1$ jobs of $[n]\setminus T$ to every machine gives a load $m^3+m-1$ on every machine. Hence $\OPT(q^T)=m^3+m-1$. It follows that
$$
\frac{C_{\text{max}}(\mathcal S,q^T)}{\OPT(q^T)}
\geq
\frac{m^4}{m^3+m-1}
\geq
\frac{m}{2},
$$
where the last inequality is equivalent to $m^4-m^2+m\geq0$. Since the maximum is taken over all instances in $\mathcal U_m$, the instance $q^T$ may depend on the deterministic schedule $\mathcal S$.

We next consider a randomized algorithm. The instance used below is selected independently of the internal randomness of the algorithm. Let $T\subseteq[n]$ be a random set obtained by including every job independently with probability $1/m$, and let $q^T$ be the instance in which every job in $T$ has size $m^3$ and every job in $[n]\setminus T$ has size $1$.

Fix an arbitrary deterministic schedule $\mathcal S=(S_1,\ldots,S_m)$. For every machine $i\in[m]$, let $X_i:=|S_i\cap T|$, and let $M:=\max_{i\in[m]}X_i$ and $K:=|T|$. Since every job counted by $X_i$ has size $m^3$, we have $C_{\text{max}}(\mathcal S,q^T)\geq m^3M$.

Set $t:=\left\lfloor\frac{\log m}{16\log\log m}\right\rfloor$. For all sufficiently large $m$, we have $t\geq1$ and $t\leq m/2$. We first prove that $\text{Pr}_T[M\geq t]\geq1/2$.

For every machine $i\in[m]$, let $s_i:=|S_i|$ and $\mu_i:=s_i/m$. Then $X_i\sim\operatorname{Bin}(s_i,1/m)$ and $\mathbb E_T[X_i]=\mu_i$. Moreover, $\sum_{i=1}^m\mu_i=\frac{1}{m}\sum_{i=1}^ms_i=m$.

Suppose first that $\mu_i\geq4t$ for some machine $i$. Since $t\leq\mu_i/4$, the event $\{X_i<t\}$ is contained in the event $\{X_i\leq\mu_i/2\}$. The multiplicative Chernoff bound gives $\text{Pr}_T[X_i<t]\leq\text{Pr}_T[X_i\leq\mu_i/2]\leq\exp(-\mu_i/8)\leq\exp(-t/2)$. Since $t$ tends to infinity with $m$, we have $\exp(-t/2)\leq1/2$ for all sufficiently large $m$. Thus, $\text{Pr}_T[M\geq t]\geq\text{Pr}_T[X_i\geq t]\geq1/2$.

Suppose now that $\mu_i<4t$ for every machine $i$, and let $I:=\{i\in[m]:\mu_i\geq1/2\}$. The machines outside $I$ contribute less than $m/2$ to $\sum_{i=1}^m\mu_i=m$, and therefore $\sum_{i\in I}\mu_i>m/2$. Since $\mu_i<4t$ for every $i\in I$, we obtain $|I|\geq m/(8t)$.

Fix a machine $i\in I$. Since $1/2\leq\mu_i<4t$, we have $m/2\leq s_i<4mt$. In particular, $s_i\geq t$ for all sufficiently large $m$. Therefore,
$$
\text{Pr}_T [X_i\geq t]
\geq
\text{Pr}_T [X_i=t]
=
\binom{s_i}{t}
\left(\frac{1}{m}\right)^t
\left(1-\frac{1}{m}\right)^{s_i-t}.
$$
We next lower bound the two factors in this expression. Since $s_i\geq t$,
$$
\binom{s_i}{t}
=
\prod_{r=0}^{t-1}\frac{s_i-r}{t-r}
\geq
\left(\frac{s_i}{t}\right)^t,
$$
where the inequality follows because $(s_i-r)/(t-r)\geq s_i/t$ is equivalent to $r(s_i-t)\geq0$. Hence $\binom{s_i}{t}(1/m)^t\geq(s_i/(tm))^t=(\mu_i/t)^t\geq(1/(2t))^t$.

For every $m\geq2$, we have $\log(1-1/m)\geq-2/m$. Therefore,
$$
\left(1-\frac{1}{m}\right)^{s_i-t}
\geq
\exp\left(-\frac{2(s_i-t)}{m}\right)
\geq
\exp\left(-\frac{2s_i}{m}\right)
>
e^{-8t},
$$
where the last inequality follows from $s_i/m=\mu_i<4t$. Combining the two estimates gives $\text{Pr}_T[X_i\geq t]\geq(1/(2t))^te^{-8t}$.

We now verify that $(1/(2t))^te^{-8t}\geq m^{-1/8}$ for all sufficiently large $m$. Let $L:=\log m$ and $\ell:=\log\log m$. Since $t\leq L/(16\ell)$ and $\log(2t)\leq\ell$ for all sufficiently large $m$, we have $t\log(2t)\leq L/16$. Moreover, $8t\leq L/(2\ell)\leq L/16$ whenever $\ell\geq8$. Hence $t\log(2t)+8t\leq L/8$, and therefore $(1/(2t))^te^{-8t}=\exp(-t\log(2t)-8t)\geq e^{-L/8}=m^{-1/8}$. Thus, $\text{Pr}_T[X_i\geq t]\geq m^{-1/8}$ for every $i\in I$.

The random variables $(X_i)_{i\in I}$ are mutually independent because the sets $(S_i)_{i\in I}$ are pairwise disjoint and each $X_i$ depends only on the independent inclusion indicators of the jobs in $S_i$. Consequently,
$$
\text{Pr}_T[M<t]
\leq
\prod_{i\in I}\text{Pr}_T[X_i<t]
\leq
(1-m^{-1/8})^{|I|}
\leq
\exp\left(-m^{-1/8}|I|\right)
\leq
\exp\left(-\frac{m^{7/8}}{8t}\right).
$$
Since $t=O(\log m/\log\log m)$, the final expression is at most $1/2$ for all sufficiently large $m$. Hence $\text{Pr}_T[M\geq t]\geq1/2$ also in this case.

It remains to control the total number of jobs in $T$. Since every one of the $m^2$ jobs is included independently with probability $1/m$, we have $K\sim\operatorname{Bin}(m^2,1/m)$ and $\mathbb E_T[K]=m$. The multiplicative Chernoff bound gives $\text{Pr}_T[K\geq2m]\leq(e/4)^m$. In particular, $\text{Pr}_T[K>2m]\leq1/4$ for all sufficiently large $m$. Therefore,
$$
\text{Pr}_T[M\geq t\text{ and }K\leq2m]
\geq
\text{Pr}_{T}[M\geq t]-\text{Pr}_T[K>2m]
\geq
\frac{1}{4}.
$$

Consider an instance $q^T$ for which $K\leq2m$. The $K$ jobs of size $m^3$ can be assigned so that every machine receives at most $\lceil K/m\rceil\leq2$ of them. The remaining $m^2-K$ jobs have size $1$ and can be assigned so that every machine receives at most $\lceil(m^2-K)/m\rceil\leq m$ of them. Hence $\OPT(q^T)\leq2m^3+m\leq3m^3$.

On the event $E:=\{M\geq t\text{ and }K\leq2m\}$, we have $C_{\text{max}}(\mathcal S,q^T)\geq m^3t$ and $\OPT(q^T)\leq3m^3$. Hence, for every deterministic schedule $\mathcal S$,
$$
\mathbb E_T\left[\frac{C_{\text{max}}(\mathcal S,q^T)}{\OPT(q^T)}\right]
\geq \mathbb E_T\left[\frac{C_{\text{max}}(\mathcal S,q^T)}{\OPT(q^T)}\mathbf 1_E\right] \geq \frac{t}{3}\text{Pr}_{T}[E] \geq \frac{t}{12}.
$$
The first inequality follows since the ratio is nonnegative, the second since $C_{\text{max}}(\mathcal S,q^T)/\OPT(q^T)\geq t/3$ on the event $E$, and the last since $\text{Pr}_T[E]\geq1/4$.

Let now $\mathcal S:=\textsc{ALG}(\mathcal U_m,\hat q)$ be the random schedule returned by a randomized algorithm. The random set $T$ is selected independently of the internal randomness of the algorithm. Conditioning on the realized schedule and applying the preceding bound, we obtain
$$
\mathbb E_T\left[\frac{\mathbb E_{\textsc{ALG}}\!\left[C_{\text{max}}(\mathcal S,q^T)\right]}{\OPT(q^T)}\right] = \mathbb E_T\mathbb E_{\textsc{ALG}}\left[\frac{C_{\text{max}}(\mathcal S,q^T)}{\OPT(q^T)}\right] = \mathbb E_{\textsc{ALG}}\mathbb E_T\left[\frac{C_{\text{max}}(\mathcal S,q^T)}{\OPT(q^T)}\right] \geq \frac{t}{12}.
$$
The first equality follows since $\OPT(q^T)$ is independent of the internal randomness of the algorithm, the second follows by interchanging the two finite expectations, and the inequality follows from the preceding bound after conditioning on the realized schedule.

By averaging over $T$, there exists a fixed set $T^\star\subseteq[n]$ such that $ \frac{\mathbb E_{\textsc{ALG}}\!\left[C_{\text{max}}(\mathcal S,q^{T^\star})\right]}{\OPT(q^{T^\star})} \geq \frac{t}{12}$.
Indeed, otherwise every term in the preceding expectation would be strictly smaller than $t/12$, which is impossible.

The instance $q^{T^\star}$ is fixed independently of the realized random schedule. Therefore, the lower bound holds against an oblivious adversary. Finally, $t=\left\lfloor\log m/(16\log\log m)\right\rfloor=\Omega(\log m/\log\log m)$, and hence
$$
\max_{q\in\mathcal U_m}
\frac{\mathbb E_{\textsc{ALG}}\!\left[C_{\text{max}}(\mathcal S,q)\right]}{\OPT(q)}
=
\Omega\left(\frac{\log m}{\log\log m}\right).
$$

The proof does not use any consistency property of the algorithm. Since the deterministic lower bound is at least $m/2$ and the randomized lower bound is $\Omega(\log m/\log\log m)$, both ratios diverge with $m$. Thus, no deterministic or randomized algorithm has an instance-independent finite $\beta_{\text{ALPS}}$.
\end{proof}

\section{Proofs Missing from \Cref{sec:preliminaries}}
\label{app:proofs-section-preliminaries}

\subsection{Proof of \Cref{lem:consistency-implies-smoothness}}
\label{app:proof-consistency-implies-smoothness}

\LemConsImpliesSmoothness*
\begin{proof}
If $\eta = \max_{i,j}\left\{\frac{q_{i,j}}{\hat{q}_{i,j}}, \frac{\hat{q}_{i,j}}{q_{i,j}}\right\}$, then every processing time satisfies $q_{ij}\le \eta \hat q_{ij}$ and $\hat q_{ij}\le \eta q_{ij}$. Therefore, for every schedule $\mathcal{S}$, we have $C_{\text{max}}(\mathcal S,q)\le \eta C_{\text{max}}(\mathcal S,\hat q)$ and $C_{\max}(\mathcal S,\hat q)\le \eta C_{\max}(\mathcal S,q)$. Applying the second inequality to an optimal schedule for $q$ gives $\OPT(\hat q)\le \eta \OPT(q)$. Using the first inequality for the schedule $\textsc{ALG}(\mathcal U,\hat q)$, followed by $\alpha$-consistency for $\mathcal U$, and then the preceding optimum comparison, we obtain $C_{\text{max}}(\textsc{ALG}(\mathcal U,\hat q),q)\le \eta C_{\text{max}}(\textsc{ALG}(\mathcal U,\hat q),\hat q)\le \eta\alpha\OPT(\hat q)\le \eta^2\alpha\OPT(q)$. This proves the claim.
\end{proof}

\section{Proofs Missing from \Cref{sec:interval-uncertainty}} \label{app:proofs-section-3}

\subsection{Proof of \Cref{thm:polynomial-implementation-threshold-splitting}}
\label{app:proof-polynomial-implementation-threshold-splitting}
\CorPolyimeThresholdSplitting*
\begin{proof}
Let $A(q)$ be the output of the $\gamma$-approximation algorithm over instance $q$. We have $\OPT(\hat q) \le C_{\text{max}}(A(\hat{q}),\hat q)\le \gamma \cdot \OPT(\hat q)$ and $\OPT(q^+)\le C_{\text{max}}(A(q^+),q^+)\le \gamma \cdot\OPT(q^+)$. The polynomial-time algorithm uses these two makespans in the threshold. It defines $\hat{H} \leftarrow \Big\{ j\in [n]: \lambda\cdot\frac{\hat{p}_j}{C_{\text{max}}(A(\hat{q}),\hat q)}>\frac{p^+_j}{C_{\text{max}}(A(q^+),q^+)}\Big\}$. It then computes a $\gamma$-approximate schedule $(\hat{S}_1, \ldots, \hat{S}_m)$ for the induced subinstance $\hat{q}[\hat{H}]$, and a $\gamma$-approximate schedule $(S_1^+, \ldots, S_m^+)$ for the induced subinstance $q^+[[n] \setminus \hat{H}]$, and returns the merged schedule $(\hat{S}_1 \cup S^+_1,\ldots,\hat{S}_m \cup S^+_m)$.

\textbf{Consistency.} The load of machine $i$ over instance $\hat{q}$ with job sizes $\hat{p}$ is $\frac{1}{s_i}\sum_{j\in \hat{S}_i \cup S^+_i}\hat{p}_j$. We have

$$\frac{1}{s_i} \sum_{j\in \hat{S}_i}\hat{p}_j \leq \gamma \cdot\OPT(\hat{q}[\hat{H}]) \leq \gamma \cdot \OPT(\hat{q}) $$
where the first inequality is since $(\hat{S}_1, \ldots, \hat{S}_m)$ is a $\gamma$-approximate optimal schedule for instance $\hat{q}[\hat{H}]$ and the second since $\hat{H} \subseteq [n]$ and by the monotonicity of the optimal makespan over instances ordered by inclusion. We also have

$$ \frac{1}{s_i} \sum_{j\in S^+_i}\hat{p}_j \leq \frac{ C_{\text{max}}(A(\hat{q}),\hat q)}{\lambda s_i\cdot C_{\text{max}}(A(q^+),q^+)} \sum_{j\in S^+_i} p^+_j \leq \frac{ \gamma^2 \cdot \OPT(\hat{q}) \cdot \OPT(q^+[[n] \setminus \hat{H}])}{ \lambda \cdot \OPT(q^+)} \leq \frac{\gamma^2 \cdot\OPT(\hat{q})}{\lambda}$$

where the first inequality is by definition of $\hat{H}$ and since jobs $j \in S_i^+$ are not in $\hat{H}$, the second since $C_{\text{max}}(A(\hat{q}),\hat q)\le \gamma\cdot\OPT(\hat q)$, $\OPT(q^+)\le C_{\text{max}}(A(q^+),q^+)$, and $(S_1^+, \ldots, S_m^+)$ is a $\gamma$-approximate optimal schedule for instance $q^+[[n] \setminus \hat{H}]$, and the third  by the monotonicity of $\OPT$. Thus, each machine has load at most $\gamma(1+\gamma/\lambda) \OPT(\hat{q})$ over instance $\hat{q}$.

\textbf{Robustness.} The robustness analysis follows similarly to the consistency analysis, but is for instance $q^+$ instead of $\hat{q}$. We have  $\frac{1}{s_i} \sum_{j\in S^+_i} p^+_j \leq \gamma\cdot \OPT(q^+[[n] \setminus \hat{H}]) \leq \gamma\cdot\OPT(q^+)$
and
$$\frac{1}{s_i} \sum_{j\in \hat{S}_i} p^+_j \leq \frac{\lambda \cdot C_{\text{max}}(A(q^+),q^+)}{ s_i\cdot C_{\text{max}}(A(\hat{q}),\hat q)} \sum_{j\in \hat S_i} \hat{p}_j \leq \frac{\gamma^2\cdot \lambda \cdot \OPT(q^+)\cdot \OPT(\hat{q}[\hat{H}])}{\OPT(\hat{q})} \leq \gamma^2 \cdot \lambda \cdot \OPT(q^+).$$ Thus, each machine has load at most $\gamma(1+\gamma\lambda) \OPT(q^+)$ over instance $q^+$.
\end{proof}

\subsection{Polynomial time implementation via vector scheduling}
\label{app:vector-scheduling-polytime-identical}

In this subsection, we record an alternative polynomial-time implementation for identical machines under interval uncertainty. This implementation is different from the threshold-splitting construction. It uses \Cref{thm:threshold-splitting-rr} only as an existence result, and then applies the fixed-dimensional vector-scheduling approximation algorithm of Chekuri and Khanna \cite{DBLP:journals/siamcomp/ChekuriK04}.

\begin{theorem} \label{thm:vector-scheduling-polytime-identical} For identical machines and interval uncertainty, for every $\lambda>0$ and every $\varepsilon>0$, there is a polynomial-time algorithm that is  $(1+\varepsilon)\left(1+\frac1\lambda\right)$-consistent and $(1+\varepsilon)(1+\lambda)$-robust.
\end{theorem}
\begin{proof}
Fix $\varepsilon_0,\varepsilon_1>0$ such that $(1+\varepsilon_0)(1+\varepsilon_1)\leq 1+\varepsilon.$ Let $A(q)$ be the output of a $(1+\varepsilon_0)$-approximation algorithm for identical-machine makespan minimization over instance $q$. We have $\OPT(\hat q)\leq C_{\max}(A(\hat q),\hat q)\leq (1+\varepsilon_0)\OPT(\hat q)$ and $\OPT(q^+)\leq C_{\max}(A(q^+),q^+)\leq(1+\varepsilon_0)\OPT(q^+)$.

Define a two-dimensional vector-scheduling instance in which every job $j\in[n]$ has vector $ v_j= \left(\frac{\hat p_j}{(1+1/\lambda)C_{\max}(A(\hat q),\hat q)},\frac{p_j^+}{(1+\lambda)C_{\max}(A(q^+),q^+)}\right)$. By \Cref{thm:threshold-splitting-rr}, there exists a schedule $\mathcal S$ such that $C_{\max}(\mathcal S,\hat q)\leq \left(1+\frac{1}{\lambda}\right)\OPT(\hat q)$ and $C_{\max}(\mathcal S,q^+)\leq (1+\lambda)\OPT(q^+)$.

Since $\OPT(\hat q)\leq C_{\max}(A(\hat q),\hat q)$ and $\OPT(q^+)\leq C_{\max}(A(q^+),q^+)$, the load of every machine under $\mathcal S$ in each coordinate of the vector-scheduling instance is at most one. Thus, the optimal maximum load of the vector-scheduling instance is at most one.

Run the fixed-dimensional vector-scheduling PTAS of Chekuri and Khanna on this instance with accuracy parameter $\varepsilon_1$. It returns a schedule $\mathcal S'$ whose load on every machine in each coordinate is at most $1+\varepsilon_1$. Therefore,
$$
\begin{aligned}
C_{\max}(\mathcal S',\hat q) &\leq (1+\varepsilon_1)\left(1+\frac{1}{\lambda}\right) C_{\max}(A(\hat q),\hat q)\\ &\leq (1+\varepsilon_1)(1+\varepsilon_0) \left(1+\frac{1}{\lambda}\right)\OPT(\hat q)\\ &\leq (1+\varepsilon)\left(1+\frac{1}{\lambda}\right)\OPT(\hat q)
\end{aligned}
$$
where the first inequality is since the vector-scheduling PTAS returns a schedule whose load in the first coordinate is at most $1+\varepsilon_1$, the second since $C_{\text{max}}(A(\hat q),\hat q)\leq (1+\varepsilon_0)\OPT(\hat q)$, and the third by the choice of $\varepsilon_0$ and $\varepsilon_1$ such that $(1+\varepsilon_0)(1+\varepsilon_1)\leq 1+\varepsilon$.

Similarly,
$$
\begin{aligned}
C_{\max}(\mathcal S',q^+) &\leq
(1+\varepsilon_1)(1+\lambda)C_{\max}(A(q^+),q^+)\\ &\leq (1+\varepsilon_1)(1+\varepsilon_0)(1+\lambda)\OPT(q^+)\\ &\leq (1+\varepsilon)(1+\lambda)\OPT(q^+)
\end{aligned}
$$
where the first inequality is since the vector-scheduling PTAS returns a schedule whose load in the second coordinate is at most $1+\varepsilon_1$, the second since $C_{\text{max}}(A(q^+),q^+)\leq (1+\varepsilon_0)\OPT(q^+)$, and the third by the choice of $\varepsilon_0$ and $\varepsilon_1$ such that $(1+\varepsilon_0)(1+\varepsilon_1)\leq 1+\varepsilon$.
\end{proof}

\subsection{Proof of \Cref{thm:interval-identical-lower}}
\label{app:proof-interval-identical-lower}

\IntervalIndenticalLowerBound*
\begin{proof}
We first prove the deterministic lower bound. Fix an integer $1\leq \lambda\leq m-1$, a constant $\alpha<1+1/\lambda$, and a parameter $L\geq m-1$. Consider an instance with $\lambda$ jobs satisfying $(\hat p_j,p_j^+)=(1/\lambda,L)$ and $m-1$ jobs satisfying $(\hat p_j,p_j^+)=(1,1)$. In particular, $\hat p_j\leq p_j^+$ for every job $j\in[n]$.

We first compute $\OPT(\hat q)$. Assign the $m-1$ jobs with predicted processing time $1$ to distinct machines and assign all $\lambda$ jobs with predicted processing time $1/\lambda$ to the remaining machine. This schedule has a predicted makespan of $1$. Since the instance contains a job with predicted processing time $1$, every schedule has predicted makespan at least $1$. Thus, $\OPT(\hat q)=1$.

We next compute $\OPT(q^+)$. Assign the $\lambda$ jobs with worst-case processing time $L$ to distinct machines and assign all $m-1$ remaining jobs to one of the machines that does not contain a job of the first type. Such a machine exists since $\lambda\leq m-1$. The worst-case load of this machine is $m-1$, while every machine containing a job of the first type has worst-case load $L$. Since $L\geq m-1$, we have $\OPT(q^+)\leq\max\{L,m-1\}=L$. The instance contains a job with worst-case processing time $L$, and therefore $\OPT(q^+)\geq L$. Thus, $\OPT(q^+)=L$.

Let $\mathcal S$ be a deterministic schedule satisfying $C_{\text{max}}(\mathcal S,\hat q)\leq\alpha\OPT(\hat q)=\alpha$. Since $\alpha<1+1/\lambda\leq2$, no machine can contain two jobs with predicted processing time $1$. Hence, the $m-1$ jobs of the second type occupy $m-1$ distinct machines. No job with predicted processing time $1/\lambda$ can be assigned to any of these machines, since this would give predicted load $1+1/\lambda>\alpha$. Thus, all $\lambda$ jobs of the first type are assigned to the remaining machine. The worst-case load of this machine is $\lambda L$, and therefore $C_{\text{max}}(\mathcal S,q^+)\geq\lambda L=\lambda\OPT(q^+)$. Hence, no deterministic algorithm with consistency strictly better than $1+1/\lambda$ can have robustness strictly smaller than $\lambda$.

We now prove the randomized lower bound. Let $h:=\lceil\lambda/2\rceil$. Since $\lambda\geq1$, we have $1\leq h\leq\lambda$, and since $\lambda\leq m-1$, we have $h\leq m-1$. Consider the same construction with $h$ in place of $\lambda$ and any $L\geq m-1$. Thus, there are $h$ jobs satisfying $(\hat p_j,p_j^+)=(1/h,L)$ and $m-1$ jobs satisfying $(\hat p_j,p_j^+)=(1,1)$. By the deterministic calculations above, we have $\OPT(\hat q)=1$ and $\OPT(q^+)=L$. Moreover, every deterministic schedule $\mathcal S$ satisfying $C_{\text{max}}(\mathcal S,\hat q)<1+1/h$ also satisfies $C_{\text{max}}(\mathcal S,q^+)\geq hL$.

Let $\mathcal S:=\textsc{ALG}(\mathcal U,\hat q)$ be the random schedule returned by the algorithm, and suppose that $\mathbb E[C_{\text{max}}(\mathcal S,\hat q)]\leq1+1/\lambda$. Let $E$ be the event that $C_{\text{max}}(\mathcal S,\hat q)\geq1+1/h$. Since every deterministic schedule has predicted makespan at least $\OPT(\hat q)=1$, we have
$$
\mathbb E\!\left[C_{\text{max}}(\mathcal S,\hat q)\right]
\geq
\text{Pr}[E^c]
+
\left(1+\frac{1}{h}\right)\text{Pr}[E]
=
1+\frac{\text{Pr}[E]}{h}.
$$
The inequality follows since the predicted makespan is at least $1$ on $E^c$ and at least $1+1/h$ on $E$. Since $\mathbb E_{\textsc{ALG}}[C_{\text{max}}(\mathcal S,\hat q)]\leq1+1/\lambda$, it follows that $\text{Pr}_{\textsc{ALG}}[E]\leq h/\lambda$.

On the event $E^c$, the deterministic property established above gives $C_{\text{max}}(\mathcal S,q^+)\geq hL$. On the event $E$, we use $C_{\text{max}}(\mathcal S,q^+)\geq\OPT(q^+)=L$. Therefore,
$$
\begin{aligned}
\mathbb E\!\left[C_{\text{max}}(\mathcal S,q^+)\right]
&\geq
hL\text{Pr}[E^c]+L\text{Pr}[E] =
L\left(h-(h-1)\text{Pr}[E]\right)\\
&\geq
L\left(h-\frac{h(h-1)}{\lambda}\right) =
L\frac{h(\lambda-h+1)}{\lambda}\\
&\geq
\frac{\lambda}{4}L =
\frac{\lambda}{4}\OPT(q^+).
\end{aligned}
$$
The first inequality follows from the lower bounds on the worst-case makespan on $E^c$ and $E$. The second inequality follows from $\text{Pr}[E]\leq h/\lambda$. For the last inequality, since $h=\lceil\lambda/2\rceil$, we have $h\geq\lambda/2$ and $\lambda-h+1\geq\lambda/2$. Hence $h(\lambda-h+1)/\lambda\geq\lambda/4$. The final equality follows from $\OPT(q^+)=L$. Thus, no randomized algorithm satisfying the expected consistency guarantee $1+1/\lambda$ can have expected robustness strictly smaller than $\lambda/4$.
\end{proof}

\subsection{Proof of \Cref{thm:interval-unrelated-no-constant-tradeoff}}
\label{app:proof-interval-unrelated-no-constant-tradeoff}
\IntervalUnrelatedNoTradeoff*

\begin{proof}
Fix constants $\alpha\geq1$ and $\beta\geq1$. Choose $A>\alpha$, set $\rho_0:=\frac{A-\alpha}{A-1}>0$,
and choose $B>\max\{A,\beta A/\rho_0\}$. Consider an instance with two machines and one job $j$. On machine $1$, let $p^-_{1,j}=\hat p_{1,j}=1$ and $p^+_{1,j}=B$, and on machine $2$, let $p^-_{2,j}=\hat p_{2,j}=p^+_{2,j}=A$. Thus, $p^-_{i,j}\leq \hat p_{i,j}\leq p^+_{i,j}$ for both machines.

We first compute $\OPT(\hat q)$. Assigning job $j$ to machine $1$ gives predicted makespan $1$, while assigning it to machine $2$ gives predicted makespan $A>1$. Thus $\OPT(\hat q)=1$.

Let $\mathcal S$ be a deterministic schedule satisfying $C_{\max}(\mathcal S,\hat q)\leq \alpha\OPT(\hat q)=\alpha$. If $\mathcal S$ assigns job $j$ to machine $2$, then $C_{\max}(\mathcal S,\hat q)=A>\alpha$, which is a contradiction. Hence, every deterministic $\alpha$-consistent schedule assigns job $j$ to machine $1$. For every such schedule, $C_{\max}(\mathcal S,q^+)=B$.

We next compute $\OPT(q^+)$. Assigning job $j$ to machine $1$ gives worst-case makespan $B$, while assigning it to machine $2$ gives worst-case makespan $A$. Since $B>A$, we have $\OPT(q^+)=A$. Therefore, every deterministic schedule $\mathcal S$ satisfying $C_{\max}(\mathcal S,\hat q)\leq \alpha\OPT(\hat q)$ also satisfies
$$
        \frac{C_{\max}(\mathcal S,q^+)}{\OPT(q^+)}=\frac{B}{A}>\beta.
$$
Thus, no deterministic algorithm is simultaneously $\alpha$-consistent and $\beta$-robust on this instance.

It remains to rule out randomized algorithms. Let $\textsc{ALG}$ be a randomized algorithm, and let $\rho$ be the probability that $\textsc{ALG}(\mathcal U,\hat q)$ assigns job $j$ to machine $1$. If $\textsc{ALG}$ is $\alpha$-consistent in expectation, then
$$
        \rho+(1-\rho)A
        =\mathbb E\!\left[C_{\max}(\textsc{ALG}(\mathcal U,\hat q),\hat q)\right]
        \leq \alpha\OPT(\hat q)=\alpha .
$$
Since $A>\alpha$, this implies
$$
        \rho\geq \frac{A-\alpha}{A-1}=\rho_0.
$$
Under the worst-case scenario $q^+$, the expected makespan is at least $\rho B$. Hence
$$
        \mathbb E\!\left[C_{\max}(\textsc{ALG}(\mathcal U,\hat q),q^+)\right]
        \geq \rho B
        \geq \rho_0 B
        > \beta A
        = \beta\OPT(q^+).
$$
Since interval uncertainty has $\OPT^{\mathrm{rob}}(\mathcal U)=\OPT(q^+)=A$, the randomized algorithm is not $\beta$-robust in expectation. Therefore, no randomized algorithm is simultaneously $\alpha$-consistent and $\beta$-robust on this instance.

\end{proof}

\section{Proofs Missing from \Cref{sec:budgeted-uncertainty}}
\label{app:proofs-section-4}

\subsection{Polynomial-time implementation of algorithm for budgeted uncertainty}
\label{app:proof-polynomial-implementation-cardinality-budget}

\begin{restatable}{corollary}{CorPolytimeCardinalityBudget}
\label{thm:polytime-cardinality-budget}
Assume that there is a polynomial-time $\gamma$-approximation algorithm for makespan minimization on restricted assignment and a polynomial-time $\gamma_{\mathrm{rob}}$-approximation algorithm for robust makespan minimization under cardinality-budget uncertainty. Then, for every $\lambda>0$, there is a polynomial-time algorithm with consistency $\gamma(1+\gamma/\lambda)$ and robustness $\gamma_{\mathrm{rob}}+\gamma(1+\gamma\lambda)$ for cardinality-budget uncertainty and restricted-assignment machines.
\end{restatable}

The black boxes can be instantiated with $\gamma=2$ using the classical $2$-approximation for unrelated machines of Lenstra, Shmoys, and Tardos~\cite{lenstra1990approximation}, and with $\gamma_{\mathrm{rob}}=3$ using the $3$-approximation for cardinality-budgeted uncertainty on unrelated machines due to Bougeret, Jansen, Poss, and Rohwedder~\cite{DBLP:journals/mst/BougeretJPR21}. This gives consistency $2+4/\lambda$ and robustness $5+4\lambda$. For the special case of identical machines, an even stronger EPTAS is known for the robust black box~\cite{DBLP:journals/mst/BougeretJPR21}.
\begin{proof}
Let $A^{\mathrm{rob}}(\mathcal U)$ be the output of the $\gamma_{\mathrm{rob}}$-approximation algorithm for robust makespan minimization over uncertainty set $\mathcal U$. We have $\OPT^{\mathrm{rob}}(\mathcal U)\leq\max_{q\in\mathcal U}C_{\text{max}}(A^{\mathrm{rob}}(\mathcal U),q)\leq\gamma_{\mathrm{rob}}\OPT^{\mathrm{rob}}(\mathcal U)$. The polynomial-time algorithm defines $ \tau:=\frac{C_{\max}^{\mathrm{rob}}(A^{\mathrm{rob}}(\mathcal U),\mathcal U)}{\Gamma}$ and $p_j^{+,\tau}:=p^-_j+\bigl(p_j^+-p^-_j-\tau\bigr)_+$ for every job $j\in[n]$. Let $q^{+,\tau}$ be the restricted-assignment instance with job sizes $p^{+,\tau}$ and the same eligibility sets as $\hat q$. Let $A(q)$ be the output of the $\gamma$-approximation algorithm for makespan minimization over instance $q$. The algorithm applies the polynomial-time implementation of \Cref{alg:threshold-splitting-rr} to instances $\hat q$ and $q^{+,\tau}$ using the makespans $C_{\text{max}}(A(\hat q),\hat q)$ and $C_{\text{max}}(A(q^{+,\tau}),q^{+,\tau})$.

\textbf{Consistency.} By \Cref{thm:polynomial-implementation-threshold-splitting}, the returned schedule $\mathcal S$ satisfies $C_{\text{max}}(\mathcal S,\hat q)\leq\gamma\left(1+\frac{\gamma}{\lambda}\right)\OPT(\hat q)$.

\textbf{Robustness.} We first show that $\OPT(q^{+,\tau})\leq\OPT^{\mathrm{rob}}(\mathcal U)$. Let $\mathcal T=(T_1,\ldots,T_m)$ be a schedule attaining $\OPT^{\mathrm{rob}}(\mathcal U)$. For every machine $i\in[m]$, let $R_i:=\{j\in T_i:p_j^+-p^-_j>\tau\}$. We have $|R_i|\leq\Gamma$. Indeed, if $|R_i|\geq\Gamma+1$, then a scenario that modifies any $\Gamma$ jobs in $R_i$ gives additional load strictly larger than $\Gamma\tau = \max_{q\in\mathcal U}C_{\text{max}}(A^{\mathrm{rob}}(\mathcal U),q) \geq\OPT^{\mathrm{rob}}(\mathcal U)$ on machine $i$, which contradicts the definition of $\mathcal T$.

Let $(p^+-p^-)_{T_i,(\ell)}$ denote the $\ell$-th largest value of $p_j^+-p^-_j$ among the jobs in $T_i$, where missing terms are taken to be zero. Since $|R_i|\leq\Gamma$, we have $\sum_{j\in T_i}\bigl(p_j^+-p^--\tau\bigr)_+\leq\sum_{\ell=1}^{\Gamma}(p^+-p^-)_{T_i,(\ell)}$. Therefore,$$\sum_{j\in T_i}p_j^{+,\tau}\leq\sum_{j\in T_i}p^-_j+\sum_{\ell=1}^{\Gamma}(p^+-p^-)_{T_i,(\ell)}\leq\OPT^{\mathrm{rob}}(\mathcal U),$$ where the first inequality is by the definition of $p_j^{+,\tau}$ and the second since the right-hand side is the maximum load of machine $i$ over the cardinality-budget uncertainty set. Since this holds for every machine of schedule $\mathcal T$, we have $\OPT(q^{+,\tau})\leq\OPT^{\mathrm{rob}}(\mathcal U)$.

By \Cref{thm:polynomial-implementation-threshold-splitting}, the returned schedule $\mathcal S=(S_1,\ldots,S_m)$ satisfies
$C_{\text{max}}(\mathcal S,q^{+,\tau})\leq\gamma(1+\gamma\lambda)\OPT(q^{+,\tau})\leq\gamma(1+\gamma\lambda)\OPT^{\mathrm{rob}}(\mathcal U)$, where the second inequality follows from $\OPT(q^{+,\tau})\leq\OPT^{\mathrm{rob}}(\mathcal U)$.

For every machine $i\in[m]$, \Cref{lem:budgeted-cardinality-reduction} applied to the set $S_i$ and $\theta=\tau$ gives $\max_{q\in\mathcal U}\sum_{j\in S_i}q_{i,j}\leq \Gamma\tau+\sum_{j\in S_i}p_j^{+,\tau}=\max_{q\in\mathcal U}C_{\text{max}}(A^{\mathrm{rob}}(\mathcal U),q)+\sum_{j\in S_i}p_j^{+,\tau}$. Thus,
$$
\max_{q\in\mathcal U}C_{\text{max}}(\mathcal S,q)\leq \max_{q\in\mathcal U}C_{\text{max}}(A^{\mathrm{rob}}(\mathcal U),q)+C_{\text{max}}(\mathcal S,q^{+,\tau})\leq\left(\gamma_{\mathrm{rob}}+\gamma(1+\gamma\lambda)\right)\OPT^{\mathrm{rob}}(\mathcal U),
$$
where the first inequality follows by taking the maximum over machines in the preceding inequality and the second since $A^{\mathrm{rob}}(\mathcal U)$ is a $\gamma_{\mathrm{rob}}$-approximate robust schedule and by the preceding bound on $C_{\text{max}}(\mathcal S,q^{+,\tau})$.

\end{proof}

\subsection{Proof of \Cref{thm:budgeted-related-no-constant-tradeoff}}
\label{app:proof-budgeted-related-no-constant-tradeoff}
\BudgetedRelatedNoTradeoff*
\begin{proof}
Fix constants $\alpha \ge 1$ and $\beta \ge 1$, and choose integers $k > 2\beta$ and $\ell \ge 2\alpha k$. Consider a related-machines instance with one machine of speed $k$ and $\ell$ machines of speed $1$. Let $n := k+\ell$. There are $n$ jobs, and every job $j\in[n]$ has lower job size $p^-_j=1$, predicted job size $\hat p_j=1$, and worst-case job size $p^+_j=L$, where $L\ge1$ is chosen such that $\frac{kL}{2(n+L-1)}>\beta$. Such a choice exists since $kL/(2(n+L-1))$ converges to $k/2>\beta$ as $L$ tends to infinity.

We first compute $\OPT(\hat q)$. The total predicted job size is $n$, while the total machine speed is $k+\ell=n$, and therefore $\OPT(\hat q)\geq 1$. Assigning $k$ jobs to the machine of speed $k$ and one job to each machine of speed $1$ gives a predicted makespan of $1$. Thus, $\OPT(\hat q)=1$.

Let $\mathcal S$ be a schedule satisfying $C_{\text{max}}(\mathcal S,\hat q)\leq\alpha\OPT(\hat q)=\alpha$. If $\mathcal S$ assigns every job to the machine of speed $k$, then $C_{\text{max}}(\mathcal S,\hat q)=\frac{n}{k}=1+\frac{\ell}{k}>\alpha$, where the last inequality follows from $\ell \ge  2\alpha k > (\alpha-1)k$. Hence, every $\alpha$-consistent schedule assigns some job $j$ to a machine of speed $1$. Consider the scenario in $\mathcal U$ in which job $j$ has size $p_j^+=L$ and every other job has size $p^-_j=1$. The load of the machine containing job $j$ is at least $L$, and therefore $C_{\text{max}}^{\mathrm{rob}}(\mathcal S,\mathcal U)= \max_{q\in\mathcal U} C_{\text{max}}(\mathcal S,q) \geq L$.

We next bound $\OPT^{\mathrm{rob}}(\mathcal U)$. Consider the schedule that assigns every job to the machine of speed $k$. Under any instance $q\in\mathcal U$, at most one job has size $L$, while every other job has size $1$. Thus, this schedule has makespan at most $(n+L-1)/k$, and hence $\OPT^{\mathrm{rob}}(\mathcal U)\leq\frac{n+L-1}{k}$. Therefore, every schedule $\mathcal S$ satisfying $C_{\text{max}}(\mathcal S,\hat q)\leq\alpha\OPT(\hat q)$ also satisfies
$$
\frac{C_{\text{max}}^{\mathrm{rob}}(\mathcal S,\mathcal U)}{\OPT^{\mathrm{rob}}(\mathcal U)}\geq\frac{L}{(n+L-1)/k}=\frac{kL}{n+L-1}>\beta.$$

We now prove that the same construction rules out randomized algorithms. Let $\textsc{ALG}$ be a randomized algorithm and $\mathcal S=\textsc{ALG}(\mathcal U,\hat q)$ be the random schedule returned by the algorithm, and let $X$ be the number of jobs assigned to machines of speed $1$. Since the load on the machine of speed $k$ is at least $(n-X)/k$, expected consistency implies

$$\alpha \ge \mathbb E\!\left[C_{\max}(\mathcal S,\hat q)\right] \ge \mathbb E\!\left[\frac{n-X}{k}\right] = \frac{n-\mathbb E[X]}{k}.$$

Thus $\mathbb E[X]\ge n-\alpha k=\ell-(\alpha-1)k$. Since $\ell\ge 2\alpha k$, we have $\mathbb E[X]\ge \frac n2$. Equivalently, the sum over all jobs of the probability that the job is assigned to a machine of speed $1$ is at least $n/2$. Hence, by averaging, there exists a job $j^\star$ that is assigned to a machine of speed $1$ with probability at least $1/2$.

Consider the scenario $q^{j^\star}$ in $\mathcal U$ in which job $j^\star$ has size $L$, while every other job has size $1$. Under this scenario, whenever $j^\star$ is assigned to a machine of speed $1$, the makespan is at least $L$. Therefore $\mathbb E\!\left[C_{\max}(\mathcal S,q^{j^\star})\right]\ge \frac L2$. On the other hand, assigning every job to the machine of speed $k$ gives $ \OPT^{\mathrm{rob}}(\mathcal U) \le \frac{n+L-1}{k}$.
Therefore,
$$
\frac{\mathbb E\!\left[C_{\max}(\mathcal S,q^{j^\star})\right]}{\OPT^{\mathrm{rob}}(\mathcal U)} \ge \frac{L/2}{(n+L-1)/k} = \frac{kL}{2(n+L-1)} > \beta .
$$
Thus, no randomized algorithm satisfying expected $\alpha$-consistency can be $\beta$-robust in expectation. Since deterministic algorithms are the special case of randomized algorithms with a degenerate distribution, the theorem follows.

\end{proof}
\section{Additional Budgeted Uncertainty Models} \label{app:budgeted-additional-models}

In this appendix, we give the analogous results for two weighted extensions of the cardinality-budget model. The weighted fractional and weighted binary uncertainty sets are defined in their respective subsections below. The structure is as follows: we first reduce the robust load of a set of jobs to the load of an effective interval instance, then apply the threshold splitting \Cref{alg:threshold-splitting-rr} to that interval instance, and finally give the corresponding impossibility result for related machines.

\subsection{Weighted fractional budget} \label{app:budgeted-weighted-fractional}

Under a weighted fractional budget, the adversary may distribute the budget fractionally across the jobs. Each job $j\in[n]$ has an interval $[p^-_j,p^+_j]$ and weight $w_j>0$. For a budget $\Gamma\geq0$, the corresponding uncertainty set is $\mathcal U = \{q: \exists z\in[0,1]^n \text{ s.t. } \sum_{j\in[n]}w_jz_j \leq \Gamma \text{ and } q_{i,j}=p_j^-+z_j(p_j^+-p_j^-) \}$. Thus, increasing the size of job $j$ by a fraction $z_j$ of its deviation consumes $w_jz_j$ units of budget. For a schedule $\mathcal S$, consistency is measured with respect to the predicted instance $\hat q$, while robustness is measured with respect to the min-max benchmark $\OPT^{\mathrm{rob}}(\mathcal U)$.

\paragraph{Algorithm description.}
The algorithm below is the weighted fractional analogue of \Cref{alg:budgeted-cardinality-restricted}. It chooses a cutoff $\tau$, defines the effective worst-case size $ p_j^{+,\tau}= p_j^-+\bigl(p_j^+-p_j^--\tau w_j\bigr)_+$ for every job $j\in[n]$, and lets $q^{+,\tau}$ be the corresponding restricted-assignment instance. It then applies \Cref{alg:threshold-splitting-rr} to the interval instance $(\hat q,q^{+,\tau})$.

Let $q(p, \mathcal M)$ denote the processing times corresponding to job sizes $p$ and eligibility sets $\mathcal M = \{M_j\}_{j\in [n]}$, i.e., $q(p, \mathcal M)_{i,j} = p_j$ if $i \in M_j$ and $q(p, \mathcal M)_{i,j} = \infty$ otherwise. Similarly, we let $\mathcal U(p^-, p^+, \mathcal M,w, \Gamma)$ denote the weighted fractional budgeted uncertainty set for restricted-assignment machines corresponding to $p^-, p^+, \mathcal M,w,$ and $\Gamma$.

\begin{algorithm}[H]
\caption{\btiAlg algorithm for weighted fractional budget uncertainty and restricted assignment}
\label{alg:budgeted-weighted-fractional-restricted}
\begin{algorithmic}[1]
\setlength{\itemsep}{0.25em}
\Input Lower, predicted, and worst-case sizes $p^-, \hat p,$ and $p^+$, eligibility sets $\mathcal M=\{M_j\}_j$, job weights $(w_j)_{j\in[n]}$, weighted fractional budget $\Gamma>0$, and parameter $\lambda>0$
\State $\tau\leftarrow\OPT^{\mathrm{rob}}(\mathcal U(p^-,p^+,\mathcal M,w,\Gamma))/\Gamma$ \Comment{Truncation threshold using the optimal robust makespan}
\State $p_j^{+,\tau}\leftarrow p_j^-+\bigl(p_j^+-p_j^--\tau w_j\bigr)_+$ for every job $j\in[n]$ \Comment{Effective worst-case job sizes}
\State \Return $\tsAlg(q(\hat p,\mathcal M),q(p^{+,\tau},\mathcal M),\lambda)$ \Comment{Apply the interval-uncertainty algorithm}
\end{algorithmic}
\end{algorithm}
\paragraph{Algorithm analysis.} We show that \btiAlg inherits the same consistency guarantee as \tsAlg, while losing only an additional additive constant in the robustness bound when passing back to the weighted fractional budget uncertainty set. Again, the parameter $\lambda$ controls how much the algorithm trusts the predicted instance. The next result states the resulting tradeoff.

\begin{restatable}{theorem}{WeightedFractionalTradeoff}
\label{thm:budgeted-weighted-fractional-restricted}
For any $\lambda > 0$, there exists an algorithm that has consistency $1+1/\lambda$ and robustness $2+\lambda$ for weighted fractional budget uncertainty and restricted assignment.
\end{restatable}

The analysis of \Cref{alg:budgeted-weighted-fractional-restricted} relies on a duality argument. For a fixed set of jobs, the worst weighted fractional load can be upper bounded by the load in an effective interval instance, up to an additive term $\Gamma\theta$, where the effective worst-case size of job $j$ is obtained by truncating its deviation at the weighted cutoff $\theta w_j$. The next lemma formalizes this observation.

Let $p(q)$ denote the job sizes corresponding to processing times $q$. Given a weighted fractional budgeted uncertainty set $\mathcal U$, we let $\mathcal{P}(\mathcal U) = \{p(q) : q \in \mathcal U\}$ be the uncertain job sizes  corresponding to the uncertain processing times $\mathcal U$.

\begin{lemma}
\label{lem:budgeted-weighted-fractional-reduction}
Let $p_j^{+,\theta}:=p_j^-+\bigl(p_j^+-p_j^--\theta w_j\bigr)_+$, then, for any weighted fractional budgeted uncertainty set $\mathcal U$ with budget $\Gamma$, jobs $J\subseteq[n]$, and $\theta\geq0$, we have $\max_{p\in\mathcal P(\mathcal U)}\sum_{j\in J}p_j\leq\Gamma\theta+\sum_{j\in J}p_j^{+,\theta}$.
\end{lemma}
\begin{proof}
Fix a set of jobs $J\subseteq[n]$ and a machine $i\in[m]$. The maximum additional load above the lower load is the maximum of $\{\sum_{j\in J}(p_j^+-p_j^-)z_j: \sum_{j\in J}w_jz_j \leq \Gamma \text{ and } 0 \leq z_j \leq1 \text{ for every } j\in J \}$. Its dual is the minimum of $\{ \Gamma \theta' + \sum_{j\in J} \mu_j : \theta'w_j + \mu_j \geq p_j^+-p_j^- \text{ for every } j\in J, \ \theta' \geq0 , \text{ and } \mu_j \geq 0 \text{ for every } j\in J \}$. For a fixed value of $\theta'\geq0$, the minimum feasible choice is $\mu_j=(p_j^+-p_j^--\theta'w_j)_+$ for every $j\in J$. Strong duality applies since the primal linear program is feasible and bounded. Thus, $\max_{q\in\mathcal U}\sum_{j\in J}q_{i,j}=\sum_{j\in J}p_j^-+\min_{\theta'\geq0}\left\{\Gamma\theta'+\sum_{j\in J}(p_j^+-p_j^--\theta'w_j)_+\right\}$. Fixing $\theta'=\theta$ gives $\max_{q\in\mathcal U}\sum_{j\in J}q_{i,j}\leq\sum_{j\in J}p_j^-+\Gamma\theta+\sum_{j\in J}(p_j^+-p_j^--\theta w_j)_+=\Gamma\theta+\sum_{j\in J}p_j^{+,\theta}$, where the equality follows from the definition of $p_j^{+,\theta}$.
\end{proof}

We now apply \Cref{lem:budgeted-weighted-fractional-reduction} to prove the weighted fractional guarantee. We choose the cutoff $\tau$ from the min-max optimum $\OPT^{\mathrm{rob}}(\mathcal U)$, define the effective worst-case job sizes $p_j^{+,\tau}$ and the corresponding instance $q^{+,\tau}$, and run \Cref{alg:threshold-splitting-rr} on the interval instance $(\hat q,q^{+,\tau})$. \Cref{lem:budgeted-weighted-fractional-reduction} bounds the loss when transferring the robustness guarantee from $q^{+,\tau}$ to the original weighted fractional budget uncertainty set.

\begin{proof}[Proof of \Cref{thm:budgeted-weighted-fractional-restricted}]
Let $\tau:=\OPT^{\mathrm{rob}}(\mathcal U)/\Gamma$. For every job $j\in[n]$, define $p_j^{+,\tau}:= p_j^-+\bigl(p_j^+-p_j^--\tau w_j\bigr)_+$, and let $q^{+,\tau}$ be the restricted-assignment instance with job sizes $p^{+,\tau}$ and the same eligibility sets as $\hat q$. Let $\mathcal T=(T_1,\ldots,T_m)$ be a schedule attaining $\OPT^{\mathrm{rob}}(\mathcal U)$. Thus, for every machine $i\in[m]$, $\max_{q\in\mathcal U}\sum_{j\in T_i}q_{i,j}\leq\OPT^{\mathrm{rob}}(\mathcal U)$.

We first show that $\OPT(q^{+,\tau})\leq\OPT^{\mathrm{rob}}(\mathcal U)$. Fix a machine $i\in[m]$, and let $R_i:=\{j\in T_i:p_j^+-p_j^->\tau w_j\}$. We have $\sum_{j\in R_i}w_j\leq\Gamma$. Indeed, suppose that $\sum_{j\in R_i}w_j>\Gamma$. Since the adversary may choose fractional deviations, there exists a vector $z\in[0,1]^n$ supported on $R_i$ such that $\sum_{j\in R_i}w_jz_j=\Gamma$. Since $p_j^+-p_j^->\tau w_j$ for every $j\in R_i$, this vector gives $\sum_{j\in R_i}(p_j^+-p_j^-)z_j>\tau\sum_{j\in R_i}w_jz_j=\Gamma\tau=\OPT^{\mathrm{rob}}(\mathcal U)$. The resulting load of machine $i$ is strictly larger than $\OPT^{\mathrm{rob}}(\mathcal U)$, which contradicts the definition of $\mathcal T$.

Since $\sum_{j\in R_i}w_j\leq\Gamma$, the vector that sets $z_j=1$ for every $j\in R_i$ and $z_j=0$ otherwise is feasible for the weighted fractional budget. Therefore,
$$
\sum_{j\in T_i}p_j^{+,\tau}=\sum_{j\in T_i}p_j^-+\sum_{j\in R_i}\bigl(p_j^+-p_j^--\tau w_j\bigr) \leq \sum_{j\in T_i}p_j^- + \sum_{j\in R_i}\bigl(p_j^+-p_j^-\bigr) \leq \OPT^{\mathrm{rob}}(\mathcal U),
$$
where the first equality is by the definition of $R_i$ and $p_j^{+,\tau}$, the first inequality follows since $\tau w_j\geq0$, and the second since the deviation vector supported on $R_i$ is feasible for the weighted fractional uncertainty set. Since this holds for every machine of schedule $\mathcal T$, we have $\OPT(q^{+,\tau})\leq\OPT^{\mathrm{rob}}(\mathcal U)$.

\Cref{alg:budgeted-weighted-fractional-restricted} runs \Cref{alg:threshold-splitting-rr} on instances $\hat q$ and $q^{+,\tau}$. By \Cref{thm:threshold-splitting-rr}, the returned schedule $\mathcal S=(S_1,\ldots,S_m)$ satisfies $C_{\text{max}}(\mathcal S,\hat q)\leq\left(1+\frac{1}{\lambda}\right)\OPT(\hat q)$ and
$$
C_{\text{max}}(\mathcal S,q^{+,\tau})
\leq
(1+\lambda)\OPT(q^{+,\tau})
\leq
(1+\lambda)\OPT^{\mathrm{rob}}(\mathcal U),
$$
where the second inequality follows from $\OPT(q^{+,\tau})\leq\OPT^{\mathrm{rob}}(\mathcal U)$.

It remains to bound the makespan of $\mathcal S$ over the original uncertainty set. By \Cref{lem:budgeted-weighted-fractional-reduction}, for every machine $i\in[m]$, $\max_{q\in\mathcal U}\sum_{j\in S_i}q_{i,j} \leq \Gamma\tau+\sum_{j\in S_i}p_j^{+,\tau} = \OPT^{\mathrm{rob}}(\mathcal U)+\sum_{j\in S_i}p_j^{+,\tau}$.
Thus,
$$ C_{\max}^{\mathrm{rob}}(\mathcal S,\mathcal U) = \max_{q\in\mathcal U}C_{\text{max}}(\mathcal S,q) \leq \OPT^{\mathrm{rob}}(\mathcal U) + C_{\text{max}}(\mathcal S,q^{+,\tau}) \leq (\lambda+2)\OPT^{\mathrm{rob}}(\mathcal U),
$$
where the first inequality follows by taking the maximum over machines in the preceding inequality and the second by the robustness guarantee of \Cref{alg:threshold-splitting-rr}.
\end{proof}

\Cref{alg:budgeted-weighted-fractional-restricted} is not directly polynomial-time implementable since it requires computing $\OPT^{\mathrm{rob}}(\mathcal U)$ to define the cutoff, which is NP-hard. By estimating $\OPT^{\mathrm{rob}}(\mathcal U)$ with an approximation algorithm for the weighted fractional min-max problem and applying \Cref{thm:polynomial-implementation-threshold-splitting} to the resulting interval instance, we obtain the following consistency and robustness.

\begin{restatable}{corollary}{CorPolytimeWeightedFractional}
\label{thm:polytime-weighted-fractional-budget}
Assume that there is a polynomial-time $\gamma$-approximation algorithm for makespan minimization on restricted assignment and a polynomial-time $\gamma_{\mathrm{rob}}$-approximation algorithm for robust makespan minimization under weighted fractional budget uncertainty. Then, for every $\lambda>0$, there is a polynomial-time algorithm with consistency $\gamma(1+\gamma/\lambda)$ and robustness $\gamma_{\mathrm{rob}}+\gamma(1+\gamma\lambda)$.
\end{restatable}

\begin{proof}
Let $A^{\mathrm{rob}}(\mathcal U)$ be the output of the $\gamma_{\mathrm{rob}}$-approximation algorithm for robust makespan minimization over uncertainty set $\mathcal U$. We have $\OPT^{\mathrm{rob}}(\mathcal U)\leq C_{\max}^{\mathrm{rob}}(A^{\mathrm{rob}}(\mathcal U),\mathcal U) \leq \gamma_{\mathrm{rob}}\OPT^{\mathrm{rob}}(\mathcal U)$. The polynomial-time algorithm defines $\tau:=C_{\max}^{\mathrm{rob}}(A^{\mathrm{rob}}(\mathcal U),\mathcal U)/\Gamma$ and $p_j^{+,\tau}:=p_j^-+\bigl(p_j^+-p_j^--\tau w_j\bigr)_+$ for every job $j\in[n]$. Let $q^{+,\tau}$ be the restricted-assignment instance with job sizes $p^{+,\tau}$ and the same eligibility sets as $\hat q$. Let $A(q)$ be the output of the $\gamma$-approximation algorithm for makespan minimization over instance $q$. The algorithm applies the polynomial-time implementation of \Cref{alg:threshold-splitting-rr} to instances $\hat q$ and $q^{+,\tau}$ using the makespans $C_{\text{max}}(A(\hat q),\hat q)$ and $C_{\text{max}}(A(q^{+,\tau}),q^{+,\tau})$.

\textbf{Consistency.} By \Cref{thm:polynomial-implementation-threshold-splitting}, the returned schedule $\mathcal S$ satisfies $C_{\text{max}}(\mathcal S,\hat q) \leq\gamma\left(1+\frac{\gamma}{\lambda}\right)\OPT(\hat q)$.

\textbf{Robustness.} We first show that $\OPT(q^{+,\tau})\leq\OPT^{\mathrm{rob}}(\mathcal U)$. Let $\mathcal T=(T_1,\ldots,T_m)$ be a schedule attaining $\OPT^{\mathrm{rob}}(\mathcal U)$. For every machine $i\in[m]$, let $R_i:=\{j\in T_i:p_j^+-p_j^->\tau w_j\}$. We have $\sum_{j\in R_i}w_j\leq\Gamma$. Indeed, suppose that $\sum_{j\in R_i}w_j>\Gamma$. Since the adversary may choose fractional deviations, there is a vector $z\in[0,1]^n$ supported on $R_i$ such that $\sum_{j\in R_i}w_jz_j=\Gamma$. Since $p_j^+-p_j^->\tau w_j$ for every $j\in R_i$, this vector gives additional load strictly larger than $\tau\sum_{j\in R_i}w_jz_j=\Gamma\tau= C_{\max}^{\mathrm{rob}}(A^{\mathrm{rob}}(\mathcal U),\mathcal U) \geq \OPT^{\mathrm{rob}}(\mathcal U)$ on machine $i$, which contradicts the definition of $\mathcal T$.

Since $\sum_{j\in R_i}w_j\leq\Gamma$, the vector that sets $z_j=1$ for every $j\in R_i$ and $z_j=0$ otherwise is feasible for the weighted fractional uncertainty set. Therefore,
$$
\sum_{j\in T_i}p_j^{+,\tau}
=
\sum_{j\in T_i}p_j^-
+
\sum_{j\in R_i}\bigl(p_j^+-p_j^--\tau w_j\bigr)
\leq
\sum_{j\in T_i}p_j^-
+
\sum_{j\in R_i}\bigl(p_j^+-p_j^-\bigr)
\leq
\OPT^{\mathrm{rob}}(\mathcal U),
$$
where the equality is by the definitions of $R_i$ and $p_j^{+,\tau}$, the first inequality follows since $\tau w_j\geq0$, and the second since the deviation vector supported on $R_i$ is feasible for the weighted fractional uncertainty set. Since this holds for every machine of schedule $\mathcal T$, we have $\OPT(q^{+,\tau})\leq\OPT^{\mathrm{rob}}(\mathcal U)$.

By \Cref{thm:polynomial-implementation-threshold-splitting}, the returned schedule $\mathcal S=(S_1,\ldots,S_m)$ satisfies
$$
C_{\text{max}}(\mathcal S,q^{+,\tau})
\leq
\gamma(1+\gamma\lambda)\OPT(q^{+,\tau})
\leq
\gamma(1+\gamma\lambda)\OPT^{\mathrm{rob}}(\mathcal U),
$$
where the second inequality follows from $\OPT(q^{+,\tau})\leq\OPT^{\mathrm{rob}}(\mathcal U)$.

For every machine $i\in[m]$, \Cref{lem:budgeted-weighted-fractional-reduction} applied to the set $S_i$ and $\theta=\tau$ gives
$\max_{q\in\mathcal U}\sum_{j\in S_i}q_{i,j}\leq\Gamma\tau+\sum_{j\in S_i}p_j^{+,\tau}=C_{\max}^{\mathrm{rob}}(A^{\mathrm{rob}}(\mathcal U),\mathcal U)+\sum_{j\in S_i}p_j^{+,\tau}$. Thus,
$$
\max_{q\in\mathcal U}C_{\text{max}}(\mathcal S,q)
\leq
C_{\max}^{\mathrm{rob}}(A^{\mathrm{rob}}(\mathcal U),\mathcal U)
+
C_{\text{max}}(\mathcal S,q^{+,\tau})
\leq
\left(\gamma_{\mathrm{rob}}+\gamma(1+\gamma\lambda)\right)
\OPT^{\mathrm{rob}}(\mathcal U),
$$
where the first inequality follows by taking the maximum over machines in the preceding inequality and the second since $A^{\mathrm{rob}}(\mathcal U)$ is a $\gamma_{\mathrm{rob}}$-approximate robust schedule and by the preceding bound on $C_{\text{max}}(\mathcal S,q^{+,\tau})$.
\end{proof}

The same obstruction as in \Cref{subsec:budgeted-related-impossibility} also applies to weighted fractional budgets. Even when all weights are equal to one, and the budget is one, the adversary can fully deviate the job placed on a slow machine.

\begin{theorem}
\label{thm:budgeted-related-weighted-fractional-no-constant-tradeoff}
For every constants $\alpha\geq 1$ and $\beta\geq 1$, there exists a related-machines instance with weights $w_j=1$ for every job and weighted fractional budget $\Gamma=1$ such that no randomized algorithm is both $\alpha$-consistent and $\beta$-robust.
\end{theorem}

Consequently, related machines admit no universal constant consistency robustness guarantee for weighted fractional budgeted uncertainty relative to the min-max benchmark.

\begin{proof}
Fix constants $\alpha \ge 1$ and $\beta \ge 1$, and choose integers $k>2\beta$ and $\ell\ge 2\alpha k$. Consider a related-machines instance with one machine of speed $k$ and $\ell$ machines of speed $1$. Let $n:=k+\ell$. There are $n$ jobs, and every job $j\in[n]$ has lower job size $p^-_j=1$, predicted job size $\hat p_j=1$, worst-case job size $p^+_j=L$, and weight $w_j=1$, where $L\ge1$ is chosen such that $kL/(2(n+L-1))>\beta$. Such a choice exists since $kL/(2(n+L-1))$ converges to $k/2>\beta$ as $L$ tends to infinity.

We first compute $\OPT(\hat q)$. The total predicted job size is $n$, while the total machine speed is $k+\ell=n$, and therefore $\OPT(\hat q)\ge 1$. Assigning $k$ jobs to the machine of speed $k$ and one job to each machine of speed $1$ gives a predicted makespan of $1$. Thus, $\OPT(\hat q)=1$.

We next bound the robust optimum. Consider the schedule that assigns every job to the machine of speed $k$. For any instance $q\in\mathcal U$, there is a vector $z\in[0,1]^n$ satisfying $\sum_{j\in[n]}z_j\le1$ such that the realized size of job $j$ is $1+z_j(L-1)$. Thus, the total realized job size assigned to the machine of speed $k$ is at most $n+L-1$. Hence $\OPT^{\mathrm{rob}}(\mathcal U)\le \frac{n+L-1}{k}$.

We first prove the deterministic lower bound. Let $\mathcal S$ be a schedule satisfying $C_{\max}(\mathcal S,\hat q)\le\alpha\OPT(\hat q)=\alpha$. If $\mathcal S$ assigns every job to the machine of speed $k$, then $C_{\max}(\mathcal S,\hat q)=n/k=1+\ell/k>\alpha$, where the last inequality follows from $\ell\ge 2\alpha k$. Hence, every $\alpha$-consistent schedule assigns some job $j$ to a machine of speed $1$. Under the weighted fractional budget with $w_j=1$ for every job and $\Gamma=1$, the adversary may set $z_j=1$ and $z_{j'}=0$ for every $j'\neq j$. The load of the machine containing job $j$ is then at least $L$, and therefore $C^{\mathrm{rob}}_{\max}(\mathcal S,\mathcal U)\ge L$. Thus every deterministic $\alpha$-consistent schedule has robustness ratio at least $L/((n+L-1)/k)=kL/(n+L-1)>\beta$.

We now prove that the same construction rules out randomized algorithms. Let $\textsc{ALG}$ be a randomized algorithm and $S=\textsc{ALG}(\mathcal U,\hat q)$ be the random schedule returned by the algorithm, and let $X$ be the number of jobs assigned to machines of speed $1$. Since the load on the machine of speed $k$ is at least $(n-X)/k$, expected consistency implies
$$
\alpha
\ge
\mathbb E\!\left[C_{\max}(\mathcal S,\hat q)\right]
\ge
\frac{n-\mathbb E[X]}{k}.
$$
Thus $\mathbb E[X]\ge n-\alpha k=\ell-(\alpha-1)k$. Since $\ell\ge2\alpha k$, we have $\mathbb E[X]\ge n/2$. Equivalently, the sum over all jobs of the probability that the job is assigned to a machine of speed $1$ is at least $n/2$. Hence, by averaging, there exists a job $j^\star$ that is assigned to a machine of speed $1$ with probability at least $1/2$.

Consider the scenario in which $z_{j^\star}=1$ and $z_j=0$ for every $j\neq j^\star$. This scenario is feasible for the weighted fractional budget because $w_{j^\star}=1$ and $\Gamma=1$. Under this scenario, whenever $j^\star$ is assigned to a machine of speed $1$, the makespan is at least $L$. Therefore, $\mathbb E\!\left[C_{\max}(\mathcal S,q^{j^\star})\right]\ge \frac L2$.

Using $\OPT^{\mathrm{rob}}(\mathcal U)\le (n+L-1)/k$, the expected robustness ratio is at least
$$
\frac{\mathbb E\!\left[C_{\max}(\mathcal S,q^{j^\star})\right]}{\OPT^{\mathrm{rob}}(\mathcal U)}\geq \frac{L/2}{(n+L-1)/k} = \frac{kL}{2(n+L-1)} > \beta .
$$
Thus, no randomized algorithm satisfying expected $\alpha$-consistency can be $\beta$-robust in expectation. Since $\alpha$ and $\beta$ are arbitrary constants, related machines admit no instance-independent constant consistency-robustness tradeoff for weighted fractional budgeted uncertainty.
\end{proof}

\subsection{Weighted binary budget} \label{app:budgeted-weighted-binary}

Under a weighted binary budget, the adversary may fully deviate a subset of jobs whose total weight is at most $\Gamma$. Each job $j\in[n]$ has an interval $[p_j^-,p_j^+]$ and weight $w_j>0$. For a budget $\Gamma\geq0$, the corresponding uncertainty set is $\mathcal U=\{ q:\exists T\subseteq [n]\text{ s.t. }\sum_{j\in T}w_j \leq \Gamma, q_{i,j}=p^-_j \text{ if } j\notin T, \text{ and } q_{i,j}=p_j\in [p^-_j,p^+_j] \text{ if } j\in T \}$.
Thus, increasing job $j$ from $p^-_j$ to $p^+_j$ consumes $w_j$ units of budget. The cardinality-budget model is the special case in which $w_j=1$ for every job. We assume that every job $j$ with positive deviation $p_j^+-p_j^-$ satisfies $w_j\leq\Gamma$. This is without loss of generality: if $w_j>\Gamma$, then job $j$ cannot be selected by any feasible scenario, and therefore its processing time can never be increased from $p_j^-$ to $p_j^+$; replacing $p_j^+$ by $p_j^-$ leaves the uncertainty set unchanged. For a schedule $\mathcal S$, consistency is measured with respect to the predicted instance $\hat q$, while robustness is measured with respect to the min-max benchmark $\OPT^{\mathrm{rob}}(\mathcal U)$.

\paragraph{Algorithm description.}
The algorithm below reduces the weighted binary model to its weighted fractional relaxation. It applies the weighted fractional construction to the relaxed uncertainty set and returns the resulting schedule. Under the assumption that every job with positive deviation satisfies $w_j\leq\Gamma$, the weighted fractional robust benchmark is at most twice the weighted binary robust benchmark. Combining this comparison with the weighted fractional guarantee gives the stated consistency and robustness.

Let $q(p, \mathcal M)$ denote the processing times corresponding to job sizes $p$ and eligibility sets $\mathcal M = \{M_j\}_{j\in [n]}$, i.e., $q(p, \mathcal M)_{i,j} = p_j$ if $i \in M_j$ and $q(p, \mathcal M)_{i,j} = \infty$ otherwise. Similarly, we let $\mathcal U(p^-, p^+, \mathcal M,w, \Gamma)$ denote the weighted fractional budgeted uncertainty set for restricted-assignment machines corresponding to $p^-, p^+, \mathcal M,w,$ and $\Gamma$.

\begin{algorithm}[H]
\caption{\btiAlg algorithm for weighted binary budget uncertainty and restricted assignment}
\label{alg:budgeted-weighted-binary-restricted}
\begin{algorithmic}[1]
\setlength{\itemsep}{0.25em}
\Input Lower, predicted, and worst-case sizes $p^-, \hat p,$ and $p^+$, eligibility sets $\mathcal M=\{M_j\}_j$, job weights $(w_j)_{j\in[n]}$, weighted binary budget $\Gamma>0$, and parameter $\lambda>0$
\State Let $\mathcal U'$ be the weighted fractional relaxation of $\mathcal U(p^-,p^+,\mathcal M,w,\Gamma)$
\State \Return the schedule returned by \Cref{alg:budgeted-weighted-fractional-restricted} on $\mathcal U'$ with predicted sizes $\hat p$ and parameter $\lambda$ \Comment{Apply the weighted fractional construction}
\end{algorithmic}
\end{algorithm}
\paragraph{Algorithm analysis.} We show that \btiAlg inherits the consistency guarantee of the weighted fractional construction, while losing a factor of two in the robustness bound when passing from the weighted fractional relaxation back to the weighted binary budget uncertainty set. Again, the parameter $\lambda$ controls how much the algorithm trusts the predicted instance. The next result states the resulting tradeoff.

\begin{restatable}{theorem}{WeightedBinaryTradeoff}
\label{thm:budgeted-weighted-binary-restricted}
Assume that every job with positive deviation satisfies $w_j\leq\Gamma$. For any $\lambda>0$, there exists an algorithm that has consistency $1+1/\lambda$ and robustness $4+2\lambda$ for weighted binary budget uncertainty and restricted assignment.
\end{restatable}

The analysis of \Cref{alg:budgeted-weighted-binary-restricted} proceeds through the weighted fractional relaxation. For a fixed set of jobs, the worst weighted binary load is upper bounded by the corresponding weighted fractional load, which can in turn be bounded by the load in an effective interval instance up to an additive term $\Gamma\theta$. The next lemma formalizes this bound.

Let $p(q)$ denote the job sizes corresponding to processing times $q$. Given a weighted binary budgeted uncertainty set $\mathcal U$, we let $\mathcal{P}(\mathcal U) = \{p(q) : q \in \mathcal U\}$ be the uncertain job sizes  corresponding to the uncertain processing times $\mathcal U$.

\begin{lemma}
\label{lem:budgeted-weighted-binary-reduction}
Let $p_j^{+,\theta}:=p_j^-+\bigl(p_j^+-p_j^--\theta w_j\bigr)_+$, then, for any weighted binary budgeted uncertainty set $\mathcal U$ with budget $\Gamma$, jobs $J\subseteq[n]$, and $\theta\geq0$, we have $\max_{p\in\mathcal P(\mathcal U)}\sum_{j\in J}p_j\leq\Gamma\theta+\sum_{j\in J}p_j^{+,\theta}$.
\end{lemma}
\begin{proof}
Fix a set of jobs $J\subseteq[n]$. The maximum additional load above the lower load under the weighted binary budget is at most the value of its fractional relaxation, $\max\{\sum_{j\in J}(p_j^+-p_j^-)z_j:\sum_{j\in J}w_jz_j\leq\Gamma\text{ and }0\leq z_j\leq1\text{ for every }j\in J\}$. The dual of this linear program is the minimum of $\{\Gamma\theta'+\sum_{j\in J}\mu_j:\theta'w_j+\mu_j\geq p_j^+-p_j^-\text{ for every }j\in J,\ \theta'\geq0,\text{ and }\mu_j\geq0\text{ for every }j\in J\}$. For a fixed value of $\theta'\geq0$, the minimum feasible choice is $\mu_j=(p_j^+-p_j^--\theta'w_j)_+$ for every $j\in J$. Strong duality applies since the fractional relaxation is feasible and bounded. Therefore, $\max_{p\in\mathcal P(\mathcal U)}\sum_{j\in J}p_j\leq\sum_{j\in J}p_j^-+\min_{\theta'\geq0}\{\Gamma\theta'+\sum_{j\in J}(p_j^+-p_j^--\theta'w_j)_+\}$. Fixing $\theta'=\theta$ gives $\max_{p\in\mathcal P(\mathcal U)}\sum_{j\in J}p_j\leq\sum_{j\in J}p_j^-+\Gamma\theta+\sum_{j\in J}(p_j^+-p_j^--\theta w_j)_+=\Gamma\theta+\sum_{j\in J}p_j^{+,\theta}$, where the equality follows from the definition of $p_j^{+,\theta}$.
\end{proof}

The next lemma compares the robust loads under the weighted fractional and weighted binary uncertainty sets. Under the assumption that every job with positive deviation satisfies $w_j\leq\Gamma$, the weighted fractional robust load is at most twice the weighted binary robust load.

\begin{lemma}
\label{lem:budgeted-weighted-binary-gap}
Assume that every job with positive deviation satisfies $w_j\leq\Gamma$. Let $\mathcal U$ be the weighted binary uncertainty set and let $\mathcal U'$ be its weighted fractional relaxation. Then, for every set of jobs $J\subseteq[n]$ and every machine $i\in[m]$, $ \max_{q\in\mathcal U'}\sum_{j\in J}q_{i,j} \leq 2\max_{q\in\mathcal U}\sum_{j\in J}q_{i,j}$.  Consequently, $\OPT^{\mathrm{rob}}(\mathcal U')\leq2\OPT^{\mathrm{rob}}(\mathcal U)$.
\end{lemma}
\begin{proof}
Fix a set of jobs $J\subseteq[n]$. Order the jobs in $J$ by nonincreasing value of $(p_j^+-p_j^-)/w_j$. An optimal fractional-knapsack solution for the uncertain part of the load selects a set $J'\subseteq J$ of jobs fully and possibly one additional job $h\in J\setminus J'$ fractionally. If there is no fractional job, set $p_h^+-p_h^-:=0$. Therefore, $\max_{p\in\mathcal P(\mathcal U')}\sum_{j\in J}p_j\leq\sum_{j\in J}p_j^-+\sum_{j\in J'}(p_j^+-p_j^-)+(p_h^+-p_h^-)$.

The set $J'$ satisfies $\sum_{j\in J'}w_j\leq\Gamma$ and is therefore feasible for the weighted binary uncertainty set. Moreover, if $p_h^+-p_h^->0$, then $w_h\leq\Gamma$ by assumption, and hence the singleton $\{h\}$ is also feasible. Thus, $\max_{p\in\mathcal P(\mathcal U)}\sum_{j\in J}p_j\geq\sum_{j\in J}p_j^-+\sum_{j\in J'}(p_j^+-p_j^-)$ and $\max_{p\in\mathcal P(\mathcal U)}\sum_{j\in J}p_j\geq\sum_{j\in J}p_j^-+(p_h^+-p_h^-)$. Adding these two inequalities and using $\sum_{j\in J}p_j^-\geq0$ gives
$$
\begin{aligned}
2\max_{p\in\mathcal P(\mathcal U)}\sum_{j\in J}p_j
&\geq
2\sum_{j\in J}p_j^-+\sum_{j\in J'}(p_j^+-p_j^-)+(p_h^+-p_h^-)\\
&\geq
\sum_{j\in J}p_j^-+\sum_{j\in J'}(p_j^+-p_j^-)+(p_h^+-p_h^-)\\
&\geq
\max_{p\in\mathcal P(\mathcal U')}\sum_{j\in J}p_j.
\end{aligned}
$$
This proves the first statement. Let $\mathcal T=(T_1,\ldots,T_m)$ be a schedule attaining $\OPT^{\mathrm{rob}}(\mathcal U)$. Applying the preceding inequality to $T_i$ for every machine $i\in[m]$ gives $C_{\max}^{\mathrm{rob}}(\mathcal T,\mathcal U')\leq2C_{\max}^{\mathrm{rob}}(\mathcal T,\mathcal U)=2\OPT^{\mathrm{rob}}(\mathcal U)$. Since $\OPT^{\mathrm{rob}}(\mathcal U')$ is at most the robust makespan of $\mathcal T$ over $\mathcal U'$, we obtain $\OPT^{\mathrm{rob}}(\mathcal U')\leq2\OPT^{\mathrm{rob}}(\mathcal U)$.
\end{proof}

We now apply \Cref{lem:budgeted-weighted-binary-gap} to prove the weighted binary guarantee. We run the weighted fractional construction on the fractional relaxation $\mathcal U'$ of the weighted binary uncertainty set $\mathcal U$. The consistency guarantee is unchanged, while \Cref{lem:budgeted-weighted-binary-gap} bounds the loss when transferring the robustness guarantee from $\mathcal U'$ back to the original weighted binary uncertainty set $\mathcal U$ by a factor of two.

\begin{proof}[Proof of \Cref{thm:budgeted-weighted-binary-restricted}]
Let $\mathcal U$ be the weighted binary uncertainty set and let $\mathcal U'$ be its weighted fractional relaxation. \Cref{alg:budgeted-weighted-binary-restricted} runs \Cref{alg:budgeted-weighted-fractional-restricted} on $\mathcal U'$. By \Cref{thm:budgeted-weighted-fractional-restricted}, the returned schedule $\mathcal S$ satisfies
$C_{\text{max}}(\mathcal S,\hat q)\leq\left(1+\frac{1}{\lambda}\right)\OPT(\hat q)$
and $\max_{q\in\mathcal U'}C_{\text{max}}(\mathcal S,q)\leq (\lambda+2)\OPT^{\mathrm{rob}}(\mathcal U')$.

Since every instance in $\mathcal U$ also belongs to $\mathcal U'$, we have
$\max_{q\in\mathcal U}C_{\text{max}}(\mathcal S,q)\leq\max_{q\in\mathcal U'}C_{\text{max}}(\mathcal S,q)$. Moreover, by \Cref{lem:budgeted-weighted-binary-gap}, $\OPT^{\mathrm{rob}}(\mathcal U')\leq2\OPT^{\mathrm{rob}}(\mathcal U)$.
Therefore,
$$
\max_{q\in\mathcal U}C_{\text{max}}(\mathcal S,q)
\leq
(\lambda+2)\OPT^{\mathrm{rob}}(\mathcal U')
\leq
2(\lambda+2)\OPT^{\mathrm{rob}}(\mathcal U).
$$
Thus, $\mathcal S$ is $(2\lambda+4)$-robust. The consistency guarantee is unchanged from the weighted fractional construction, and hence $\mathcal S$ is $\left(1+\frac{1}{\lambda}\right)$-consistent.
\end{proof}

\Cref{alg:budgeted-weighted-binary-restricted} is not directly polynomial-time implementable since it runs the weighted fractional construction, which requires computing the robust optimum of the fractional relaxation to define the cutoff. By estimating this value with an approximation algorithm for the weighted fractional min-max problem, applying \Cref{thm:polynomial-implementation-threshold-splitting} to the resulting interval instance, and using the factor-two comparison between the weighted fractional and weighted binary benchmarks from \Cref{lem:budgeted-weighted-binary-gap}, we obtain the following consistency and robustness.

\begin{restatable}{corollary}{CorPolytimeWeightedBinary}
\label{thm:polytime-weighted-binary-budget}
Assume that every job with positive deviation satisfies $w_j\leq\Gamma$, that there is a polynomial-time $\gamma$-approximation algorithm for makespan minimization on restricted assignment, and that there is a polynomial-time $\gamma_{\mathrm{rob}}$-approximation algorithm for robust makespan minimization under weighted fractional budget uncertainty. Then, for weighted binary budget uncertainty and restricted assignment, for every $\lambda>0$, there is a polynomial-time algorithm with consistency $\gamma(1+\gamma/\lambda)$ and robustness $2\bigl(\gamma_{\mathrm{rob}}+\gamma(1+\gamma\lambda)\bigr)$.
\end{restatable}

\begin{proof}
Let $\mathcal U$ be the weighted binary uncertainty set and let $\mathcal U'$ be its weighted fractional relaxation. The polynomial-time implementation runs the polynomial-time implementation of \Cref{alg:budgeted-weighted-fractional-restricted} on $\mathcal U'$. By \Cref{thm:polytime-weighted-fractional-budget}, the returned schedule $\mathcal S$ satisfies
$C_{\text{max}}(\mathcal S,\hat q)\leq\gamma\left(1+\frac{\gamma}{\lambda}\right)\OPT(\hat q)$ and $\max_{q\in\mathcal U'}C_{\text{max}}(\mathcal S,q)\leq \left(\gamma_{\mathrm{rob}}+\gamma(1+\gamma\lambda)\right)\OPT^{\mathrm{rob}}(\mathcal U')$. Since $\mathcal U\subseteq\mathcal U'$, we have $\max_{q\in\mathcal U}C_{\text{max}}(\mathcal S,q)\leq\max_{q\in\mathcal U'}C_{\text{max}}(\mathcal S,q)$.

Moreover, by \Cref{lem:budgeted-weighted-binary-gap}, $\OPT^{\mathrm{rob}}(\mathcal U')\leq2\OPT^{\mathrm{rob}}(\mathcal U)$. Therefore $ \max_{q\in\mathcal U}C_{\text{max}}(\mathcal S,q) \leq 2  (\gamma_{\mathrm{rob}} + \gamma(1+\gamma\lambda)) \OPT^{\mathrm{rob}}(\mathcal U)$. This proves the robustness bound, and the consistency bound follows from \Cref{thm:polytime-weighted-fractional-budget}.
\end{proof}

Finally, the impossibility result for cardinality budgets on related machines extends directly to weighted binary budgets, since the cardinality-budget model with $\Gamma=1$ is the special case in which $w_j=1$ for every job.

\begin{theorem}
\label{thm:budgeted-related-weighted-binary-no-constant-tradeoff}
For every constants $\alpha\geq 1$ and $\beta\geq 1$, there exists a related-machines instance with weights $w_j=1$ for every job and weighted binary budget $\Gamma=1$ such that no randomized algorithm is both $\alpha$-consistent and $\beta$-robust.
\end{theorem}
Consequently, related machines admit no universal constant consistency-robustness guarantee for weighted binary budgeted uncertainty relative to the min-max benchmark.

\begin{proof}
Consider the instance from \Cref{thm:budgeted-related-no-constant-tradeoff}, and let $w_j=1$ for every job $j\in[n]$. Let $\mathcal U$ be the weighted binary uncertainty set with budget $\Gamma=1$, and let $\mathcal U'$ be the corresponding cardinality-budget uncertainty set with budget $\Gamma=1$. Since all weights are equal to one, a set of jobs is feasible for $\mathcal U$ if and only if it contains at most one job. Hence $\mathcal U=\mathcal U'$.

Therefore, consistency and robustness under $\mathcal U$ are identical to consistency and robustness under $\mathcal U'$. The result follows directly from \Cref{thm:budgeted-related-no-constant-tradeoff}.
\end{proof}

\section{Proofs Missing from \Cref{sec:general-uncertainty}}
\label{app:general-uncertainty}

\subsection{Proof of \Cref{thm:general-uncertainty-tradeoff}}
\label{app:proof-general-uncertainty-tradeoff}
In this subsection, we prove \Cref{thm:general-uncertainty-tradeoff}. For every scenario $q\in\mathcal U$, let $p(q)$ be the corresponding vector of job sizes. For every set of jobs $J\subseteq[n]$, let $\Psi_{\mathcal U}(J):=\sup_{q\in\mathcal U}\sum_{j\in J}p(q)_j$ denote its maximum load over the uncertainty set. We first show the properties of $\Psi_{\mathcal U}$ used in the analysis.

\begin{lemma}
\label{lem:general-psi-properties}
Let $\mathcal U\subseteq\mathbb R_+^n$ be nonempty and bounded. The set function $\Psi_{\mathcal U}$ is monotone and subadditive. Moreover, for every schedule $\mathcal S=(S_1,\ldots,S_m)$, $\sup_{q\in\mathcal U}C_{\text{max}}(\mathcal S,q)=\max_{i\in[m]}\Psi_{\mathcal U}(S_i)$.
\end{lemma}

\begin{proof}
If $J\subseteq J'$, then $\sum_{j\in J}p(q)_j\leq\sum_{j\in J'}p(q)_j$ for every $q\in\mathcal U$. Taking the supremum over $q\in\mathcal U$ gives $\Psi_{\mathcal U}(J)\leq\Psi_{\mathcal U}(J')$.

Let $J_1,\ldots,J_t$ be pairwise disjoint sets of jobs. For every $q\in\mathcal U$,
$$
\sum_{j\in\bigcup_{s=1}^t J_s}p(q)_j
=
\sum_{s=1}^t\sum_{j\in J_s}p(q)_j
\leq
\sum_{s=1}^t\Psi_{\mathcal U}(J_s),
$$
where the inequality follows from the definition of $\Psi_{\mathcal U}(J_s)$, since $\sum_{j\in J_s}p(q)_j\leq\Psi_{\mathcal U}(J_s)$ for every $s\in[t]$ and every $q\in\mathcal U$. Taking the supremum over $q\in\mathcal U$ gives $\Psi_{\mathcal U}\bigl(\bigcup_{s=1}^tJ_s\bigr)\leq\sum_{s=1}^t\Psi_{\mathcal U}(J_s)$.

For every $q\in\mathcal U$, $C_{\text{max}}(\mathcal S,q)=\max_{i\in[m]}\sum_{j\in S_i}p(q)_j\leq\max_{i\in[m]}\Psi_{\mathcal U}(S_i)$. Taking the supremum over $q\in\mathcal U$ gives $\sup_{q\in\mathcal U}C_{\text{max}}(\mathcal S,q)\leq\max_{i\in[m]}\Psi_{\mathcal U}(S_i)$. Conversely, for every machine $i\in[m]$ and every $q\in\mathcal U$, $\sum_{j\in S_i}p(q)_j\leq C_{\text{max}}(\mathcal S,q)$. Taking the supremum over $q\in\mathcal U$ and then the maximum over machines gives the reverse inequality.
\end{proof}

The next lemma bounds the additional predicted load introduced by the first-crossing assignment.

\begin{lemma}
\label{lem:general-first-crossing-packing}
Let $a_1,\ldots,a_m\geq0$ be machine capacities, and let $b_1,\ldots,b_N\geq0$ be item sizes satisfying $b_t\leq c$ for every $t\in[N]$ and $\sum_{t=1}^Nb_t\leq\sum_{i=1}^ma_i$. The first-crossing assignment partitions the items into sets $I_1,\ldots,I_m$ such that $\sum_{t\in I_i}b_t\leq a_i+c$ for every machine $i\in[m]$.
\end{lemma}

\begin{proof}
The first-crossing assignment scans the items in their given order and assigns consecutive items to the machines. If the total size of the remaining items is at most $a_i$, then all remaining items are assigned to machine $i$ and the procedure stops. Otherwise, machine $i$ receives the shortest prefix of the remaining items whose total size exceeds $a_i$. Since the last item of this prefix has size at most $c$, the total size assigned to machine $i$ is at most $a_i+c$.

It remains to show that all items are assigned. Whenever the procedure crosses the capacity of machine $i$, it assigns a total size strictly larger than $a_i$. Thus, after crossing the capacities of machines $1,\ldots,i$, the remaining total size is strictly smaller than $\sum_{t=1}^Nb_t-\sum_{h=1}^ia_h\leq\sum_{h=i+1}^ma_h$. In particular, when the procedure reaches the last machine, the remaining total size is at most $a_m$, and all remaining items are assigned.
\end{proof}

We next give the predicted and robust bounds satisfied by the blocks constructed by \Cref{alg:general-uncertainty-identical}.

\begin{lemma}
\label{lem:general-block-bounds}
Every non-tiny block $B$ constructed by \Cref{alg:general-uncertainty-identical} satisfies
$$
\frac{\OPT(\hat q)}{2\lambda}<\sum_{j\in B}\hat p_j\leq\frac{3\OPT(\hat q)}{2\lambda},
$$
and every tiny block $D$ satisfies
$\sum_{j\in D}\hat p_j\leq\OPT(\hat q)/(2\lambda)$. Moreover, every block $B$, whether tiny or non-tiny, satisfies $\Psi_{\mathcal U}(B)\leq\OPT^{\mathrm{rob}}(\mathcal U)$. Finally, each machine $i\in[m]$ receives blocks of total predicted load at most $y_i+2\OPT(\hat q)/\lambda$.
$$
$$
\end{lemma}

\begin{proof}
Every job in $\hat L=[n]\setminus \hat H$ has predicted size at most $\OPT(\hat q)/\lambda$. A non-tiny block is the shortest prefix whose predicted load exceeds $\OPT(\hat q)/(2\lambda)$. Therefore, its predicted load is greater than $\OPT(\hat q)/(2\lambda)$ and at most
$\OPT(\hat q)/(2\lambda)+\OPT(\hat q)/\lambda=3\OPT(\hat q)/(2\lambda)$.
The stopping condition gives $\sum_{j\in D}\hat p_j\leq\OPT(\hat q)/(2\lambda)$ for every tiny block $D$.

Every block $B$ is contained in $T_h\cap\hat L$ for some machine $h$ of the optimal min-max schedule $\mathcal T=(T_1,\ldots,T_m)$. By monotonicity of $\Psi_{\mathcal U}$, $\Psi_{\mathcal U}(B)\leq\Psi_{\mathcal U}(T_h)\leq\max_{i\in[m]}\Psi_{\mathcal U}(T_i)=\OPT^{\mathrm{rob}}(\mathcal U)$, where the equality follows from \Cref{lem:general-psi-properties}.

Each machine of $\mathcal T$ contributes at most one tiny block, and hence the tiny blocks can be assigned injectively to distinct machines. Their predicted load on each machine is at most $\OPT(\hat q)/(2\lambda)$. The total predicted load of all non-tiny blocks is at most $\sum_{i=1}^my_i$. Applying \Cref{lem:general-first-crossing-packing} with capacities $y_1,\ldots,y_m$ and item-size bound $3\OPT(\hat q)/(2\lambda)$ shows that machine $i$ receives non-tiny blocks of total predicted load at most $y_i+3\OPT(\hat q)/(2\lambda)$. Adding the possible tiny block gives a total predicted load at most
$y_i+2\OPT(\hat q)/\lambda$.
\end{proof}
\GeneralTradeoff*
\begin{proof}[Proof of \Cref{thm:general-uncertainty-tradeoff}]
We first show the consistency guarantee. Fix a machine $i\in[m]$. The jobs assigned to machine $i$ are $S_i=H_i\cup L_i$, where $H_i=\hat S_i\cap\hat H$ and $L_i$ is the set of jobs contained in the blocks assigned to machine $i$. We have $\sum_{j\in H_i}\hat p_j\leq\OPT(\hat q)-y_i$, where the inequality follows since $(\hat S_1,\ldots,\hat S_m)$ is an optimal schedule for the predicted instance and $y_i=\sum_{j\in\hat S_i\setminus\hat H}\hat p_j$. By \Cref{lem:general-block-bounds}, we also have $\sum_{j\in L_i}\hat p_j\leq y_i+2\OPT(\hat q)/\lambda$. Therefore,
$$\sum_{j\in S_i}\hat p_j=\sum_{j\in H_i}\hat p_j+\sum_{j\in L_i}\hat p_j\leq\OPT(\hat q)-y_i+y_i+\frac{2\OPT(\hat q)}{\lambda}=\left(1+\frac{2}{\lambda}\right)\OPT(\hat q),
$$
where the inequality follows from the preceding bounds on the predicted loads of $H_i$ and $L_i$. Thus, each machine has load at most $\left(1+\frac{2}{\lambda}\right)\OPT(\hat q)$ over the predicted instance, and hence $C_{\text{max}}(\mathcal S,\hat q)\leq\left(1+\frac{2}{\lambda}\right)\OPT(\hat q)$.

We now show the robustness guarantee. By subadditivity of $\Psi_{\mathcal U}$, we have $\Psi_{\mathcal U}(S_i)\leq\Psi_{\mathcal U}(H_i)+\Psi_{\mathcal U}(L_i)$. Let $h_i:=|H_i|$. Since every job in $H_i$ satisfies $\hat p_j>\OPT(\hat q)/\lambda$ and $\sum_{j\in H_i}\hat p_j\leq\OPT(\hat q)-y_i$, we have $h_i\leq\lambda\left(1-\frac{y_i}{\OPT(\hat q)}\right)$. Moreover, every job $j\in[n]$ belongs to some machine $T_h$ of the optimal min-max schedule. Hence $\Psi_{\mathcal U}(\{j\})\leq\Psi_{\mathcal U}(T_h)\leq\OPT^{\mathrm{rob}}(\mathcal U)$, where the first inequality follows from the monotonicity of $\Psi_{\mathcal U}$ and the second since $(T_1,\ldots,T_m)$ is an optimal min-max schedule. By subadditivity, $\Psi_{\mathcal U}(H_i)\leq h_i\OPT^{\mathrm{rob}}(\mathcal U)\leq\lambda\left(1-\frac{y_i}{\OPT(\hat q)}\right)\OPT^{\mathrm{rob}}(\mathcal U)$.

Let $r_i$ be the number of non-tiny blocks assigned to machine $i$. Every non-tiny block has predicted load strictly larger than $\OPT(\hat q)/(2\lambda)$, while the total predicted load of the non-tiny blocks assigned to machine $i$ is at most $y_i+3\OPT(\hat q)/(2\lambda)$. Therefore, $r_i\OPT(\hat q)/(2\lambda)<y_i+3\OPT(\hat q)/(2\lambda)$, and hence $r_i\leq2\lambda y_i/\OPT(\hat q)+3$. Machine $i$ receives at most one tiny block, so $L_i$ is the union of at most $r_i+1$ blocks. By \Cref{lem:general-block-bounds}, every block has support value at most $\OPT^{\mathrm{rob}}(\mathcal U)$. Thus, by subadditivity, $\Psi_{\mathcal U}(L_i)\leq(r_i+1)\OPT^{\mathrm{rob}}(\mathcal U)\leq\left(\frac{2\lambda y_i}{\OPT(\hat q)}+4\right)\OPT^{\mathrm{rob}}(\mathcal U)$.

Combining the bounds on $H_i$ and $L_i$ gives
$$
\begin{aligned}
\Psi_{\mathcal U}(S_i)
&\leq
\left[
\lambda\left(1-\frac{y_i}{\OPT(\hat q)}\right)
+
\frac{2\lambda y_i}{\OPT(\hat q)}
+
4
\right]\OPT^{\mathrm{rob}}(\mathcal U)\\
&=
\left(
\lambda+4+\frac{\lambda y_i}{\OPT(\hat q)}
\right)\OPT^{\mathrm{rob}}(\mathcal U)\\
&\leq
(2\lambda+4)\OPT^{\mathrm{rob}}(\mathcal U),
\end{aligned}
$$
where the last inequality follows from $0\leq y_i\leq\OPT(\hat q)$. Thus, every machine has support value at most $(2\lambda+4)\OPT^{\mathrm{rob}}(\mathcal U)$. Taking the maximum over machines and applying \Cref{lem:general-psi-properties} gives $\sup_{q\in\mathcal U}C_{\text{max}}(\mathcal S,q)\leq(2\lambda+4)\OPT^{\mathrm{rob}}(\mathcal U)$.
\end{proof}

\subsection{Polynomial-time implementation of algorithm for general uncertainty}
\label{app:proof-polytime-general-uncertainty}
\begin{restatable}{corollary}{CorPolytimeGeneral}
\label{thm:polytime-general-uncertainty}
Assume that there is a polynomial-time $\gamma$-approximation algorithm for makespan minimization on identical machines and a polynomial-time $\gamma_{\mathrm{rob}}$-approximation algorithm for robust makespan minimization over $\mathcal U$. Then, for general uncertainty sets and identical machines, for every $\lambda>0$, there is a polynomial-time algorithm with consistency $\gamma(1+2/\lambda)$ and robustness $(2\lambda+4)\gamma_{\mathrm{rob}}$.
\end{restatable}
\begin{proof}
Let $A(\hat q)=(\hat S_1,\ldots,\hat S_m)$ be the output of the $\gamma$-approximation algorithm for makespan minimization on the predicted instance. We have $\OPT(\hat q)\leq C_{\text{max}}(A(\hat q),\hat q)\leq\gamma\OPT(\hat q)$. Let $A^{\mathrm{rob}}(\mathcal U)=(T_1,\ldots,T_m)$ be the output of the $\gamma_{\mathrm{rob}}$-approximation algorithm for robust makespan minimization over $\mathcal U$. We have $\OPT^{\mathrm{rob}}(\mathcal U)\leq C_{\max}^{\mathrm{rob}}(A^{\mathrm{rob}}(\mathcal U),\mathcal U)\leq\gamma_{\mathrm{rob}}\OPT^{\mathrm{rob}}(\mathcal U)$.

The polynomial-time algorithm applies \Cref{alg:general-uncertainty-identical} using schedules $A(\hat q)$ and $A^{\mathrm{rob}}(\mathcal U)$. It replaces the threshold $\OPT(\hat q)/\lambda$ by $C_{\text{max}}(A(\hat q),\hat q)/\lambda$, defines $\hat H:=\{j\in[n]:\hat p_j>C_{\text{max}}(A(\hat q),\hat q)/\lambda\}$, and sets $y_i:=\sum_{j\in\hat S_i\setminus\hat H}\hat p_j$ for every machine $i\in[m]$. The non-tiny blocks are constructed using the cutoff $C_{\text{max}}(A(\hat q),\hat q)/(2\lambda)$.

\textbf{Consistency.} Fix a machine $i\in[m]$. The predicted large jobs assigned to machine $i$ have total predicted load at most $C_{\text{max}}(A(\hat q),\hat q)-y_i$. By the proof of \Cref{lem:general-block-bounds}, with $C_{\text{max}}(A(\hat q),\hat q)$ in place of $\OPT(\hat q)$, the blocks assigned to machine $i$ have total predicted load at most $y_i+2C_{\text{max}}(A(\hat q),\hat q)/\lambda$. Therefore, $\sum_{j\in S_i}\hat p_j\leq C_{\text{max}}(A(\hat q),\hat q)-y_i+y_i+2C_{\text{max}}(A(\hat q),\hat q)/\lambda=(1+2/\lambda)C_{\text{max}}(A(\hat q),\hat q)\leq\gamma(1+2/\lambda)\OPT(\hat q)$. Taking the maximum over machines gives $C_{\text{max}}(\mathcal S,\hat q)\leq\gamma(1+2/\lambda)\OPT(\hat q)$.

\textbf{Robustness.} Fix a machine $i\in[m]$, and let $H_i:=\hat S_i\cap\hat H$. Let $L_i$ be the set of jobs contained in the blocks assigned to machine $i$, so that $S_i=H_i\cup L_i$. By subadditivity of $\Psi_{\mathcal U}$, we have $\Psi_{\mathcal U}(S_i)\leq\Psi_{\mathcal U}(H_i)+\Psi_{\mathcal U}(L_i)$.

Since every job in $H_i$ has predicted size greater than $\frac{C_{\text{max}}(A(\hat q),\hat q)}{\lambda}$ and $\sum_{j\in H_i}\hat p_j\leq C_{\text{max}}(A(\hat q),\hat q) -y_i$, we have $|H_i|\leq\lambda\left(1-\frac{y_i}{C_{\text{max}}(A(\hat q),\hat q)}\right)$. Every job is contained in some machine $T_h$ of $A^{\mathrm{rob}}(\mathcal U)$, and therefore $\Psi_{\mathcal U}(\{j\})\leq\Psi_{\mathcal U}(T_h)\leq C_{\max}^{\mathrm{rob}}(A^{\mathrm{rob}}(\mathcal U),\mathcal U)\leq\gamma_{\mathrm{rob}}\OPT^{\mathrm{rob}}(\mathcal U)$. Thus, $\Psi_{\mathcal U}(H_i)\leq\lambda\left(1-\frac{y_i}{C_{\text{max}}(A(\hat q),\hat q)}\right)\gamma_{\mathrm{rob}}\OPT^{\mathrm{rob}}(\mathcal U)$.

Let $r_i$ be the number of non-tiny blocks assigned to machine $i$. Each non-tiny block has predicted load greater than $C_{\text{max}}(A(\hat q),\hat q)/(2\lambda)$, while their total predicted load is at most $y_i+3C_{\text{max}}(A(\hat q),\hat q)/(2\lambda)$. Hence $r_i\leq2\lambda y_i/C_{\text{max}}(A(\hat q),\hat q)+3$. Machine $i$ receives at most one tiny block, so $L_i$ is the union of at most $r_i+1$ blocks. Every block is contained in some machine of $A^{\mathrm{rob}}(\mathcal U)$ and therefore has support value at most $\gamma_{\mathrm{rob}}\OPT^{\mathrm{rob}}(\mathcal U)$. By subadditivity, $\Psi_{\mathcal U}(L_i)\leq\left(2\lambda y_i/C_{\text{max}}(A(\hat q),\hat q)+4\right)\gamma_{\mathrm{rob}}\OPT^{\mathrm{rob}}(\mathcal U)$.

Combining the two bounds gives
$$
\begin{aligned}
\Psi_{\mathcal U}(S_i)&\leq\left[\lambda\left(1-\frac{y_i}{C_{\text{max}}(A(\hat q),\hat q)}\right)+\frac{2\lambda y_i}{C_{\text{max}}(A(\hat q),\hat q)}+4\right]\gamma_{\mathrm{rob}}\OPT^{\mathrm{rob}}(\mathcal U)\\ &=\left(\lambda+4+\frac{\lambda y_i}{C_{\text{max}}(A(\hat q),\hat q)}\right)\gamma_{\mathrm{rob}}\OPT^{\mathrm{rob}}(\mathcal U)\\ &\leq(2\lambda+4)\gamma_{\mathrm{rob}}\OPT^{\mathrm{rob}}(\mathcal U),
\end{aligned}
$$ where the last inequality follows from $0\leq y_i\leq C_{\text{max}}(A(\hat q),\hat q)$. Taking the maximum over machines and applying \Cref{lem:general-psi-properties} gives $C_{\max}^{\mathrm{rob}}(\mathcal S,\mathcal U)\leq(2\lambda+4)\gamma_{\mathrm{rob}}\OPT^{\mathrm{rob}}(\mathcal U)$.
\end{proof}

\subsection{Proof of \Cref{thm:general-restricted-no-constant-tradeoff}}
\label{app:proof-general-restricted-no-constant-tradeoff}
\GeneralRestrictedNoTradeoff*
\begin{proof}
Fix constants $\alpha \ge 1$ and $\beta \ge 1$, and choose an integer $k>2\beta$. Choose an integer $\ell\ge 2\alpha k$, and let $L>0$ be such that $kL/2>\beta(\ell+L)$. Such a choice exists since since $kL/(\ell+L)$ converges to $k/2>\beta$ as $L$ tends to infinity.

Consider a restricted-assignment instance with $\ell$ machines. The first $k$ machines are called color machines. For every color $c\in[k]$ and every index $i\in[\ell]$, create one job $j_{c,i}$ with lower job size $p^-_{j_{c,i}}=1$, predicted job size $\hat p_{j_{c,i}}=1$, and eligibility set $M_{j_{c,i}}:=\{c,i\}$. If $i=c$, then $M_{j_{c,i}}=\{c\}$.

For every function $\sigma:[k]\to[\ell]$, define an instance $q^\sigma$ by setting the size of job $j_{c,i}$ to $1+L$ if $i=\sigma(c)$ and to $1$ otherwise. Let $\mathcal U:=\{q^\sigma:\sigma:[k]\to[\ell]\}$. The uncertainty set $\mathcal U$ is nonempty and bounded. For every color $c\in[k]$, let $J_c:=\{j_{c,i}:i\in[\ell]\}$. For every set of jobs $A$, we have
$$\Psi_{\mathcal U}(A) = |A|+L\cdot |\{c\in[k]:A\cap J_c\neq\emptyset\}|.
$$
Indeed, every job contributes its lower size $1$. In addition, for every color represented in $A$, the adversary may choose $\sigma(c)$ so that one job of that color in $A$ receives the additional size $L$. These choices are independent across colors.

We first compute $\OPT(\hat q)$. The total predicted job size is $k\ell$, and there are $\ell$ machines, so $\OPT(\hat q)\ge k$. Assigning every job $j_{c,i}$ to machine $i$ is feasible and assigns exactly $k$ jobs to every machine. Thus, $\OPT(\hat q)=k$.

We next bound the robust optimum. Consider the schedule that assigns every job $j_{c,i}$ to its color machine $c$. This assignment is feasible. Every color machine receives all $\ell$ jobs of one color. Under every instance in $\mathcal U$, exactly one of these jobs has size $1+L$, while the remaining $\ell-1$ jobs have size $1$. Thus, every color machine has load at most $\ell+L$, and hence $\OPT^{\mathrm{rob}}(\mathcal U)\le \ell+L$.

We first prove the deterministic lower bound. Let $\mathcal S=(S_1,\ldots,S_\ell)$ be a schedule satisfying $C_{\max}(\mathcal S,\hat q)\le \alpha\OPT(\hat q)=\alpha k$. Since every predicted job size is equal to $1$, every machine contains at most $\alpha k$ jobs.

Fix a color $c\in[k]$. At most $\alpha k$ jobs of $J_c$ can be assigned to machine $c$, since machine $c$ contains at most $\alpha k$ jobs in total. Hence, at least $\ell-\alpha k$ jobs of color $c$ are assigned to machines different from $c$. If job $j_{c,i}$ is not assigned to machine $c$, then its eligibility set $M_{j_{c,i}}=\{c,i\}$ forces it to be assigned to machine $i$. Therefore, color $c$ is represented on at least $\ell-\alpha k$ machines outside its color machine.

For every machine $i\in[\ell]$, let $d_i:=|\{c\in[k]:S_i\cap J_c\neq\emptyset\}|$ be the number of colors represented on machine $i$. Summing the preceding bound over all colors gives $\sum_{i=1}^{\ell} d_i \ge k(\ell-\alpha k)$.
Thus, some machine $i^\star$ satisfies
$$
d_{i^\star} \ge \frac{k(\ell-\alpha k)}{\ell} = k\left(1-\frac{\alpha k}{\ell}\right) \ge \frac{k}{2},
$$
where the last inequality follows from $\ell\ge 2\alpha k$. Using the expression for $\Psi_{\mathcal U}$, we obtain
$$
C^{\mathrm{rob}}_{\max}(\mathcal S,\mathcal U) = \max_{i\in[\ell]}\Psi_{\mathcal U}(S_i) \ge \Psi_{\mathcal U}(S_{i^\star}) \ge Ld_{i^\star} \ge \frac{kL}{2}.
$$
Since $kL/2>\beta(\ell+L)\ge \beta\OPT^{\mathrm{rob}}(\mathcal U)$, no deterministic schedule is simultaneously $\alpha$-consistent and $\beta$-robust on this instance.

We now prove that the same construction rules out randomized algorithms. Let $\textsc{ALG}$ be a randomized algorithm and $\mathcal S=\textsc{ALG}(\mathcal U,\hat q)$ be the random schedule returned by the algorithm, and assume that $\mathbb E\!\left[C_{\max}(\mathcal S,\hat q)\right]\le \alpha\OPT(\hat q)=\alpha k$. For every color $c\in[k]$, let $Y_c$ be the number of jobs of color $c$ assigned to the color machine $c$. Since the predicted load of machine $c$ is at least $Y_c$, we have $C_{\max}(\mathcal S,\hat q)\ge Y_c$ for every realized schedule. Hence $\mathbb E[Y_c]\le \mathbb E\!\left[C_{\max}(\mathcal S,\hat q)\right]\le \alpha k$.

For $c\in[k]$ and $i\in[\ell]$, define $a_{c,i}$ as follows. If $i\neq c$, let $a_{c,i}$ be the probability that job $j_{c,i}$ is assigned to machine $i$. If $i=c$, set $a_{c,c}:=0$. Since every job $j_{c,i}$ with $i\neq c$ is assigned either to its color machine $c$ or to machine $i$, we have $\sum_{i=1}^{\ell} a_{c,i} = \ell-\mathbb E[Y_c] \ge \ell-\alpha k$ for every color $c$. Summing over the colors gives $\sum_{i=1}^{\ell}\sum_{c=1}^k a_{c,i}
\ge k(\ell-\alpha k)$.

Therefore, there exists a machine $i^\star\in[\ell]$ such that $\sum_{c=1}^k a_{c,i^\star} \ge \frac{k(\ell-\alpha k)}{\ell} \ge \frac{k}{2}.$ Now fix the scenario $q^{\sigma^\star}\in\mathcal U$ defined by $\sigma^\star(c)=i^\star$ for every color $c$. Under this scenario, job $j_{c,i^\star}$ has size $1+L$ for every color $c$. Therefore, the expected load on machine $i^\star$ is at least $L\sum_{c=1}^k a_{c,i^\star} \ge \frac{kL}{2}$. Consequently, $\mathbb E\!\left[C_{\max}(\mathcal S,q^{\sigma^\star})\right]\ge \frac{kL}{2}$. Since $\OPT^{\mathrm{rob}}(\mathcal U)\le \ell+L$ and $kL/2>\beta(\ell+L)$, we obtain
$$
\mathbb E\!\left[C_{\max}(\mathcal S,q^{\sigma^\star})\right]> \beta\OPT^{\mathrm{rob}}(\mathcal U).
$$
Thus, no randomized algorithm satisfying expected $\alpha$-consistency can be $\beta$-robust in expectation. Since $\alpha$ and $\beta$ are arbitrary constants, restricted assignment admits no instance-independent constant consistency-robustness tradeoff for general uncertainty sets, even for randomized algorithms.

\end{proof}

\section{Proofs for \Cref{sec:oblivious-cardinality}}
\label{app:oblivious-cardinality}

This appendix gives the construction and analysis of the budget-oblivious schedule used in \Cref{thm:oblivious-cardinality-tradeoff}. For every budget $\Gamma\in[n]$, let $\mathcal U_\Gamma$ be the cardinality-budget uncertainty set defined by the lower job sizes $p^-$ and the upper job sizes $p^+$. The reduction differs from the budget-to-interval reduction used in \Cref{sec:budgeted-uncertainty}, since that reduction depends explicitly on $\Gamma$. Instead, we normalize the uncertainty set associated with each budget by its robust optimum and combine the resulting sets into one general uncertainty set.

If $\OPT^{\mathrm{rob}}(\mathcal U_\Gamma)=0$ for some $\Gamma\in[n]$, then the instance in which only job $j$ takes its upper processing time belongs to $\mathcal U_\Gamma$ for every job $j\in[n]$. Hence $p_j^+=0$ for every job $j\in[n]$, and therefore all lower and upper processing times are zero. In this case, every schedule satisfies the conclusion of \Cref{thm:oblivious-cardinality-tradeoff}. We therefore assume throughout this appendix that $\OPT^{\mathrm{rob}}(\mathcal U_\Gamma)>0$ for every $\Gamma\in[n]$.

For every set of jobs $J\subseteq[n]$, let $(p^+-p^-)_{J,(1)}\geq(p^+-p^-)_{J,(2)}\geq\cdots$ be the values $p_j^+-p_j^-$ for the jobs in $J$, sorted in nonincreasing order and padded by zeros. For every budget $\Gamma\in[n]$, define $\Phi_\Gamma(J):=\sum_{j\in J}p_j^-+\sum_{\ell=1}^{\Gamma}(p^+-p^-)_{J,(\ell)}$. Thus, $\Phi_\Gamma(J)$ is the largest load of the jobs in $J$ over the uncertainty set $\mathcal U_\Gamma$.

\begin{lemma}
\label{lem:normalized-budget-profile}
Define $\mathcal U:=\bigcup_{\Gamma=1}^n\{q/\OPT^{\mathrm{rob}}(\mathcal U_\Gamma):q\in\mathcal U_\Gamma\}$. Then $\mathcal U$ is nonempty and bounded. Moreover, for every set of jobs $J\subseteq[n]$, $\Psi_{\mathcal U}(J)=\max_{\Gamma\in[n]}\Phi_\Gamma(J)/\OPT^{\mathrm{rob}}(\mathcal U_\Gamma)$. Consequently, for every schedule $\mathcal S$, $C_{\max}^{\mathrm{rob}}(\mathcal S,\mathcal U)=\max_{\Gamma\in[n]}C_{\max}^{\mathrm{rob}}(\mathcal S,\mathcal U_\Gamma)/\OPT^{\mathrm{rob}}(\mathcal U_\Gamma)$.
\end{lemma}

\begin{proof}
Each uncertainty set $\mathcal U_\Gamma$ is nonempty and bounded, and each normalizing denominator is positive. Since the union is finite, $\mathcal U$ is nonempty and bounded. Fix a set of jobs $J\subseteq[n]$. For a fixed budget $\Gamma$, the largest load of $J$ is obtained by selecting the at most $\Gamma$ jobs in $J$ with the largest values of $p_j^+-p_j^-$. Hence $\Psi_{\mathcal U_\Gamma}(J)=\Phi_\Gamma(J)$. Taking the support function of the finite union and accounting for the scaling gives $\Psi_{\mathcal U}(J)=\max_{\Gamma\in[n]}\Phi_\Gamma(J)/\OPT^{\mathrm{rob}}(\mathcal U_\Gamma)$.

Let $\mathcal S=(S_1,\ldots,S_m)$ be a schedule. By \Cref{lem:general-psi-properties}, $C_{\max}^{\mathrm{rob}}(\mathcal S,\mathcal U)=\max_{i\in[m]}\Psi_{\mathcal U}(S_i)$. Substituting the preceding identity and interchanging the finite maxima gives
$$
C_{\max}^{\mathrm{rob}}(\mathcal S,\mathcal U)
=
\max_{\Gamma\in[n]}
\frac{\max_{i\in[m]}\Phi_\Gamma(S_i)}
{\OPT^{\mathrm{rob}}(\mathcal U_\Gamma)}.
$$
Applying \Cref{lem:general-psi-properties} to $\mathcal U_\Gamma$ gives $\max_{i\in[m]}\Phi_\Gamma(S_i)=C_{\max}^{\mathrm{rob}}(\mathcal S,\mathcal U_\Gamma)$ for every $\Gamma\in[n]$, which proves the schedule identity.
\end{proof}

Order the jobs so that $p_{j_1}^+-p_{j_1}^-\geq p_{j_2}^+-p_{j_2}^-\geq\cdots\geq p_{j_n}^+-p_{j_n}^-$, and interpret $p_{j_r}^+-p_{j_r}^-$ as zero for every $r>n$.

\begin{lemma}
\label{lem:rank-spread-cardinality-budget}
For every $\Gamma\in[n]$, $\sum_{\ell=1}^{\Gamma}(p_{j_{(\ell-1)m+1}}^+-p_{j_{(\ell-1)m+1}}^-)\leq3\OPT^{\mathrm{rob}}(\mathcal U_\Gamma)$.
\end{lemma}

\begin{proof}
Fix a budget $\Gamma\in[n]$. Let $\mathcal T=(T_1,\ldots,T_m)$ be a schedule minimizing $\max_{i\in[m]}\sum_{\ell=1}^{\Gamma}(p^+-p^-)_{T_i,(\ell)}$, and let $\mathcal T^\Gamma=(T^\Gamma_1,\ldots,T^\Gamma_m)$ attain $\OPT^{\mathrm{rob}}(\mathcal U_\Gamma)$. For every machine $i\in[m]$, we have $\sum_{\ell=1}^{\Gamma}(p^+-p^-)_{T^\Gamma_i,(\ell)}\leq\Phi_\Gamma(T^\Gamma_i)\leq\OPT^{\mathrm{rob}}(\mathcal U_\Gamma)$. Hence, the deviation-only objective of $\mathcal T$ is at most $\OPT^{\mathrm{rob}}(\mathcal U_\Gamma)$, and therefore $\sum_{\ell=1}^{\Gamma}(p^+-p^-)_{T_i,(\ell)}\leq\OPT^{\mathrm{rob}}(\mathcal U_\Gamma)$ for every machine $i\in[m]$. In particular, $p_{j_1}^+-p_{j_1}^-\leq\OPT^{\mathrm{rob}}(\mathcal U_\Gamma)$.

If $\Gamma=1$, the result follows. Suppose that $\Gamma\geq2$. For every $\ell=2,\ldots,\Gamma$, the $m$ values with ranks $(\ell-2)m+1,\ldots,(\ell-1)m$ are at least $p_{j_{(\ell-1)m+1}}^+-p_{j_{(\ell-1)m+1}}^-$. Thus, $p_{j_{(\ell-1)m+1}}^+-p_{j_{(\ell-1)m+1}}^-\leq\frac{1}{m}\sum_{s=(\ell-2)m+1}^{(\ell-1)m}(p_{j_s}^+-p_{j_s}^-)$. Summing over $\ell=2,\ldots,\Gamma$ gives $\sum_{\ell=2}^{\Gamma}(p_{j_{(\ell-1)m+1}}^+-p_{j_{(\ell-1)m+1}}^-)\leq\frac{1}{m}\sum_{s=1}^{(\Gamma-1)m}(p_{j_s}^+-p_{j_s}^-)$.

Let $h:=\min\{(\Gamma-1)m,n\}$, let $H$ be the set of the $h$ jobs with largest values of $p_j^+-p_j^-$, breaking ties arbitrarily, and let $x_i:=|H\cap T_i|$ for every machine $i\in[m]$. Then $\sum_{i=1}^m x_i=h$. Moreover, by the zero-padding convention, $\sum_{s=1}^{(\Gamma-1)m}(p_{j_s}^+-p_{j_s}^-)=\sum_{j\in H}(p_j^+-p_j^-)$.

We claim that $\sum_{j\in H\cap T_i}(p_j^+-p_j^-)\leq\OPT^{\mathrm{rob}}(\mathcal U_\Gamma)+(x_i/\Gamma)\OPT^{\mathrm{rob}}(\mathcal U_\Gamma)$ for every machine $i\in[m]$. If $x_i\leq\Gamma$, then the left-hand side is at most the sum of the $\Gamma$ largest values of $p_j^+-p_j^-$ on $T_i$, and is therefore at most $\OPT^{\mathrm{rob}}(\mathcal U_\Gamma)$. If $x_i>\Gamma$, then the $\Gamma$ largest values of $p_j^+-p_j^-$ in $H\cap T_i$ have total value at most $\OPT^{\mathrm{rob}}(\mathcal U_\Gamma)$. Every remaining value in $H\cap T_i$ is at most the average of these $\Gamma$ values and is therefore at most $\OPT^{\mathrm{rob}}(\mathcal U_\Gamma)/\Gamma$. This proves the claim.

Summing over the machines gives $$\sum_{s=1}^{(\Gamma-1)m}(p_{j_s}^+-p_{j_s}^-)=\sum_{i=1}^m\sum_{j\in H\cap T_i}(p_j^+-p_j^-)\leq m\OPT^{\mathrm{rob}}(\mathcal U_\Gamma)+(h/\Gamma)\OPT^{\mathrm{rob}}(\mathcal U_\Gamma)$$. Since $h\leq(\Gamma-1)m$, we have $(h/\Gamma)\OPT^{\mathrm{rob}}(\mathcal U_\Gamma)\leq((\Gamma-1)m/\Gamma)\OPT^{\mathrm{rob}}(\mathcal U_\Gamma)\leq m\OPT^{\mathrm{rob}}(\mathcal U_\Gamma)$. Therefore, $\sum_{s=1}^{(\Gamma-1)m}(p_{j_s}^+-p_{j_s}^-)\leq2m\OPT^{\mathrm{rob}}(\mathcal U_\Gamma)$.

It follows that $\sum_{\ell=2}^{\Gamma}(p_{j_{(\ell-1)m+1}}^+-p_{j_{(\ell-1)m+1}}^-)\leq2\OPT^{\mathrm{rob}}(\mathcal U_\Gamma)$. Combining this inequality with $p_{j_1}^+-p_{j_1}^-\leq\OPT^{\mathrm{rob}}(\mathcal U_\Gamma)$ proves the lemma.
\end{proof}

\begin{lemma}
\label{lem:universal-cardinality-budget-profile}
There exists a schedule $\mathcal T^\star=(T^\star_1,\ldots,T^\star_m)$ such that $\Phi_\Gamma(T^\star_i)\leq5\OPT^{\mathrm{rob}}(\mathcal U_\Gamma)$ for every machine $i\in[m]$ and every budget $\Gamma\in[n]$. Equivalently, $\OPT^{\mathrm{rob}}(\mathcal U)\leq5$.
\end{lemma}

\begin{proof}
Order the jobs in nonincreasing order of $p_j^+-p_j^-$ and partition this order into consecutive layers $L_1,L_2,\ldots$, each containing $m$ jobs except possibly the last. Thus, layer $L_r$ contains the jobs with ranks $(r-1)m+1,\ldots,\min\{rm,n\}$. We construct $\mathcal T^\star$ layer by layer, assigning every layer injectively to the machines.

Suppose that some layers have already been assigned, and order the current lower loads as $a_1\leq a_2\leq\cdots\leq a_m$. We maintain the invariant $\max_i a_i-\min_i a_i\leq\max_j p_j^-$. The invariant holds before the first layer is assigned. Add dummy jobs of lower size zero if the next layer contains fewer than $m$ jobs, and order the lower sizes in that layer as $b_1\geq b_2\geq\cdots\geq b_m\geq0$. Assign the job of lower size $b_t$ to the machine with current load $a_t$.

The new unsorted loads are $a_t+b_t$. For $s<t$, we have $(a_s+b_s)-(a_t+b_t)\leq b_s-b_t\leq\max_jp_j^-$ and $(a_t+b_t)-(a_s+b_s)\leq a_t-a_s\leq\max_jp_j^-$. Hence the invariant is preserved after assigning the layer.

Let $a_i$ denote the final lower loads. The invariant gives $\max_i a_i\leq\min_i a_i+\max_jp_j^-\leq\frac{1}{m}\sum_{j=1}^np_j^-+\max_jp_j^-$. For every budget $\Gamma\in[n]$, the lower instance belongs to $\mathcal U_\Gamma$, and therefore $\frac{1}{m}\sum_{j=1}^np_j^-\leq\OPT^{\mathrm{rob}}(\mathcal U_\Gamma)$. Moreover, for every job $j$, the instance in which only job $j$ takes its upper processing time belongs to $\mathcal U_\Gamma$. Thus, $p_j^+\leq\OPT^{\mathrm{rob}}(\mathcal U_\Gamma)$, and hence $\max_jp_j^-\leq\OPT^{\mathrm{rob}}(\mathcal U_\Gamma)$. It follows that $\sum_{j\in T^\star_i}p_j^-\leq2\OPT^{\mathrm{rob}}(\mathcal U_\Gamma)$ for every machine $i\in[m]$ and every budget $\Gamma\in[n]$.

It remains to bound the deviating part. Since every layer is assigned injectively, each machine receives at most one job from each layer. Therefore, the $\ell$-th largest value of $p_j^+-p_j^-$ on machine $i$ is at most $p_{j_{(\ell-1)m+1}}^+-p_{j_{(\ell-1)m+1}}^-$. Otherwise, machine $i$ would contain $\ell$ jobs with value strictly larger than $p_{j_{(\ell-1)m+1}}^+-p_{j_{(\ell-1)m+1}}^-$, all of which would belong to the first $\ell-1$ layers, contradicting the injective assignment of each layer. By \Cref{lem:rank-spread-cardinality-budget},
$$
\sum_{\ell=1}^{\Gamma}(p^+-p^-)_{T^\star_i,(\ell)}
\leq
\sum_{\ell=1}^{\Gamma}(p_{j_{(\ell-1)m+1}}^+-p_{j_{(\ell-1)m+1}}^-)
\leq
3\OPT^{\mathrm{rob}}(\mathcal U_\Gamma).
$$

Combining the lower and deviating parts gives $\Phi_\Gamma(T^\star_i)=\sum_{j\in T^\star_i}p_j^-+\sum_{\ell=1}^{\Gamma}(p^+-p^-)_{T^\star_i,(\ell)}\leq5\OPT^{\mathrm{rob}}(\mathcal U_\Gamma)$ for every machine $i\in[m]$ and every budget $\Gamma\in[n]$. By \Cref{lem:normalized-budget-profile}, $C_{\max}^{\mathrm{rob}}(\mathcal T^\star,\mathcal U)\leq5$, and therefore $\OPT^{\mathrm{rob}}(\mathcal U)\leq5$.
\end{proof}

The budget-oblivious schedule is obtained by applying \Cref{alg:general-uncertainty-identical} to the predicted instance $\hat q$ and the normalized uncertainty set $\mathcal U$.

\ThmObliviousCardinalityTradeoff*

\begin{proof}
Apply \Cref{thm:general-uncertainty-tradeoff} to the predicted instance $\hat q$ and the normalized uncertainty set $\mathcal U$. This gives a schedule $\mathcal S$ satisfying $C_{\text{max}}(\mathcal S,\hat q)\leq(1+2/\lambda)\OPT(\hat q)$ and $C_{\max}^{\mathrm{rob}}(\mathcal S,\mathcal U)\leq(2\lambda+4)\OPT^{\mathrm{rob}}(\mathcal U)$. By \Cref{lem:universal-cardinality-budget-profile}, $\OPT^{\mathrm{rob}}(\mathcal U)\leq5$, and therefore $C_{\max}^{\mathrm{rob}}(\mathcal S,\mathcal U)\leq10\lambda+20$. By \Cref{lem:normalized-budget-profile}, this is equivalent to $C_{\max}^{\mathrm{rob}}(\mathcal S,\mathcal U_\Gamma)/\OPT^{\mathrm{rob}}(\mathcal U_\Gamma)\leq10\lambda+20$ for every $\Gamma\in[n]$. The consistency guarantee follows directly from \Cref{thm:general-uncertainty-tradeoff}.
\end{proof}

The schedule is independent of $\Gamma$ and satisfies the stated robustness guarantee for every cardinality budget simultaneously. In contrast, the construction in \Cref{sec:budgeted-uncertainty} gives a sharper guarantee for a fixed budget but requires that budget as part of the input.

\section{A Simple Separation for Selection Problems}
\label{app:selection-separation}

The following observation shows that extending the framework beyond scheduling cannot be expected to work in a black-box way. Already for a problem with two feasible solutions and two scenarios, the predicted optimum and the robust optimum may be incompatible.

\begin{proposition}
\label{prop:two-solution-selection-separation}
For every constants $\alpha,\beta\ge 1$, there is a robust optimization problem with two feasible solutions and two scenarios such that no solution is both $\alpha$-consistent and $\beta$-robust.
\end{proposition}

\begin{proof}
Consider a minimization problem with feasible solutions $a$ and $b$, uncertainty set $\mathcal U=\{\hat q,q\}$, and predicted scenario $\hat q$. Let $M>\max\{\alpha,\beta\}$ and define the costs by
\[
c_{\hat q}(a)=1,\qquad c_{\hat q}(b)=M,
\qquad
c_q(a)=M^2,\qquad c_q(b)=M .
\]
The predicted optimum is $\OPT(\hat q)=1$, attained by $a$. The min--max robust optimum is $\OPT^{\mathrm{rob}}(\mathcal U)=M$, attained by $b$, since
\[
\max_{r\in\mathcal U} c_r(a)=M^2
\qquad\text{and}\qquad
\max_{r\in\mathcal U} c_r(b)=M .
\]
Thus, $a$ is not $\beta$-robust, because its robustness ratio is $M>\beta$, while $b$ is not $\alpha$-consistent, because its consistency ratio is $M>\alpha$. Hence, no feasible solution is both $\alpha$-consistent and $\beta$-robust.
\end{proof}

This two-choice example embeds into many selection-type problems whenever the feasible region can be restricted to choosing one of two alternatives. This includes simple variants of knapsack, matching, or path-selection problems by making all other feasible choices infeasible or prohibitively expensive. Thus, without additional structure, even very simple selection problems can fail to admit simultaneous constant consistency and robustness.

\end{document}